\documentclass[book,authoryear,11pt]{elsarticle}
\usepackage{amsmath,amssymb,epsfig,amsfonts,epsfig,graphicx}

\makeatletter
\def\ps@pprintTitle{
\let\@oddhead\@empty
\let\@evenhead\@empty
\let\@oddfoot\@empty
\let\@evenfoot\@oddfoot
}
\makeatother

\usepackage{amsmath,amssymb,epsfig,amsfonts,epsfig,graphicx}
\usepackage{epsfig}
\usepackage{epstopdf}
\usepackage{multicol}
\usepackage{amsthm}
\usepackage{color}
\usepackage[normalem]{ulem}
\usepackage{fancyhdr}

\newtheorem{theorem}{Theorem}[section]
\newtheorem{lemma}[theorem]{Lemma}
\newtheorem{corollary}{Corollary}[section]

\newcommand{\beq}{\begin{equation}}
\newcommand{\eeq}{\end{equation}}
\newcommand{\bea}{\begin{eqnarray}}
\newcommand{\eea}{\end{eqnarray}}
\newcommand{\beas}{\begin{eqnarray*}}
\newcommand{\eeas}{\end{eqnarray*}}

\makeatletter
\def\ps@pprintTitle{%
\let\@oddhead\@empty
\let\@evenhead\@empty
\let\@oddfoot\@empty
\let\@evenfoot\@oddfoot
}
\makeatother

\begin{document}

\pagestyle{fancy}
\lhead{}
\chead{}
\rhead{}
\lfoot{G. Carpentieri, R.E. Skelton, F. Fraternali, \textit{Minimum Mass Tensegrity Bridges}}
\cfoot{}
\rfoot{\thepage}
\renewcommand{\headrulewidth}{0.6pt}
\renewcommand{\footrulewidth}{0.6pt}

\begin{frontmatter}

\title{Parametric Design of Minimal Mass Tensegrity Bridges Under Yielding and Buckling Constraints}

\author{G.~Carpentieri}
\ead{gcarpentieri@unisa.it}
\address{Department of Civil Engineering, University of Salerno, 84084 Fisciano (SA), Italy}

\author{R.E.~Skelton}
\ead{bobskelton@ucsd.edu}
\address{UCSD, MAE, 9500 Gilman Dr., La Jolla, CA 92093, USA}

\author{F.~Fraternali}
\ead{f.fraternali@unisa.it}
\address{Department of Civil Engineering, University of Salerno, 84084 Fisciano (SA), Italy}

\begin{abstract}
This study investigates the use of the most fundamental elements; cables for tension and bars for compression, in the search for the most efficient bridges. Stable arrangements of these elements are called tensegrity structures. We show herein the minimal mass arrangement of these basic elements to satisfy both yielding and buckling constraints. We show that the minimal mass solution for a simply-supported bridge subject to buckling constraints matches Michell's 1904 study which treats the case of only yield constraints, even though our boundary conditions differ. The necessary and sufficient condition is given for the minimal mass bridge to lie totally above (or below) deck. Furthermore this condition depends only on material properties. If one ignores joint mass, and considers only bridges above deck level, the optimal complexity (number of elements in the bridge) tends toward infinity (producing a material continuum). If joint mass is considered then the optimal complexity is finite. The optimal (minimal mass) bridge below deck has the smallest possible complexity (and therefore cheaper to build), and under reasonable material choices, yields the smallest mass bridge.

\end{abstract}

\end{frontmatter}

\pagebreak

\tableofcontents

\pagebreak

\section{Introduction}
\label{intro}

Tensegrity structures are axially loaded prestressable structures. Motivated by nature, where tensegrity concepts appear in every cell, in the molecular structure of the spider fiber, and in the arrangement of bones and tendons for control of locomotion in animals and humans, Engineers have only recently developed efficient analytical methods to exploit tensegrity concepts in engineering design. Previous attempts to judge the suitability of tensegrity for engineering purposes have simply evaluated the tensegrity produced as art-forms, but then judges them according to a different (engineering) criteria. The development of  "tensegrity engineering" should allow tensegrity concepts to be optimized for the particular engineering problem at hand, rather than simply evaluating a structure designed only for artistic objectives. The development of such tensegrity engineering methods is the continuing goal of our research. 

The tensegrity paradigm used for bridges in this study allows the marriage of composite structures within the design. Our tensegrity approach creates a network of tension and compressive members distributed throughout the system at many different scales (using tensegrity fractals generates many different scales). Furthermore, these tension and compression members can simultaneously serve multiple functions, as load-carrying members of the structure, and as sensing and actuating functions. Moreover, the choice of materials for each member of the network can form a system with special electrical properties, special acoustic properties, special mechanical properties (stiffness, etc). The mathematical tools of this study can be used therefore to design metamaterials and composite materials (cf., e.g., \cite{DaraioPRL10, Ngo12}) with unusual and very special properties not available with normal design methods.

This study focuses on bridge design for minimal mass. The subject of form-finding of tensegrity structures continues to be an active research area \citep{Koohestani2012,IanSmith2010b,Sakamoto2008,Sokof2012,Tibert2011,Yamamoto2011}, due to the special ability of such structures to serve as controllable systems (geometry, size, topology and prestress control), and
also because the tensegrity architecture provides minimum mass structures for a variety of loading conditions, \cite{Skelton2010a, Skelton2010b, Skelton2010c,Nagase2014}.
Particularly interesting is the use of fractal geometry as a form-finding method for tensegrity structures, which is well described in \cite{Skelton2010a, Skelton2010b, Skelton2010c,Fraternali2011}.
Such an optimization strategy exploits the use of fractal geometry to design tensegrity structures, through a finite or infinite number of self-similar subdivisions of  basic modules. The strategy looks for the optimal number of self-similar iterations to achieve minimal mass or other design criteria. This number is called the optimal \emph{complexity}, since this number fixes the total number of parts in the structure.

The self-similar tensegrity design presented in \cite{Skelton2010a, Skelton2010b, Skelton2010c} is primarily focused on  the generation of \textit{minimum mass} structures, which are of great technical relevance when dealing with tensegrity bridge structures (refer, e.g., to \cite{IanSmith2010a}).
The `fractal' approach to tensegrity form-finding paves the way to an effective implementation of the tensegrity paradigm in \textit{parametric architectural design} \citep{Sakamoto2008,IanSmith2010b,Phocas2012,Baker2013}.

Designing tensegrity for engineering objectives has produced minimal mass solutions for five fundamental (but planar) problems in engineering mechanics. Minimal mass for tensile structures, (subject to a stiffness constraint) was motivated by the molecular structure of spider fiber, and may be found in \citep{Skelton2012}. Minimal mass for compressive loads may be found in \citep{Skelton2010a}. Minimal mass for cantilevered bending loads may be found in \citep{Skelton2010b}. Minimal mass for torsional loads may be found in \citep{Skelton2010c}. Discussions of minimal mass solutions for distributed loads on simply-supported spans, where significant structure is not allowed below the roadway, may be found in \citep{Skelton2014}. 

This study finds the minimum mass design of tensegrity structures carrying simply supported and distributed bending loads. In   \cite{Skelton2014} numerical solutions where found for a specified topology, without any theoretical guarantees that those topologies produced minimal mass. This study provides more fundamental proofs that provide necessary and sufficient conditions for minimal mass.

It is also worth noting that tensegrity structures can serve multiple functions. While a cable is a load-carrying member of the structure, it might also serve as a sensor or actuator to measure or modify tension or length.  Other advantages of tensegrity structures  are related to
the possibility to integrate control functions within the design of the structure. A grand design challenge in tensegrity engineering is to coordinate the structure and control designs to minimize the control energy and produce a structure of minimal mass. This would save resources (energy and mass) in two disciplines, and therefore "integrate" the disciplines \citep{Skelton2002}.

The remainder of the study is organized as follows. Section 2 provides some basic knowledges on the mode of failure of tensile and compressive members. Section 3 describes the topology of the tensegrity bridge under examination. For a simply-supported structure of the simplest complexity,  Section 4 describes the minimal mass bridge when the admissible topology allows substructure and superstructure (that is, respectively, structure below and above the roadbed). Section 5 provides closed-form solutions to the minimal mass bridge designs (of complexity $n = 1$) when only sub- or super-structure is allowed. Section 6 first defines deck mass and provides closed-form solutions to the minimal mass bridge designs (of complexity $n, p = q = 1$) when only sub- or super-structure is allowed. This finalizes the proof that the minimal mass bridge is indeed the substructure bridge. Section 6 also adds joint mass and shows that the optimal complexity is finite. Conclusions are offered at the end.


\section{Properties of Tensile and Compressive Components of the Tensegrity Structure}

The tensegrity structures in this study will be composed of rigid compressive members called \emph{bars}, and elastic tensile members called \emph{cables}. We will assume that a tensile member obeys Hooke's law, 

\begin{align}
t_s=k(s-s_0),
\end{align}

where $k$ is cable stiffness, $t_s$ is tension in the cable, $s$ is the length of the cable, and $s_0<s$ is the rest length of the cable. The tension members cannot support compressive loads. For our purposes, a compressive member is a solid cylinder, called a bar. All results herein are trivially modified to accommodate pipes, tubes of any material, but the concepts are more easily demonstrated and the presentation is simplified by using the solid bar in our derivations.  The minimal mass of a cable with loaded length $s$, yield strength $\sigma_s$, mass density $\rho_s$, and maximal tension $t_s$ is

\beq
\label{m_s}
m_s=\frac{\rho_s}{\sigma_s}t_s s.
\eeq

To avoid yielding, a bar of length $b$, yield strength $\sigma_b$, mass density $\rho_b$ with compression force $f_b$, has the minimal mass

\beq
\label{m_b}
m_{b,Y}=\frac{\rho_b}{\sigma_b}f_b b.
\eeq

To avoid buckling, the minimal mass of a round bar of length $b$, modulus of elasticity $E_b$, and maximal force $f_b$ is 

\beq
\label{m_b}
m_{b,B}=2\rho_b b^2\sqrt{\frac{f_b}{\pi E_b}}.
\eeq

The actual mode of failure (buckling or yielding) of a compressive member can be identified by using the following well-know facts that give the basis to a correct design of the bar radius $r_b$. Define $r_Y$, the bar radius that satisfies yielding constraints, and $r_B$, the radius that satisfies buckling constraints, by

\begin{align}
r_Y=\sqrt{\frac{f_b}{\pi \sigma_b}}, \ \ \
r_B= \sqrt[4]{\frac{4 b^2 f_b}{\pi^3 E_b}}.
\end{align}

The following are well known facts:

\begin{lemma} \label{theo:yield_design_1}

Designs subject to only yield constraints (hence $r_b=r_Y$) fail to identify the actual mode of failure (buckling) if $ r_Y  < r_B$, or equivalently if,

\beq
\label{ry<rb}
\frac{f_b}{b^2} < \frac{4  \sigma_b^2}{\pi E_b} .
\eeq 

\end{lemma}

\begin{lemma} \label{theo:yield_design_2}

Designs subject to only yield constraints ($r_b=r_Y$) automatically also satisfy buckling constraints if $r_Y  > r_B$, or equivalently if,

\beq
\label{ry>rb}
\frac{f_b}{b^2} > \frac{4  \sigma_b^2}{\pi E_b}. 
\eeq 

\end{lemma}

\begin{lemma} \label{theo:buckling_design_1}

Designs subject to only buckling constraints ($r_b=r_B$) fail to identify the actual mode of failure (yielding) if $r_B  < r_Y$,  or equivalently if,

\beq
\label{rb<ry}
\frac{f_b}{b^2} > \frac{4  \sigma_b^2}{\pi E_b}.
\eeq 

\end{lemma}

\begin{lemma} \label{theo:buckling_design_2}

Designs subject to only buckling constraints ($r_b=r_B$) automatically also satisfy yielding constraints if $r_B > r_Y$, or equivalently if,

\beq
\label{rb>ry}
\frac{f_b}{b^2} < \frac{4  \sigma_b^2}{\pi E_b}.
\eeq 

\end{lemma}

\section{Planar Topologies of the Tensegrity Bridges Under Study}

The planar bridge topology is considered here to elucidate the fundamental properties that are important in the vertical plane. We use the following nomenclature, referring to Fig. \ref{basic_module}:

\begin{itemize}

\item A \emph{superstructure} bridge has no structure below the deck level.
\item A \emph{substructure} bridge has no structure above the deck level.
\item A \emph{nominal} bridge contains both \emph{substructure} and \emph{superstructure}.
\item $Y$ means the design was constrained against yielding for both cables and bars.
\item$B$ means the design was constrained against yielding for cables and buckling for bars.
\item $n$ means the number of self-similar iterations involved in the design ($n=1$ in Fig \ref{basic_module}, and $n \ge 1$ in Fig. \ref{fig:dual_exemplaries}).
\item $p$ means the complexity of each iteration in the substructure ($p=1$ in Fig \ref{basic_module}c, and $p\ge1$ in Fig \ref{fig:dual_exemplaries}).
\item $q$ means the complexity of each iteration in the superstructure ($q=1$ in Fig \ref{basic_module}b, and $q\ge1$ in Fig \ref{fig:dual_exemplaries}).
\item $\alpha$ is the aspect angle of the \emph{superstructure} measured from the horizontal.
\item $\beta$ is the aspect angle of the \emph{substructure} measured from the horizontal.
\end{itemize}

\noindent For a tensegrity bridge with generic complexities $n$, $p$ and $q$ (see Fig. \ref{fig:dual_exemplaries}), the total number of nodes $n_{n}$ of each topology is given by:

\bea
\label{nnodes}
n_{n} =\left( p + q \right) \left( 2^{n} - 1 \right) + 2^{n} + 1.
\eea

\noindent For the \emph{substructure} bridge (that is $q=0$), the number of bars $n_b$ and the number of cables $n_s$ are:

\bea
\label{nconn_sub}
n_{b} = p \left( 2^{n} - 1 \right), \ \ \ 
n_{s} = \left( p + 1 \right) \left( 2^{n} - 1 \right) + 2^{n}.
\eea

\noindent For the \emph{superstructure} bridge (that is $p=0$), the number of bars $n_b$ and the number of cables $n_s$ are:

\bea
\label{nconn_super}
n_{b} = \left( q + 1 \right) \left( 2^{n} - 1 \right), \ \ \ 
n_{s} = q \left( 2^{n} - 1 \right) + 2^{n}.
\eea

\noindent For the \emph{nominal} bridge, the number of bars $n_b$ and the number of cables $n_s$ are:

\bea
\label{nconn_combo}
n_{b} = \left( p + q + 1 \right) \left( 2^{n} - 1 \right), \ \ \
n_{s} = \left( p + q + 1 \right) \left( 2^{n} - 1 \right) + 2^{n}.
\eea

We define the \emph{superstructure} bridge of complexity ($n$, $p=0$, $q$) by Fig. \ref{fig:dual_exemplaries} where the \emph{substructure} below is deleted. We define the \emph{substructure} bridge of complexity ($n$, $p$, $q=0$) by Fig. \ref{fig:dual_exemplaries} where the \emph{superstructure} above is deleted. 

\begin{figure}[hbt] 
\unitlength1cm
\begin{picture}(12,4.25)
\put(0,-0.5){\psfig{figure=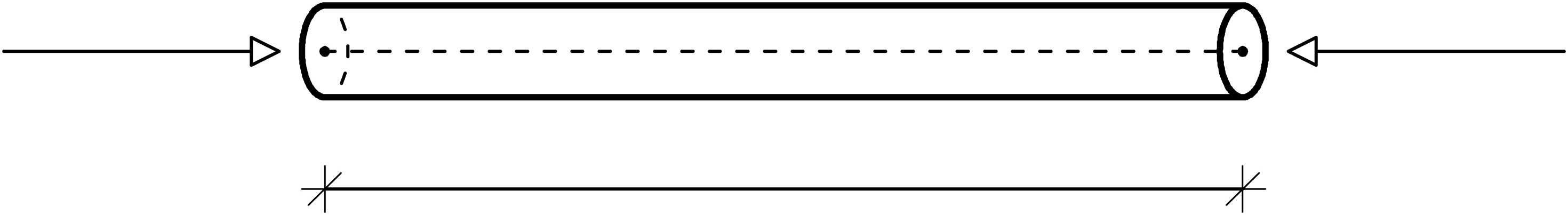,width=12cm}}
\put(0,2.5){\psfig{figure=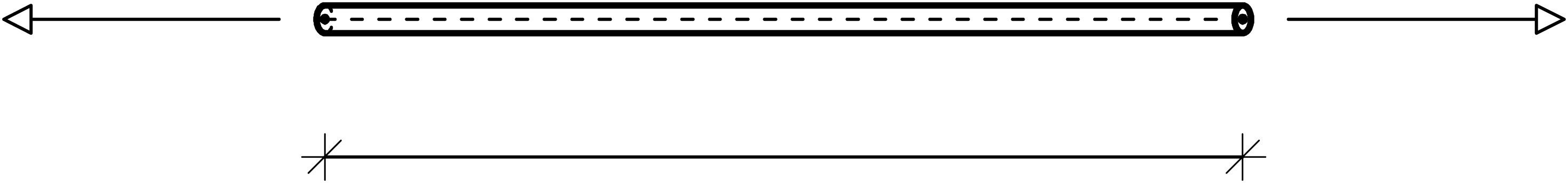,width=12cm}}
\put(5.5,4.25){{$\mbox{cable}$}}
\put(10.75,4){{$\mbox{$t_s$}$}}
\put(1,4){{$\mbox{$t_s$}$}}
\put(5.75,2.9){{$\mbox{$s$}$}}
\put(5.75,1.5){{$\mbox{bar}$}}
\put(10.75,1){{$\mbox{$f_b$}$}}
\put(1,1){{$\mbox{$f_b$}$}}
\put(5.75,-0.15){{$\mbox{$b$}$}}
\end{picture}
\medskip
\caption{Adopted notation for bars and cables of a tensegrity system.}
\label{tensegrity_model}
\end{figure}
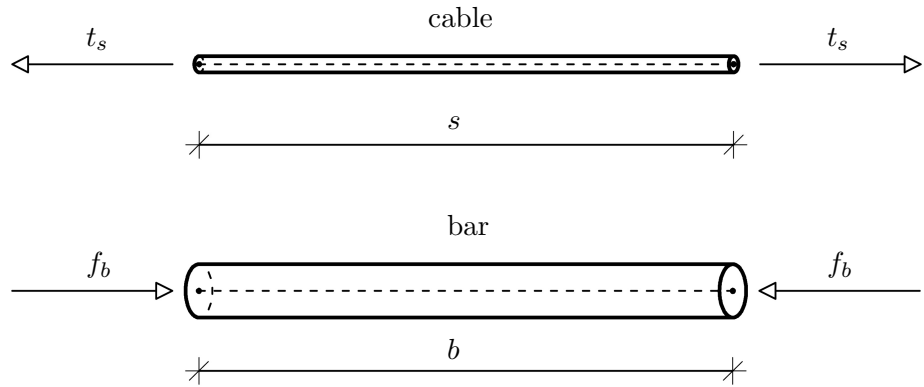

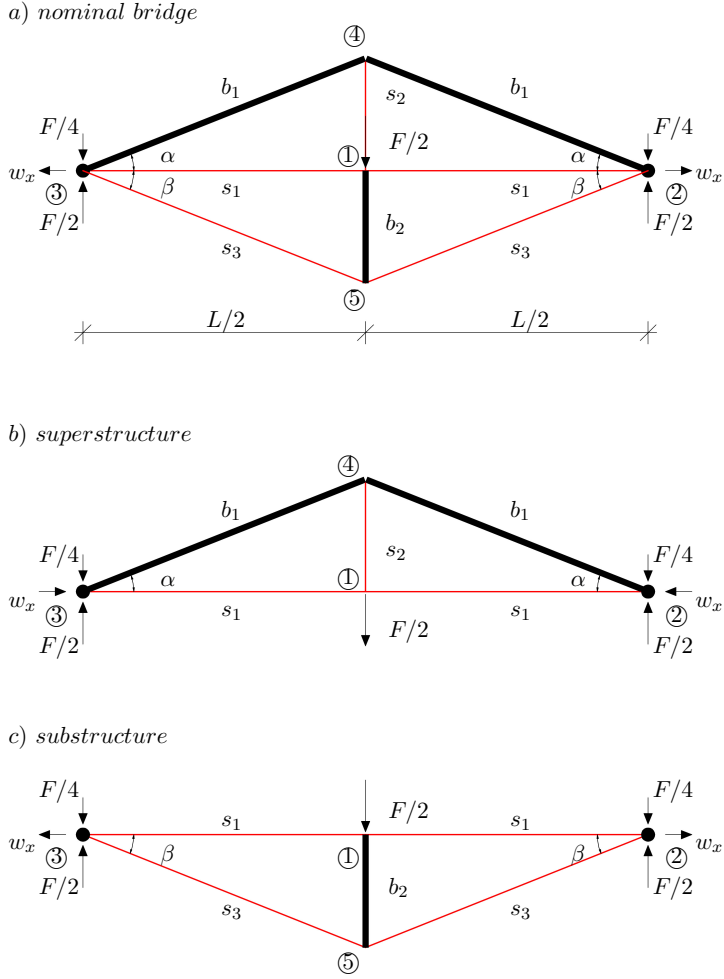
\begin{figure}[hbt] 
\unitlength1cm
\scalebox{0.8}{
\begin{picture}(10,16)
\put(2.5,10.1){\psfig{figure=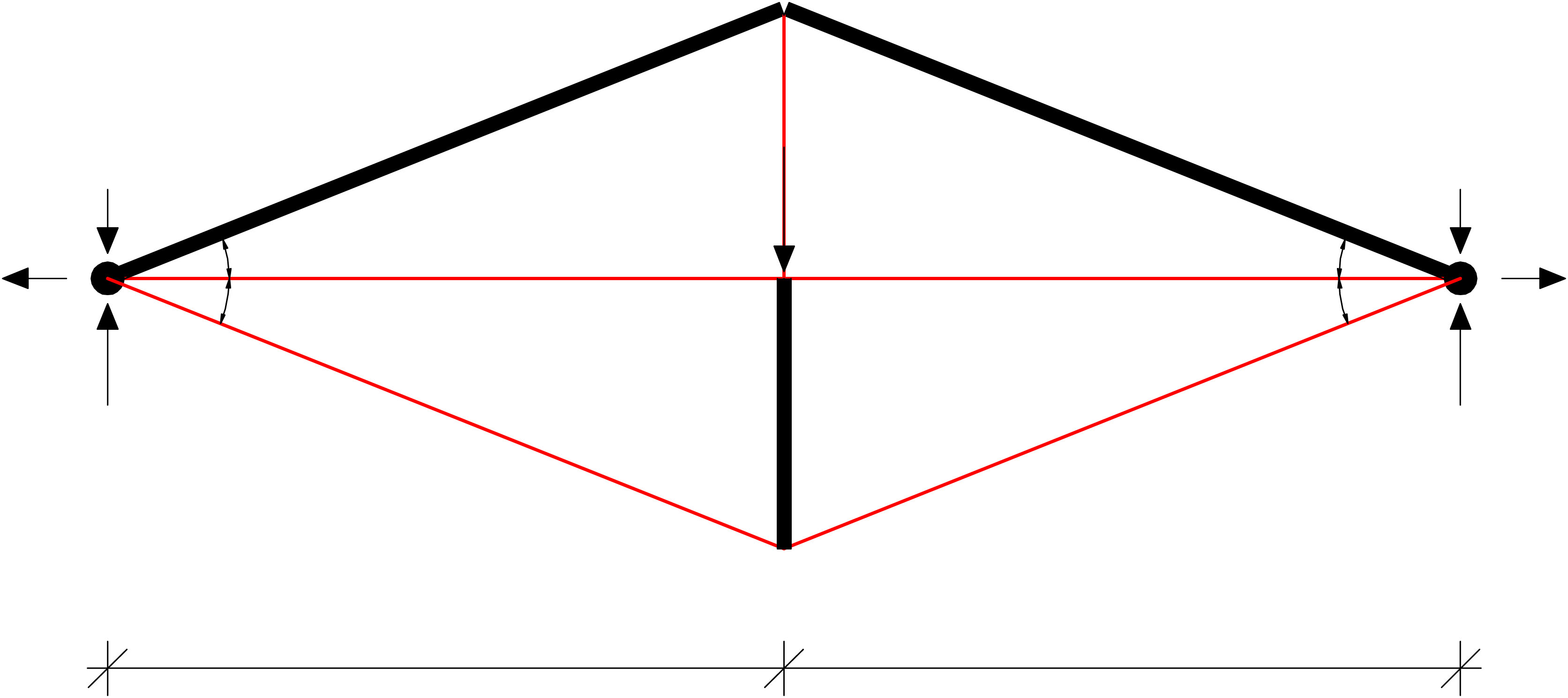,width=10.75cm}}
\put(2.5,5.1){\psfig{figure=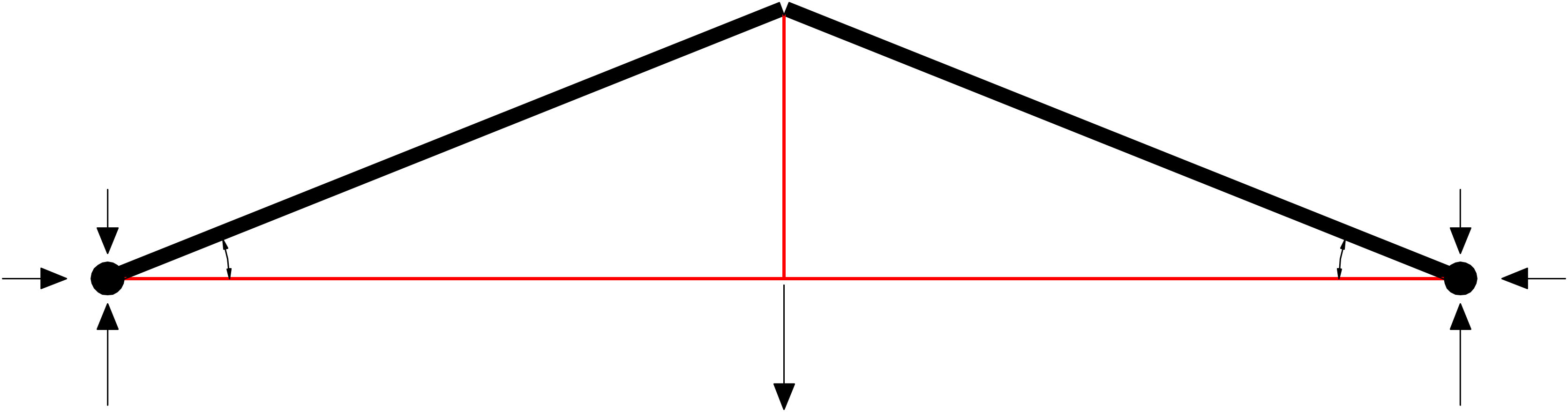,width=10.75cm}}
\put(2.5,0.1){\psfig{figure=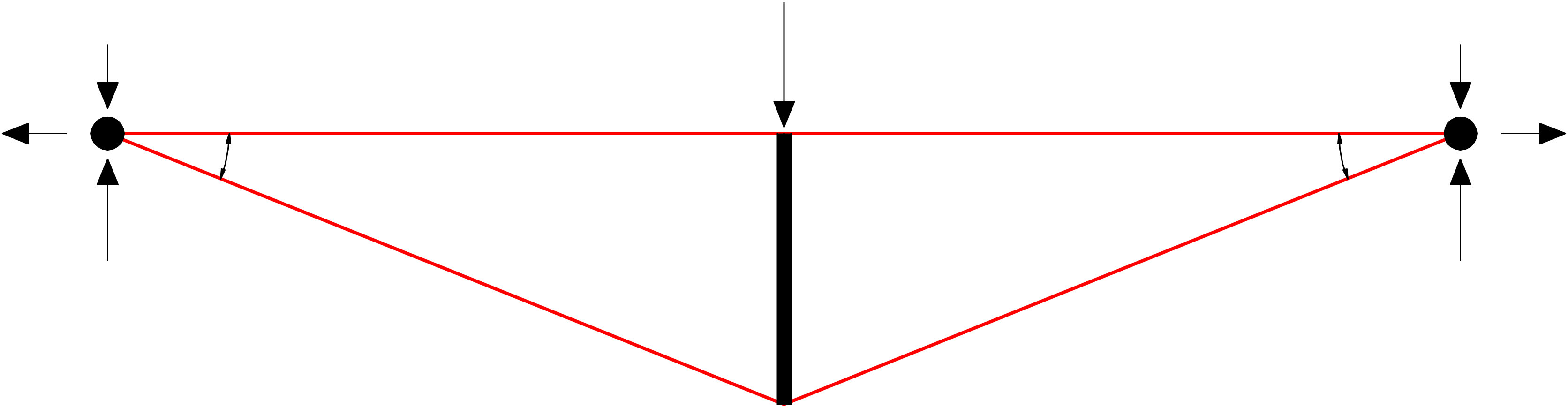,width=10.75cm}}
\put(2,15.5){{$\mbox{$a) \ nominal \ bridge$}$}}
\put(5.25,10.4){{$\mbox{$L/2$}$}}
\put(10.25,10.4){{$\mbox{$L/2$}$}}
\put(8.25,13.35){{$\mbox{$F/2$}$}}
\put(12.6,13.6){{$\mbox{$F/4$}$}}
\put(2.5,13.6){{$\mbox{$F/4$}$}}
\put(12.6,12){{$\mbox{$F/2$}$}}
\put(2.5,12){{$\mbox{$F/2$}$}}
\put(4.5,13.125){{$\mbox{$\alpha$}$}}
\put(11.25,13.125){{$\mbox{$\alpha$}$}}
\put(7.4,13.125){{$\mbox{$\textcircled{1}$}$}}
\put(12.8,12.5){{$\mbox{$\textcircled{2}$}$}}
\put(2.6,12.5){{$\mbox{$\textcircled{3}$}$}}
\put(7.5,15.125){{$\mbox{$\textcircled{4}$}$}}
\put(7.5,10.75){{$\mbox{$\textcircled{5}$}$}}
\put(4.5,12.6){{$\mbox{$\beta$}$}}
\put(11.25,12.6){{$\mbox{$\beta$}$}}
\put(10.25,11.6){{$\mbox{$s_3$}$}}
\put(5.5,11.6){{$\mbox{$s_3$}$}}
\put(10.25,14.25){{$\mbox{$b_1$}$}}
\put(5.5,14.25){{$\mbox{$b_1$}$}}
\put(8.2,14.15){{$\mbox{$s_2$}$}}
\put(8.2,12){{$\mbox{$b_2$}$}}
\put(10.25,12.6){{$\mbox{$s_1$}$}}
\put(5.5,12.6){{$\mbox{$s_1$}$}}
\put(2,12.85){{$\mbox{$w_x$}$}}
\put(13.3,12.85){{$\mbox{$w_x$}$}}
\put(2,8.5){{$\mbox{$b) \ superstructure $}$}}
\put(8.25,5.25){{$\mbox{$F/2$}$}}
\put(12.6,6.5){{$\mbox{$F/4$}$}}
\put(2.5,6.5){{$\mbox{$F/4$}$}}
\put(12.6,5){{$\mbox{$F/2$}$}}
\put(2.5,5){{$\mbox{$F/2$}$}}
\put(7.4,6.125){{$\mbox{$\textcircled{1}$}$}}
\put(12.8,5.5){{$\mbox{$\textcircled{2}$}$}}
\put(2.6,5.5){{$\mbox{$\textcircled{3}$}$}}
\put(7.4,8){{$\mbox{$\textcircled{4}$}$}}
\put(4.5,6.1){{$\mbox{$\alpha$}$}}
\put(11.25,6.1){{$\mbox{$\alpha$}$}}
\put(8.2,6.6){{$\mbox{$s_2$}$}}
\put(10.25,5.6){{$\mbox{$s_1$}$}}
\put(5.5,5.6){{$\mbox{$s_1$}$}}
\put(10.25,7.25){{$\mbox{$b_1$}$}}
\put(5.5,7.25){{$\mbox{$b_1$}$}}
\put(2,5.75){{$\mbox{$w_x$}$}}
\put(13.3,5.75){{$\mbox{$w_x$}$}}
\put(2,3.5){{$\mbox{$c) \ substructure$}$}}
\put(8.25,2.25){{$\mbox{$F/2$}$}}
\put(12.6,2.6){{$\mbox{$F/4$}$}}
\put(2.5,2.6){{$\mbox{$F/4$}$}}
\put(12.6,1){{$\mbox{$F/2$}$}}
\put(2.5,1){{$\mbox{$F/2$}$}}
\put(7.4,1.5){{$\mbox{$\textcircled{1}$}$}}
\put(12.8,1.5){{$\mbox{$\textcircled{2}$}$}}
\put(2.6,1.5){{$\mbox{$\textcircled{3}$}$}}
\put(7.4,-0.25){{$\mbox{$\textcircled{5}$}$}}
\put(4.5,1.55){{$\mbox{$\beta$}$}}
\put(11.25,1.55){{$\mbox{$\beta$}$}}
\put(10.25,0.65){{$\mbox{$s_3$}$}}
\put(5.5,0.65){{$\mbox{$s_3$}$}}
\put(10.25,2.125){{$\mbox{$s_1$}$}}
\put(5.5,2.125){{$\mbox{$s_1$}$}}
\put(8.25,1){{$\mbox{$b_2$}$}}
\put(2,1.75){{$\mbox{$w_x$}$}}
\put(13.3,1.75){{$\mbox{$w_x$}$}}
\end{picture}}
\caption{Basic modules of the tensegrity bridge with: a) nominal bridge: $n = q = p = 1$; b) superstructure: $n = q = 1$; c) substructure: $n = p = 1$.}
\label{basic_module}
\end{figure}

\begin{figure}[hbt]
\unitlength1cm
\scalebox{0.8}{
\begin{picture}(14,16)
\put(0,-0.5){\psfig{figure=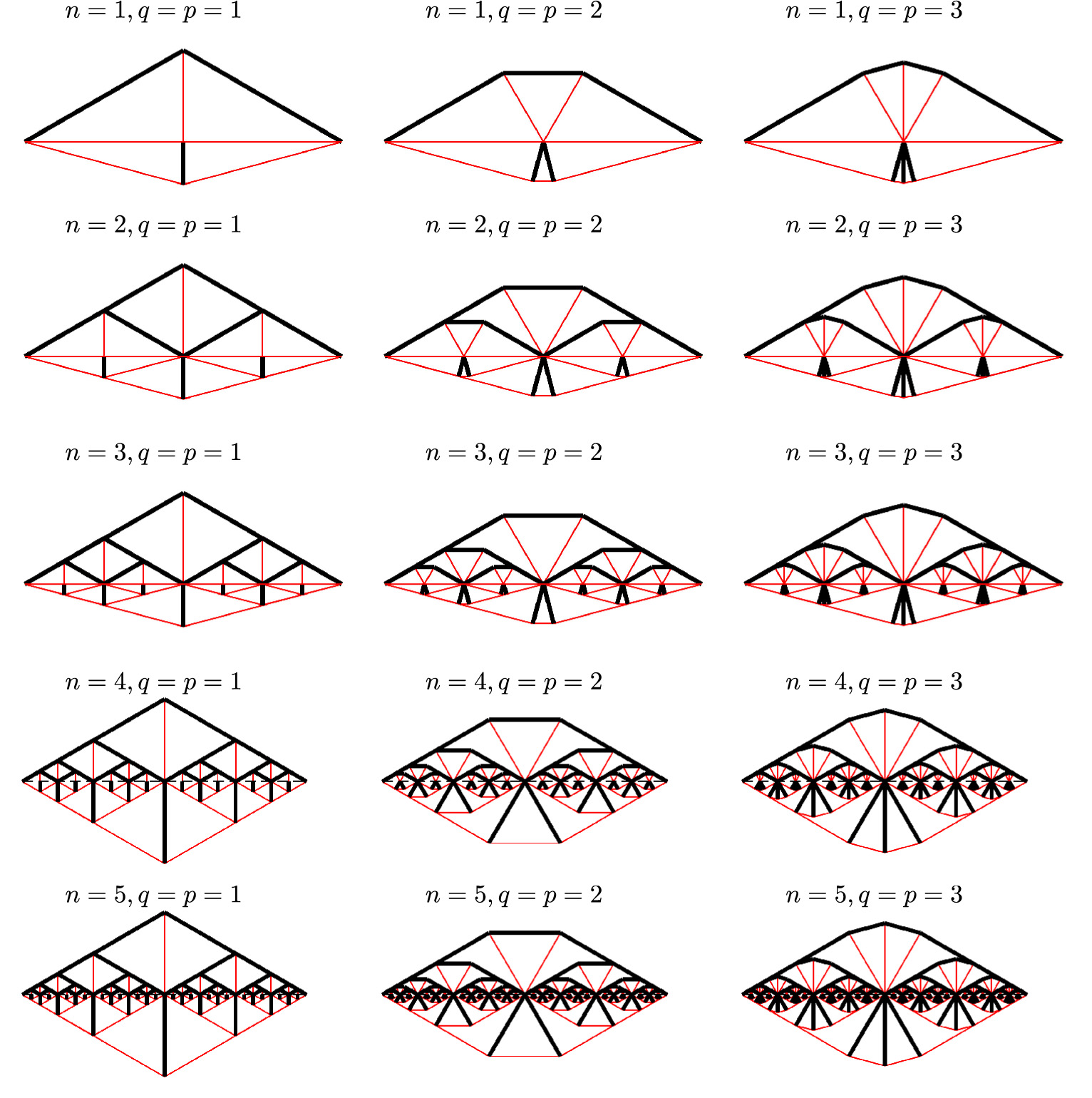,height=16cm}}
\end{picture}}
\caption{Exemplary geometries of the nominal bridges for different values of the complexity parameters $n$ (increasing downward) and $q$ (increasing leftward).}
\label{fig:dual_exemplaries}
\end{figure}

\begin{figure}[hbt]
\unitlength1cm
\scalebox{0.8}{
\begin{picture}(14,12.5)
\put(-0.5,-0.5){\psfig{figure=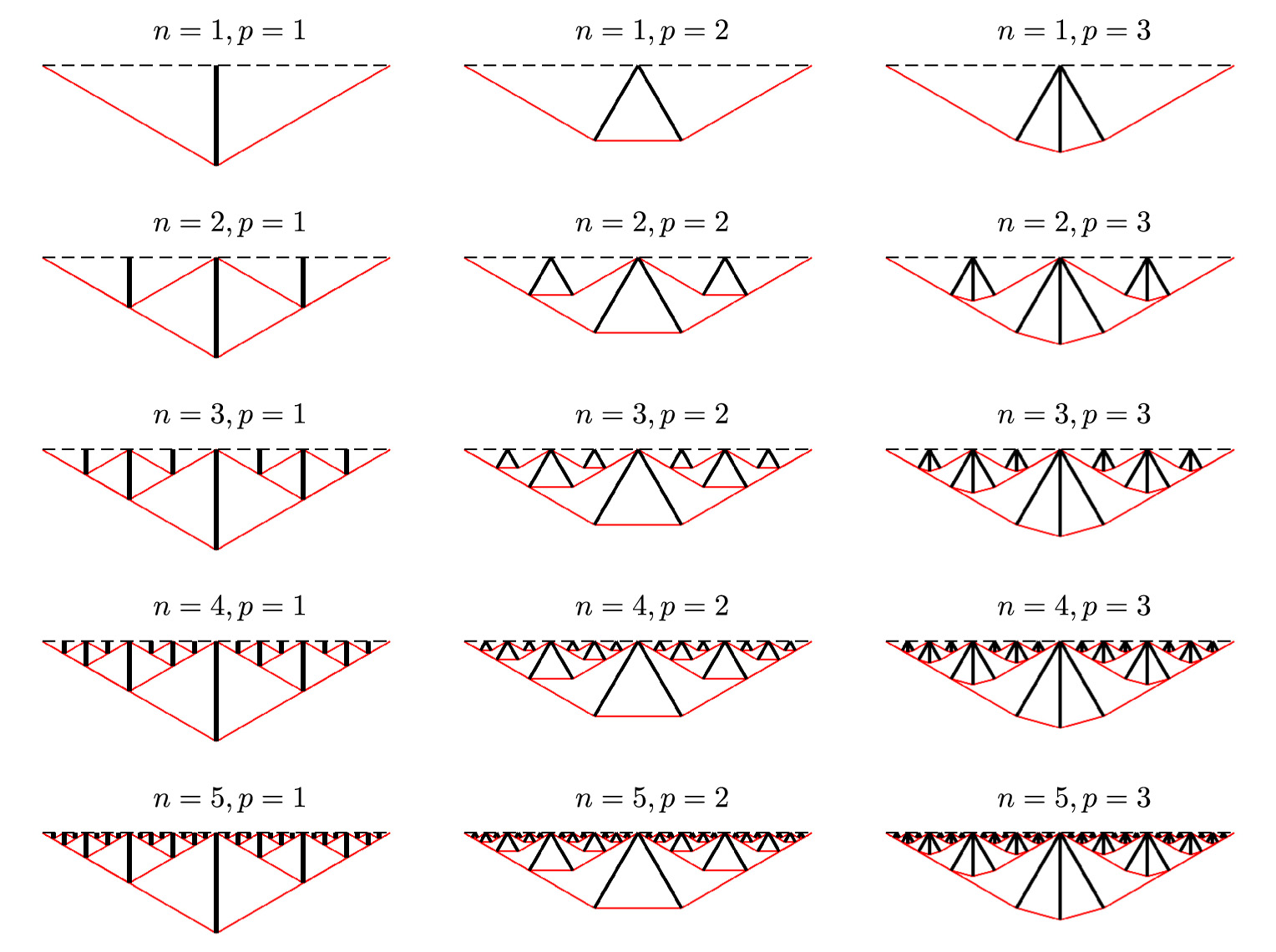,height=12.5cm}}
\end{picture}}
\caption{Exemplary geometries of the substructures for different values of the complexity parameters $n$ (increasing downward) and $p$ (increasing leftward).}
\label{fig:substructures_exemplaries}
\end{figure}

\begin{figure}[hbt]
\unitlength1cm
\scalebox{0.8}{
\begin{picture}(14,12.5)
\put(-0.5,-0.5){\psfig{figure=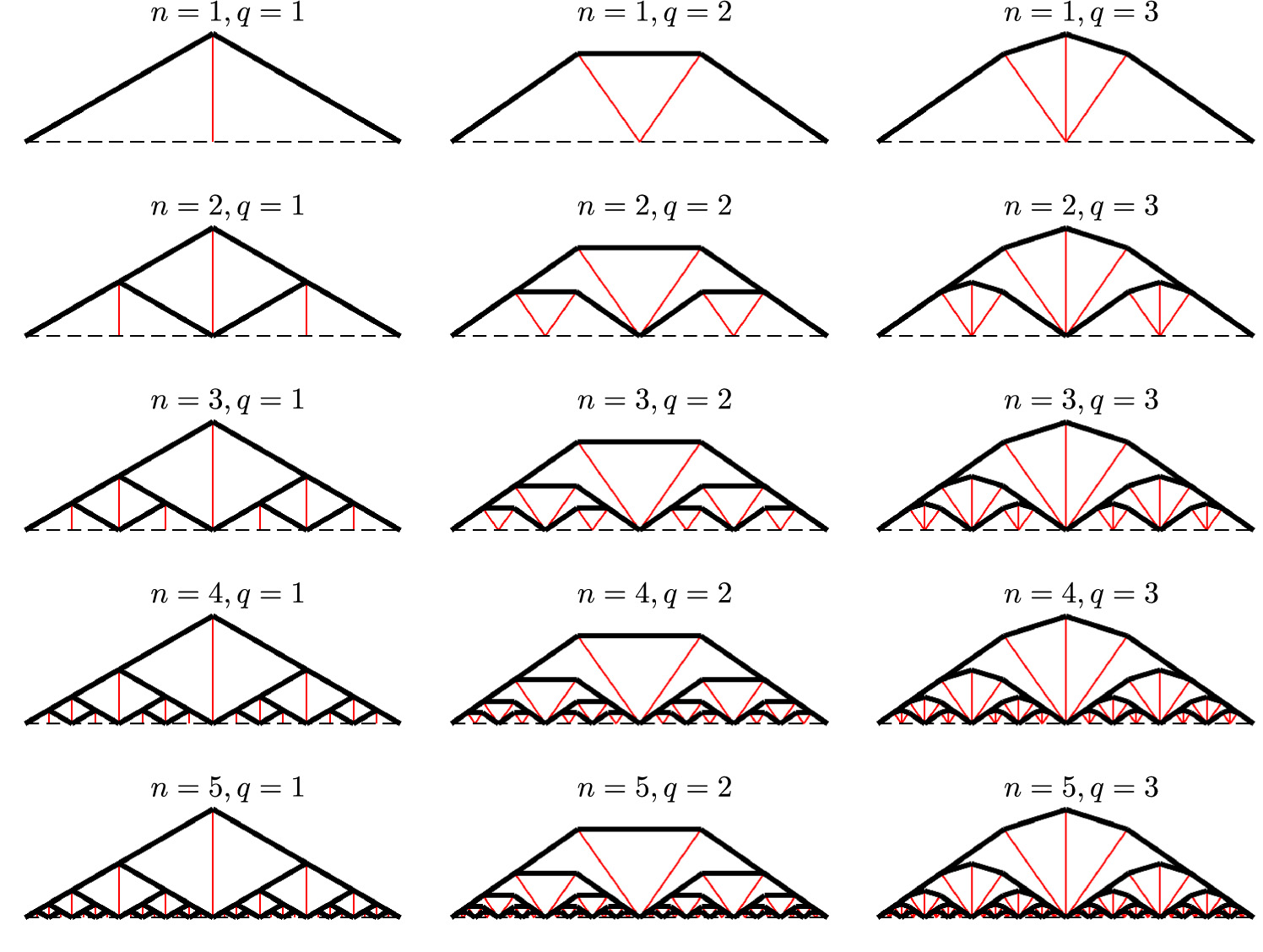,height=12.5cm}}
\end{picture}}
\caption{Exemplary geometries of the superstructures for different values of the complexity parameters $n$ (increasing downward) and $q$ (increasing leftward).}
\label{fig:superstructures_exemplaries}
\end{figure}


\section{Analysis of the Basic Modules ($n=1$, $p=1$ or $0$, $q=1$ or $0$)}

We first will examine the simplest of bridge concepts, as in Fig \ref{basic_module}.
Consider, first, the \emph{nominal} bridge, subject to yield constraints, with complexity $(n,p,q)=(1,1,1)$. This configuration, described by Fig \ref{basic_module}a, is composed of 5 cables and 3 bars. Let the bottom end of each compressive member above the deck be constrained by a hinge boundary condition, so as to allow rotation but not translation. Define $F$ as the total applied load, and $L$ as the span. All cables use the same material, and all bars use the same material. It will be convenient to define the following constants:

\bea
\rho = \frac{{\rho_b}/{\sigma_b}}{{\rho_s}/{\sigma_s}},
\label{rho}
\eea

\bea
\eta = \frac{\rho_b L}{\left( \rho_s/\sigma_s \right) \sqrt{\pi E_b F}}.
\label{eta}
\eea

\noindent Define a normalization of the system mass $m$ by the dimensionless quantity $\mu$:

\bea
\mu = \frac{m}{\left( \rho_s/\sigma_s \right) F L},
\label{mu}
\eea

\noindent where the mass $m$ at the yield condition is:

\bea
m=\frac{\rho_b}{\sigma_b}\sum{f_i b_i}+\frac{\rho_s}{\sigma_s}\sum{t_i s_i},
\label{mass}
\eea

\noindent where ($b_i$,$s_i$) is respectively the length of the $i^{th}$ bar or $i^{th}$ cable, and respectively ($f_i$,$t_i$) is the force in the $i^{th}$ bar or cable.

\noindent The mass of the \emph{nominal} bridge will be minimized over the choice of angles $\alpha$ and $\beta$. The lengths of the members are:

\begin{align}
\label{n1q1p1:topology}
\nonumber
s_1&=\frac{L}{2},\ \ \ s_2=\frac{L}{2} \tan \alpha,\ \ \ s_3=\frac{L}{2 \cos \beta} = \frac{L}{2} \sqrt{1+\tan^2 \beta}, \\ 
b_1& = \frac{L}{2 \cos \alpha} = \frac{L}{2} \sqrt{1+\tan^2 \alpha}, \ \ \ b_2 = \frac{L}{2} \tan \beta.  
\end{align}

\noindent The equilibrium equations at each node are:

\bea
t_1 + t_3 \cos \beta = w_{x} + f_1 \cos \alpha, \nonumber \\ 
\frac{F}{4} = f_1 \sin \alpha + t_3 \sin \beta, \nonumber \\ 
t_2 = 2 f_1 \sin \alpha, \nonumber \\ 
f_2 = 2 t_3 \sin \beta, \nonumber \\ 
\frac{F}{2} = t_2 + f_2.
\label{n1q1p1:equilibrium}
\eea

\noindent This system of equations  can be solved, choosing $t_1$ and $t_3$ are free independent parameters:

\bea
\frac{f_1}{F}=\frac{\sqrt{1+\tan^2 \alpha}}{4\tan\alpha}(1-\frac{t_3}{F}\frac{4\tan\beta}{\sqrt{1+\tan^2 \beta}}), \nonumber \\
\frac{f_2}{F}=\frac{t_3}{F}\frac{2 \tan\beta}{\sqrt{1+\tan^2 \beta}}, \nonumber \\
\frac{t_2}{F}=\frac{1}{2} -\frac{t_3}{F} \frac{2\tan\beta}{\sqrt{1+tan^2\beta}}\label{t2}, \nonumber \\
\frac{w_x}{F}=\frac{t_1}{F}+\frac{t_3}{F}\frac{\tan\alpha+\tan\beta}{\tan\alpha\sqrt{1+\tan^2\beta}}-\frac{1}{4\tan\alpha}.
\label{forces_t1_t3}
\eea

\subsection{Nominal Bridges under Yielding Constraints}
\label{n1q1p1_model}

\begin {theorem} \label{NY11}
Given the \emph{nominal} bridge with complexity $(n,p,q)=(1,1,1)$ (described in Fig. \ref{basic_module}a), with attendant data (\ref{n1q1p1:topology}), the minimal mass can be expressed in terms of independent variables $t_1$ and $t_3$:

\bea
\mu_{Y}(t_1,t_3) = \frac{t_1}{F} + \frac{t_3}{F} c_3\left(\alpha, \beta, \rho\right) +\frac{b_\alpha}{4}\label{MUY},
\label{muY_t1_t3}
\eea

\noindent where:

\bea
c_3(\alpha,\beta,\rho)=\frac{(1+\rho)\tan^2\beta-b_{\alpha}\tan\beta+1}{\sqrt{1+\tan^2\beta}},\qquad b_{\alpha}=\frac{\rho+(1+\rho)\tan^2 \alpha}{\tan\alpha}. \label{c_3}
\eea

\end{theorem}

\noindent An alternate expression for the mass can be written by substituting  the relation between $t_2$ and $t_3$ from (27), to get an equivalent expression $\mu_Y(t_1,t_2)=\mu_Y(t_1,t_3)$, where:

\begin{align}
 \frac{t_3}{F}&=\frac{\sqrt{1+\tan^2\beta}(1-2t_2 /F)}{4\tan\beta}, \label{n1qp1_t2_t3} \\
 \mu_Y(t_1,t_2)&=\frac{t_1}{F}+\frac{t_2}{F}c_2(\alpha,\beta,\rho)+\frac{(1+\rho)\tan^2 \beta +1}{4\tan\beta},\\
 c_2(\alpha,\beta,\rho)&=-c_3\frac{\sqrt{1+tan^2\beta}}{2\tan\beta}\qquad =-\frac{(1+\rho)tan^2\beta-b_\alpha \tan\beta +1}{2\tan\beta} \label{c_2}.
 \end{align}

\noindent Hence it follows that the minimal mass solution requires $t_3>0$ if and only if  $c_3<0$ (equivalently $c_2>0$). Note also that $c_3<0$ if and only if:

\bea
\frac{1+(1+\rho)\tan^2 \beta}{\tan\beta}<\frac{\rho+(1+\rho)\tan^2 \alpha}{\tan\alpha}.
\eea

\noindent Conversely, minimal mass requires $t_3=0$ if $c_3>0$ (equivalently $c_2<0$). This event occurs if and only if:

\bea
\frac{1+(1+\rho)\tan^2 \beta}{\tan\beta}>\frac{\rho+(1+\rho)\tan^2 \alpha}{\tan\alpha}.
\eea

\noindent Finally, $c_3=0$ (and also $c_2=0$) if and only if:

\bea
\frac{1+(1+\rho)\tan^2 \beta}{\tan\beta}=\frac{\rho+(1+\rho)\tan^2 \alpha}{\tan\alpha}.
\eea

\noindent Note also that the requirement that $t_2$ and $t_3$ both be non-negative values limits the feasible range of $t_3$ such that:

\bea
0 \le t_3 \le \frac{F\sqrt{1+\tan^2\beta}}{4\tan\beta}.
\eea

\noindent Given the relation between $t_2$ and $t_3$ in (\ref{n1qp1_t2_t3}) we have the corresponding feasible range for $t_2$:

\bea
0 \le t_2 \le \frac{F}{2}.
\eea

\noindent The proof of the theorem follows the mass calculation in (\ref{mu}), (\ref{mass}) after substituting the equilibrium forces given by (\ref{forces_t1_t3}).

\begin{corollary} \label{theo:n1q1:min_yielding}

Consider a \emph{superstructure} bridge with complexity $(n,p,q) = (1,0,1)$ (topology is defined by Fig. \ref{basic_module}b). The minimal mass $\mu_Y$ requires the following aspect angle:

\bea
\alpha^*_{Y} =\arctan \left( \sqrt{ \frac{\rho}{1+\rho} } \right),
\label{n1q1:alphastY}
\eea

\noindent which corresponds to the following dimensionless minimal mass: 

\bea
\mu_{Y}^{*} =\frac{1}{2} \sqrt{ \rho \left( 1+\rho \right) }. 
\label{n1q1:mustY}
\eea

\end{corollary} \label{theo:n1q1:min_yielding}

\begin{proof}
The mass of the superstructure can be obtained from Theorem (\ref{NY11}) by setting $t_1 = 0$ since its coefficient is positive, and $t_3 = 0$ since the cable $s_3$ is absent. Thus,

\bea
\mu_{Y}= \frac{\tan \alpha}{4} + \rho \frac{ \left( 1+\tan^2 \alpha \right) }{4 \tan \alpha}.
\label{n1q1:muy}
\eea

\noindent This function has a unique minimum satisfying, 

\bea
\frac{\partial \mu_{Y}}{\partial\tan \alpha}= \frac{\tan^2 \alpha+\rho \left(\tan^2 \alpha -1 \right)}{4\tan^2 \alpha}=0,
\label{n1q1:dmuy}
\eea

\noindent producing the unique optimal angle (\ref{n1q1:alphastY}). Substituting this angle into (\ref{n1q1:muy}) concludes the proof.

\end{proof}

\begin{figure}[hb]
\unitlength1cm
\begin{picture}(10,5)
\put(0.25,1){\psfig{figure=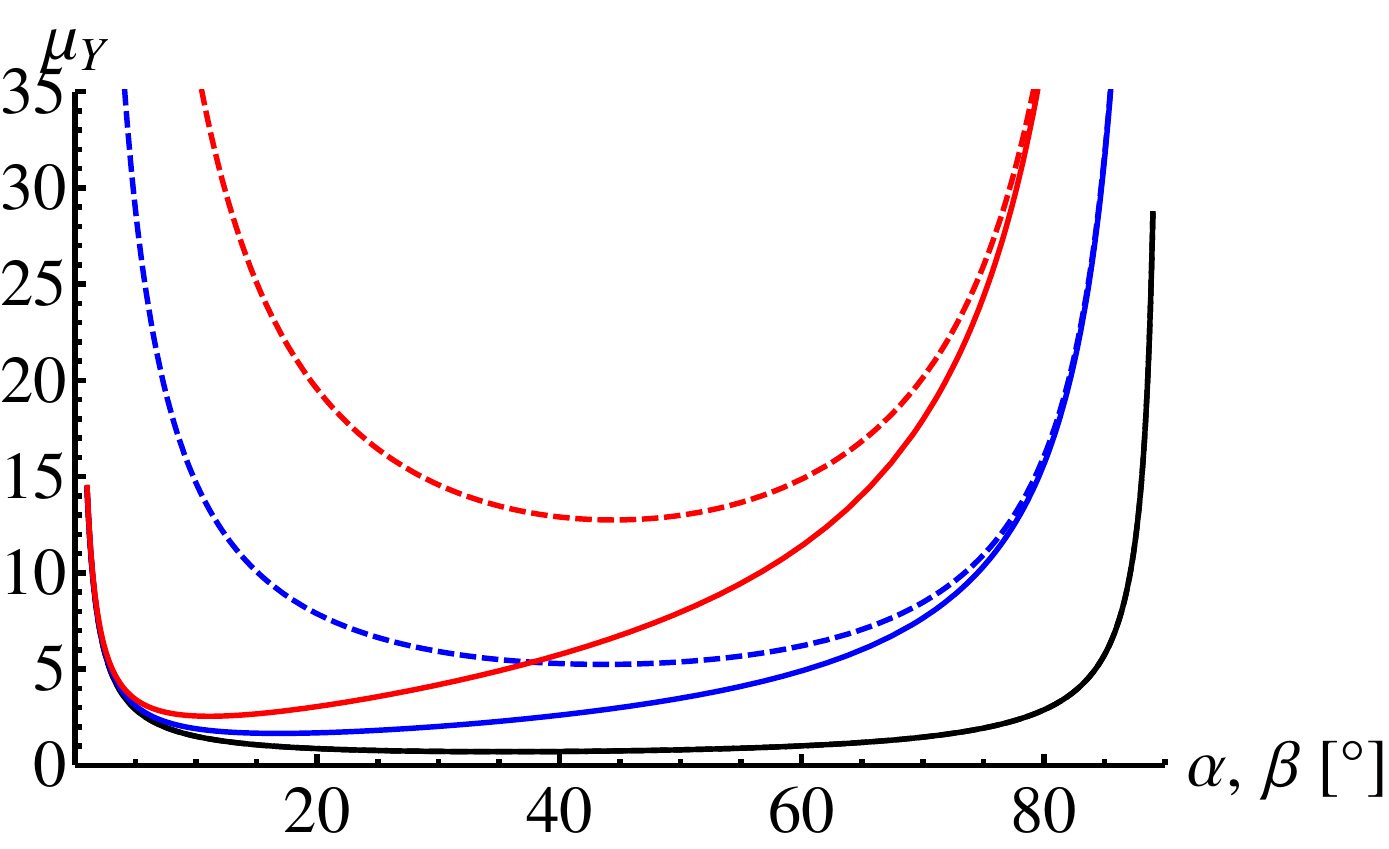,height=4.5cm}}
\put(8.25,1){\psfig{figure=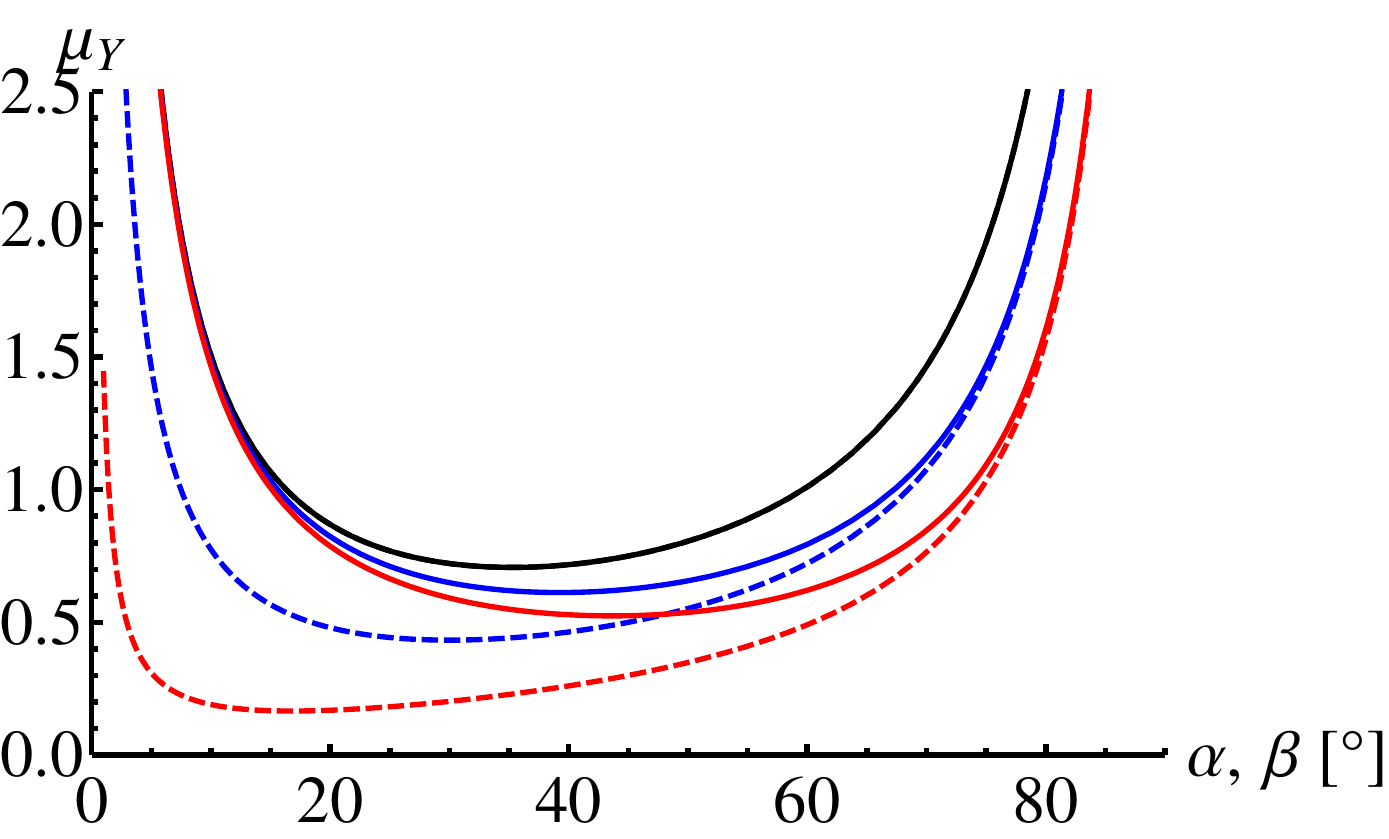,height=4.5cm}}
\put(0,-0.25){\psfig{figure=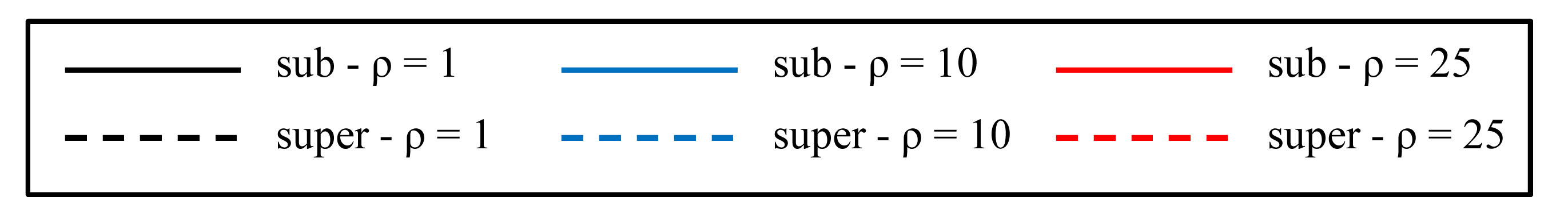,height=1cm}}
\put(8.5,-0.25){\psfig{figure=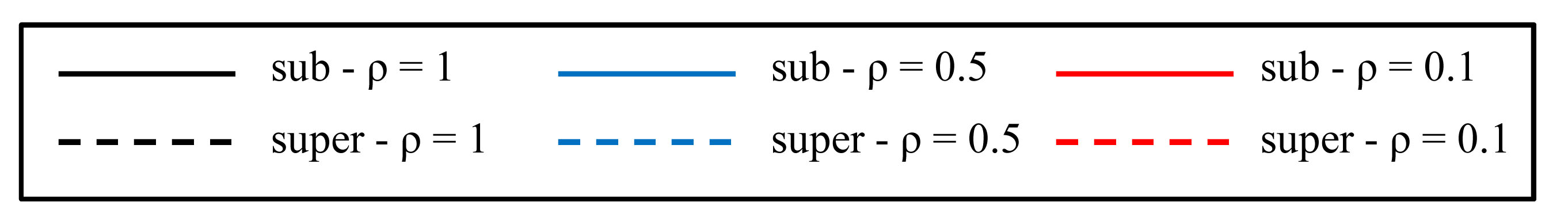,height=1cm}}
\end{picture}
\caption{Dimensionless masses of the substructure (continuous curves) and superstructure (dashed curves) for different values of the aspect angles (respectively $\beta$ or $\alpha$) and for values of the coefficient $\rho>1$ (left) and $\rho<1$ (right) under yielding constraints.}
\label{fig:n1pq1_yielding_graph}
\end{figure}

Fig \ref{fig:n1pq1_yielding_graph} plots the mass versus the angle $\beta$ and $\alpha$, yielding the minimum at the values given by (\ref{n1p1:betastY}) and (\ref{n1q1:alphastY}). All designs in this section assume failure by yielding. One must check that yielding is indeed the mode of failure.

\begin{corollary} \label{theo:n1p1:min_yielding}

Consider a \emph{substructure} bridge, with complexity $(n,p,q)=(1,1,0)$ (topology is defined by Fig. \ref{basic_module}c). The minimal mass design under only yield constraints is given by the following aspect angle:

\bea
\beta^*_{Y} =\arctan \left(\frac{1}{\sqrt{1+\rho } } \right),
\label{n1p1:betastY}
\eea

\noindent which corresponds to the following dimensionless minimal mass: 

\bea
\mu_{Y}^{*} =\frac{\sqrt{1+\rho}}{2} .
\label{n1p1:mustY}
\eea

\end{corollary} \label{theo:n1p1:min_yielding}

\begin{proof}
The mass of the substructure can be obtained from Theorem (\ref{NY11}) with $t_1 = t_2 = 0$ to obtain,

\bea
\mu_{Y}= \frac{ \left( 1+\tan^2 \beta \right) }{4 \tan \beta} +\frac{\rho}{4} \tan \beta.
\label{n1p1:muy}
\eea

\noindent The the unique minimum satisfies,

\bea
\frac{\partial \mu_{Y}}{\partial \tan \beta}=  - \frac{1+\tan^2 \beta}{4 \tan^2 \beta}+\frac{1}{2}+\frac{\rho}{4}=0,
\label{n1p1:dmuy}
\eea

\noindent producing the optimal optimal angle of (\ref{n1p1:betastY}). Substituting this angle into (\ref{n1p1:muy}) concludes the proof.
\end{proof}

\begin{corollary} \label{theo:yielding_sup}

For the designs in this section, yielding is indeed the mode of failure if the following inequalities hold:

\bea
\frac{F}{L^2}>\frac{1}{2 \left( 1 + \rho \right)} \left( \frac{4 \sigma_b^2}{\pi E_b} \right), \ \ \ &if:& 0 < \rho \le \frac{1}{4} \left(\sqrt{3}-1\right), \label{chi_1} \\ 
\frac{F}{L^2}>\frac{ \sqrt{ \rho \left( 1 + 2 \rho \right) } }{1 + \rho}  \left( \frac{4 \sigma_b^2}{\pi E_b} \right), \ \ \ &if:& \rho > \frac{1}{4} \left(\sqrt{3}-1\right). 
\label{chi_2}
\eea

\noindent  In addition, if \ $0 < \rho \le \frac{1}{4} \left(\sqrt{3}-1\right)$ and (\ref{chi_1}) holds or if $\frac{1}{4} \left(\sqrt{3}-1\right) < \rho < 1$ and (\ref{chi_2}) holds, then the minimal mass of a \emph{superstructure} bridge is less than the minimal mass of a \emph{substructure} bridge. (In this event, the minimal mass bridge is superstructure only). If $\rho=1$ and (\ref{chi_2}) also holds, then the minimal mass of the substructure bridge is equal to the minimal mass of the superstructure bridge.  If $\rho>1$ and (\ref{chi_2}) also hold, then the minimal mass of the substructure bridge is less than the minimal mass of the superstructure bridge. (The minimal mass bridge is substructure only).

\end{corollary}

\begin{proof}
Under yield constraints, if the design has the property ${f_{b,i}}/{b_i^2}> {4 \sigma_b^2}/({\pi E_b})$, then this guarantees that yielding is the mode of failure in bar $b_i$, and the buckling constraints are also satisfied (see lemma \ref{theo:yield_design_2}). For the superstructure, assuming the optimal angle (\ref{n1q1:alphastY}), the minimal mass (\ref{n1q1:mustY}), the force $f_1$  (\ref{forces_t1_t3}) and the length $b_1$ (\ref{n1q1p1:topology}), the lemma \ref{theo:yield_design_2} reduces to:

\bea
\frac{F}{L^2}>\frac{ \sqrt{ \rho \left( 1 + 2 \rho \right) } }{1 + \rho}  \left( \frac{4 \sigma_b^2}{\pi E_b} \right).
\label{FL_chi1}
\eea

\noindent Similarly, for the substructure, assuming the optimal angle (\ref{n1p1:betastY}), the minimal mass (\ref{n1p1:mustY}), the force $f_2$ (\ref{forces_t1_t3}) and the length $b_2$ (\ref{n1q1p1:topology}), the lemma (\ref{theo:yield_design_2}) reduces to:

\bea
\frac{F}{L^2}>\frac{1}{2 \left( 1 + \rho \right)}\left( \frac{4 \sigma_b^2}{\pi E_b} \right).
\label{FL_chi2}
\eea

\noindent Yielding is the mode of failure of superstructure and substructure desgns if both (\ref{FL_chi1}) and (\ref{FL_chi2}) hold or, equivalently, if the following holds:

\bea
\frac{F}{L^2}>max\left[ \frac{ \sqrt{ \rho \left( 1 + 2 \rho \right) } }{1 + \rho} , \frac{1}{2 \left( 1 + \rho \right)}\right] \left( \frac{4 \sigma_b^2}{\pi E_b} \right).
\label{FL_chi_max}
\eea

\noindent From the inequality $\frac{ \sqrt{ \rho \left( 1 + 2 \rho \right) } }{1 + \rho}  / \frac{1}{2 \left( 1 + \rho \right)}> 1$ we obtain the following conclusions:

\bea
\frac{ \sqrt{ \rho \left( 1 + 2 \rho \right) } }{1 + \rho}  < \frac{1}{2 \left( 1 + \rho \right)}, &if:&  0 < \rho \le \frac{1}{4} \left(\sqrt{3}-1\right), \label{chi_max_solution1} \\
\frac{ \sqrt{ \rho \left( 1 + 2 \rho \right) } }{1 + \rho}  >\frac{1}{2 \left( 1 + \rho \right)}, &if:&  \rho > \frac{1}{4} \left(\sqrt{3}-1\right).
\label{chi_max_solution2}
\eea

\noindent Equations (\ref{chi_max_solution1}) and (\ref{chi_max_solution2}) combined with (\ref{FL_chi_max}) give the conditions (\ref{chi_1}) and (\ref{chi_2}).
The mass of the substructure is shown to be less that the mass of the superstructure if $\rho>1$, a result that follows by taking the ratio between the optimal mass of the superstructure (\ref{n1q1:mustY}) and the optimal mass of the substructure (\ref{n1p1:mustY}). 

\end{proof}

As a practical matter, $\rho$ is almost always greater than $1$, since compressive members tend to have higher mass density than tension members ($\rho_b/\rho_s>1$), and the yield strength of tensile material tends to be greater than for compressive members ($\sigma_s/\sigma_b>1$). 

Thus far the conclusion is that if $\rho> \frac{1}{4} \left(\sqrt{3}-1\right)$ then the bridge in Fig. \ref{basic_module}a at its minimal mass configuration becomes the configuration of substructure in Fig. \ref{basic_module}c, if the bridge design is constrained against yielding. Furthermore, such a design will not buckle. Note that this design produced a topology constrained against yielding, and a design constrained against buckling might produce a different topology. Now  lets consider this possibility. 

\subsection{Nominal Bridges under Buckling Constraints}
\label{Sect:n1q1p1:B}

This section repeats all the designs of the previous section (for the three structures of Fig. \ref{basic_module}) with the added constraint that the bars cannot buckle.

\begin{theorem} \label{theo:n1q1p1:min_buckling}

Consider a\emph{ nominal} bridge of complexity $(n,p,q)=(1,1,1)$.  The minimal mass (the cable mass required at the yield conditions plus the bar mass required at the bar buckling conditions), is, in terms of $t_1$ and $t_3$:

\bea
\mu_{B}(t_1,t_3) = \frac{t_1}{F} + \frac{t_3}{F} \frac{\tan^2\beta-\tan\alpha \tan\beta+1}{\sqrt{1 + \tan^2 \beta}}+\frac{\tan \alpha}{4} \nonumber \\
+ \eta \left[ \frac{ \left( 1 + \tan^2 \alpha \right)^{5/4} }{2 \sqrt{\tan \alpha}} \left( 1 - \frac{t_3}{F} \frac{4 \tan \beta}{\sqrt{1 + \tan^2 \beta}} \right)^{1/2} + \frac{\tan^2 \beta}{\sqrt{2}} \sqrt{\frac{t_3}{F} \frac{\tan \beta}{ \left( 1 + \tan^2 \beta \right)^{1/2} } } \right],
\label{n1q1p1:muB_t1t3}
\eea

\noindent or, equivalently, in terms of $t_1$ and $t_2$:

\bea
\mu_{B}(t_1,t_2) = \frac{t_1}{F} + \frac{t_2}{F} \left[ \frac{\tan \alpha}{2} - \frac{ \left( 1+\tan^2 \beta \right) }{2 \tan \beta} \right]+\frac{ \left( 1+\tan^2 \beta \right) }{4 \tan \beta} \nonumber \\
+\eta \left[ \left( 1+\tan^2 \alpha \right)^{(5/4)} \sqrt{\frac{t_2}{2 F \tan \alpha}} + \frac{\tan^2 \beta}{2} \sqrt{\frac{1}{2}-\frac{t_2}{F}} \right].
\label{n1q1p1:muB_t1t2}
\eea

\end{theorem}

\begin{proof}

\noindent Given the solution (\ref{forces_t1_t3}), the total mass of bars is:

\bea
m_{b,B} = \frac{\rho_b L^2 \sqrt{F}}{\sqrt{\pi E_b}} \left[ \frac{ \left( 1 + \tan^2 \alpha \right)^{5/4}}{2 \sqrt{\tan \alpha}} \left( 1 - \frac{t_3}{F} \frac{4 \tan \beta}{ \sqrt{1 + \tan^2 \beta} } \right)^{1/2} + \frac{\tan^2 \beta}{\sqrt{2}} \sqrt{ \frac{t_3}{F} \frac{ \tan \beta}{ \left( 1 + \tan^2 \beta \right)^{1/2} } } \right]. \nonumber \\
\label{n1q1p1:mbB_t1t3}
\eea

\noindent Adding to (\ref{n1q1p1:mbB_t1t3}) the total mass of cables and using the (\ref{forces_t1_t3}), we obtain the total mass of (\ref{n1q1p1:muB_t1t3}) given in the theorem. It is also possible to write this mass in terms of  free parameters $t_t$ and $t_2$.

\bea
m_{b,B} = \frac{\rho_b L^2 \sqrt{F}}{\sqrt{\pi E_b}} \left[ \left( 1 + \tan^2 \alpha \right)^{5/4} \sqrt{\frac{t_2}{2 F \tan \alpha}} + \frac{\tan^2 \beta}{2} \sqrt{\frac{1}{2} - \frac{t_2}{F}} \right].
\label{n1q1p1:mbB_t1t2}
\eea

\noindent Adding to (\ref{n1q1p1:mbB_t1t2}) the total mass of cables and using the (\ref{forces_t1_t3}), we obtain the total mass of (\ref{n1q1p1:muB_t1t2}).

\end{proof}

The value of $\beta=4.25 \ deg$ minimizes  the mass (\ref{n1q1p1:muB_t1t2}) if the material choice is steel ($\rho = 7862 \ kg/m^3$; $\sigma = 6.9x10^8 \ N/m^2$; $E = 2.06x10^{11} \ N/m^2$). It will become clear that the minimal mass solution of the minimal bridge $\mu_B$, constrained against buckling, will reduce to only a \emph{substructure} (Fig 3c). It is straightforward to show that the mass of the bars is much greater than the mass of the cables under the usual condition:

\bea\label{eta2}
\eta \gg \frac{\tan^2\alpha}{2(1+\tan^2 \alpha)^{5/4}}.
\eea

To prepare for those insights, now consider the individual solutions for designs constrained to be only \emph{superstructure} or only \emph{substructure} in configuration.

\begin{figure}[hb]
\unitlength1cm
\begin{picture}(10,5)
\put(5,0){\psfig{figure=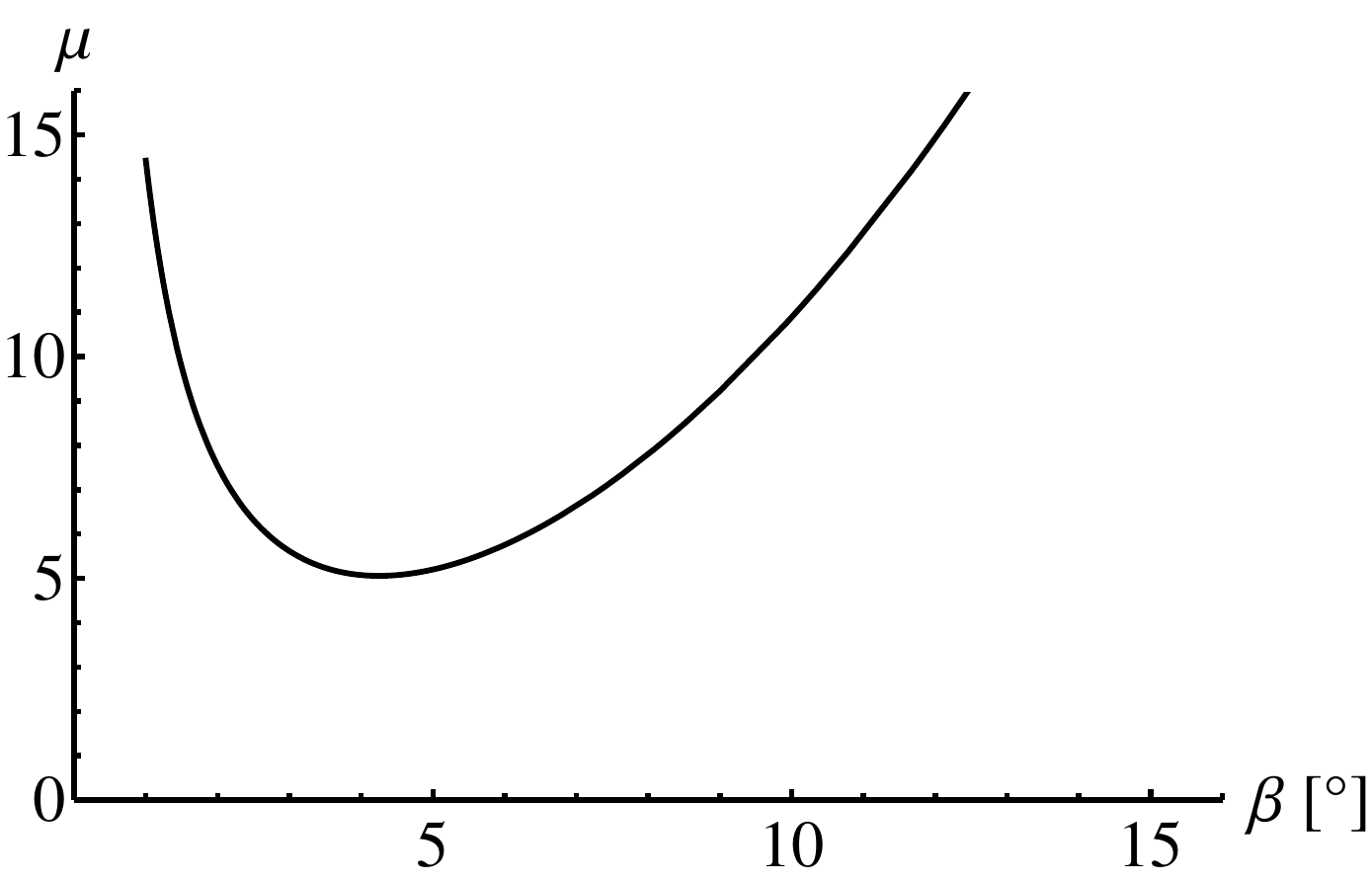,height=5cm}}
\end{picture}
\caption{Mass $\mu_B(t_1, t_2)$ (Eq. \ref{n1q1p1:muB_t1t2}) for different values of the aspect angles $\beta$ assuming steel bars and cables, $F = 1 \ N$, $L=1 \ m$ ($\eta = 857.71$), $t_1 = 0$ and $t_2 = 0$. The minimum value is $\mu^*_B = 5.0574$ at $\beta^*_B = 4.25$ $deg$.}
\label{fig:plot_muBt1t2}
\end{figure}

\begin{corollary} \label{theo:n1q1:min_buckling}

Consider a \emph{superstructure} bridge of complexity $(n,p,q)=(1,0,1)$, (Fig. \ref{basic_module}b). Suppose (\ref{eta2}) holds. The minimal mass design under yielding and buckling constraints is given by the following aspect angle:

\bea
\bar \alpha^*_{B} = \arctan \left( \frac{1}{2} \right), 
\label{n1q1:alphastB}
\eea

\noindent which corresponds to the following dimensionless minimal mass: 

\bea
\bar \mu_{B}^{*} =\frac{1}{8} \left( 1+5^{(5/4)} \eta \right).
\label{n1q1:mustB}
\eea

\end{corollary} \label{theo:n1q1:min_buckling}

\begin{proof}
\noindent The mass of the \emph{superstructure} only case can be obtained from (\ref{n1q1p1:muB_t1t3}) assuming $t_1 = t_3 = 0$:

\bea
\mu_B = \frac{\tan \alpha}{4} +\eta \frac{ \left( 1+\tan^2 \alpha \right)^{(5/4)}}{2 \sqrt{\tan \alpha}}.
\label{n1q1:muB}
\eea

\noindent Assuming that the mass of the cables, which is the first term at the rhs of the (\ref{n1q1:muB}) is neglectable if compared with the mass of the bars, which is the second term at the rhs of (\ref{n1q1:muB}). Then the dimensionless mass becomes

\bea
\bar \mu_B = \eta \frac{ \left( 1+\tan^2 \alpha \right)^{(5/4)}}{2 \sqrt{\tan \alpha}}.
\label{n1q1:mubarB}
\eea

\noindent The solution for minimal mass can be achieved from the local minimum condition,
\bea
\frac{\partial \bar \mu_{B}}{\partial \tan\alpha}= \frac{5}{4} \eta \left( 1+ \tan^2\alpha \right)^{1/4} \sqrt{ \tan\alpha} - \eta \frac{\left(1+ \tan^2\alpha\right)^{5/4}}{4  \tan^{3/2}\alpha} =0,
\label{n1q1:dmuB}
\eea

\noindent to obtain the optimal angle (\ref{n1q1:alphastB}). Substituting it into (\ref{n1q1:muB}) yields (\ref{n1q1:mustB}).

\end{proof}

It is straightforward to show that the second variation of $\mu_B(\alpha)$ with respect to $\alpha$ is always positive, indicating that there is only one minimum described by (\ref{n1q1:alphastB}).

\begin{corollary} \label{theo:n1p1:min_buckling}

Consider a\emph{ substructure} bridge, with complexity $(n,p,q)=(1,1,0)$ (Fig. \ref{basic_module}c). The minimal mass design under yielding constraints and buckling constraints is given by the following aspect angle:

\bea
\beta^*_{B} = \arctan \left[ \frac{1}{6\eta} \left( \frac{1}{2^{(1/3)} \epsilon} + \frac{\epsilon}{2^{(2/3)}} - \frac{1}{\sqrt{2} } \right) \right],
\label{n1p1:betastB}
\eea

\noindent which corresponds to the following dimensionless minimal mass: 

\bea
\mu_{B}^{*} = \frac{1+\tan^2 \beta^*_{B}}{4 \tan \beta^*_{B}} + \frac{\eta}{2 \sqrt{2}}  \tan^2 \beta^*_{B}, 
\label{n1p1:mustB}
\eea

\noindent where:

\bea
\epsilon = \left[ 108 \sqrt{2} \eta^2 + \sqrt{23328 \eta^4 - 432 \eta^2} - \sqrt{2} \right]^{1/3}.
\label{n1p1:kappa}
\eea

\end{corollary}

\begin{proof}
\noindent The mass of the \emph{substructure} bridge can be obtained from (\ref{n1q1p1:muB_t1t2}) assuming $t_1 = t_2 = 0$:

\bea
\mu_B = \frac{1+\tan^2 \beta}{4 \tan \beta} + \frac{\eta}{2 \sqrt{2}} \tan^2 \beta.
\label{n1p1:muB_B2}
\eea

\noindent The above function has its minimum value $\mu_{B}^{*}$ for an optimal angle $\beta^*_{B}$ that can be computed from the equation

\bea
\frac{\partial \mu_{B}}{\partial\tan{\beta}}= \frac{1}{2} - \frac{1+\tan^2{\beta}}{4 \tan^2{\beta}} +\frac{\eta}{\sqrt{2}}\tan{\beta}= 0.
\label{n1p1:dmuB_B2}
\eea

\noindent After rearranging (\ref{n1p1:dmuB_B2}), the optimal angle $\beta$ can be computed solving the following equation:

\bea
4 \eta \tan^3{\beta} + \sqrt{2}\tan^2{\beta}- \sqrt{2} = 0.
\label{n1p1:muB_B2}
\eea
\end{proof}

It is straightforward to show that the second variation of $\mu_B(\beta)$ with respect to $\beta$ is always positive, indicating a unique global optimal value of (\ref{n1p1:betastB}). Fig \ref{fig:n1pq1_buckling_graph} plots the mass versus the angle $\beta$ and $\alpha$, yielding the minimum at the values given by (\ref{n1q1:alphastB}) and (\ref{n1p1:betastB}). We must verify if buckling is indeed the mode of failure in the designs of this section.

\begin{figure}[hb]
\unitlength1cm
\begin{picture}(10,5.5)
\put(0.25,1.3){\psfig{figure=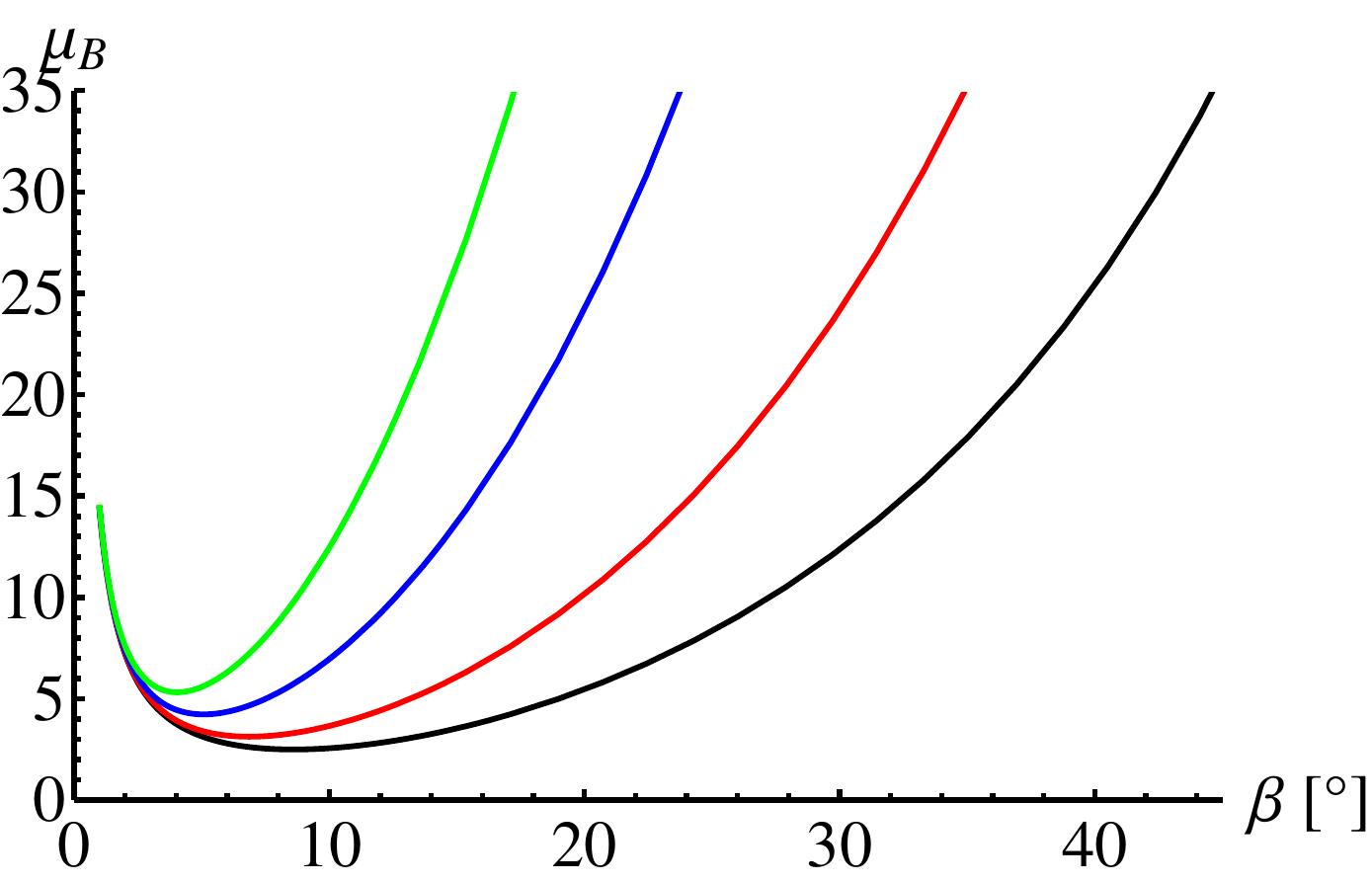,height=4.5cm}}
\put(8,1.3){\psfig{figure=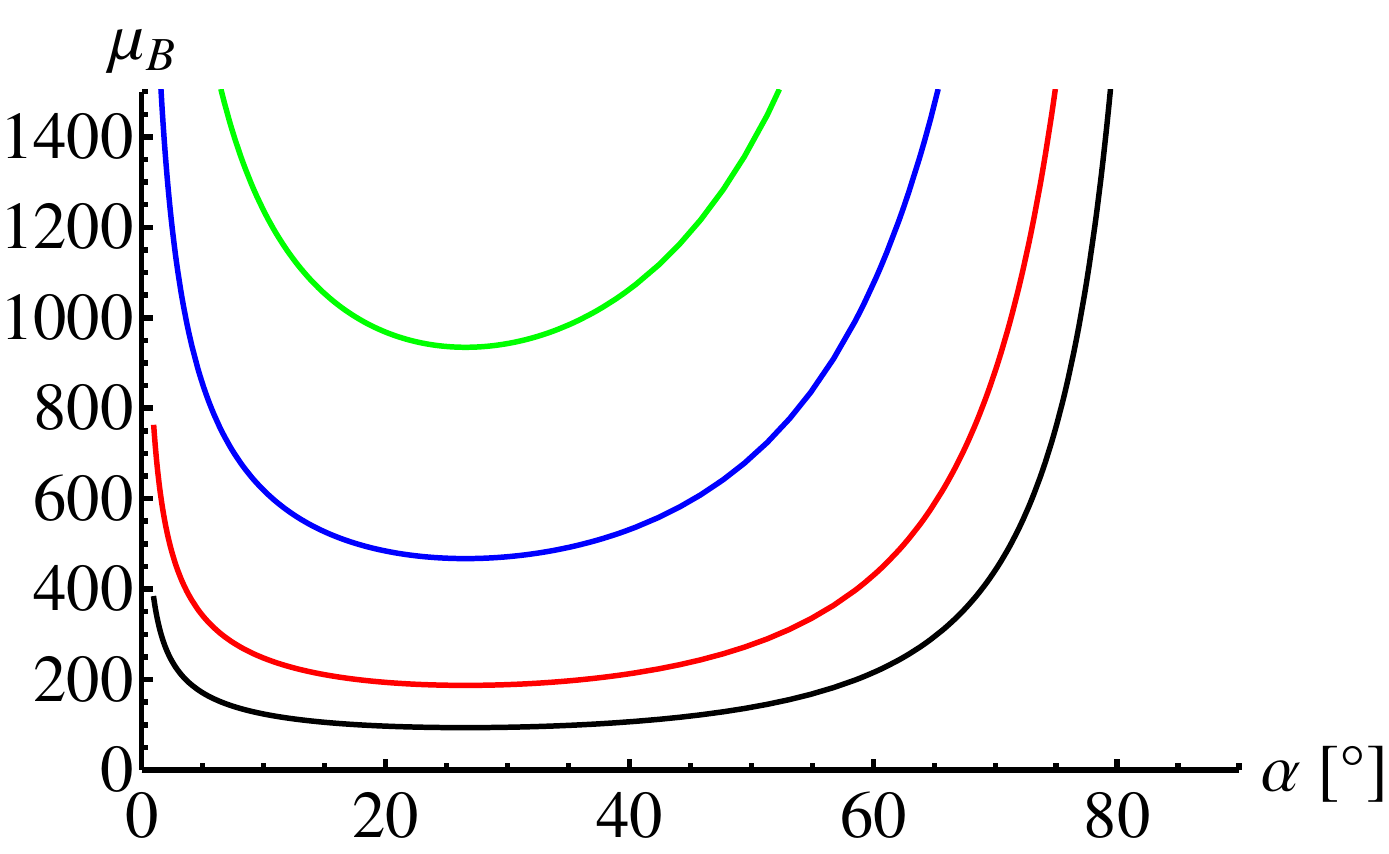,height=4.5cm}}
\put(1.8,-0.2){\psfig{figure=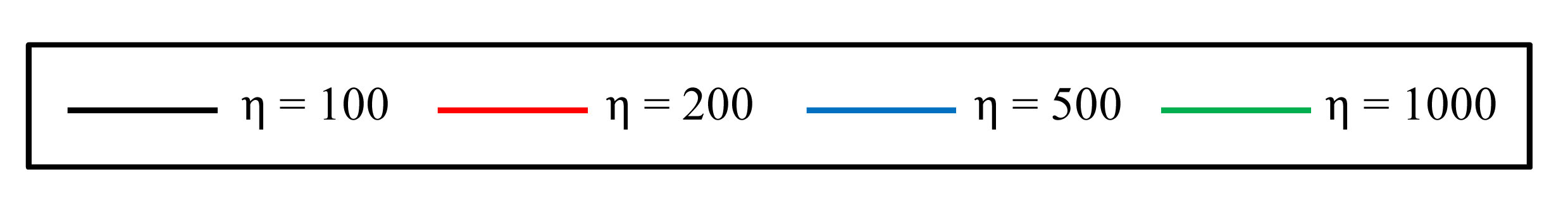,height=1.5cm}}
\end{picture}
\caption{Dimensionless masses of the substructure (left) and superstructure (right) under buckling constraints for different values of the aspect angles (respectively $\beta$ or $\alpha$) and different values of the parameter $\eta$.}
\label{fig:n1pq1_buckling_graph}
\end{figure}

\begin{corollary} \label{theo:buckling_sup_sub}

Suppose buckling constraints are considered in both the \emph{superstructure} and \emph{substructure} bridge designs. Then buckling is indeed the mode of failure if the following inequalities hold:

\bea
\frac{F}{L^2}< \tan \alpha \sqrt{1 + \tan^2 \alpha} \left( \frac{4 \sigma_b^2}{\pi E_b} \right), \ \ \ &if:& \bar \eta_{\alpha \beta} < 1, \label{chi_1_B} \\ 
\frac{F}{L^2}< \frac{\tan^2 \beta}{2} \left( \frac{4 \sigma_b^2}{\pi E_b} \right), \ \ \ &if:& \bar \eta_{\alpha \beta} > 1, 
\label{chi_2_B}
\eea

\noindent where:

\bea
\bar \eta_{\alpha \beta} = \frac{2 \tan \alpha \sqrt{1 + \tan^2 \alpha}}{\tan^2 \beta}.
\label{etabar_alpha_beta}
\eea

\noindent  In addition, if the following inequality holds:

\bea
\eta > \eta_{\alpha \beta} = \frac{- \left( \tan \alpha \right)^{(3/2)} \tan \beta + \left( 1 + \tan^2 \beta \right) \sqrt{\tan \alpha}}{2 \left( 1 + \tan^2 \beta \right)^{(5/4)} \tan \beta - \sqrt{2} \tan^3 \beta \sqrt{\tan \alpha}},
\label{eta_alpha_beta}
\eea
then the minimal mass of the \emph{substructure} bridge is less than the minimal mass of the \emph{superstructure} bridge. (The minimal mass of the \emph{nominal} bridge reduces to \emph{substructure} only. If $\eta=\eta_{\alpha \beta}$, (\ref{chi_1_B}) or (\ref{chi_2_B}) hold, then the minimal mass of the \emph{substructure} is equal to the minimal mass of the \emph{superstructure}. (The minimal mass of the \emph{nominal} bridge reduces to either \emph{superstructure} or \emph{substructure} only). If $\eta<\eta_{\alpha \beta}$, and (\ref{chi_1_B}) or (\ref{chi_2_B}) hold, then the minimal mass of the \emph{superstructure} is less than the minimal mass of the \emph{substructure}. (The minimal mass bridge is \emph{superstructure} only).

\end{corollary}

\begin{proof}

Under buckling constraints, if the design has the property ${f_{b,i}}/{b_i^2}< {4 \sigma_b^2}/({\pi E_b})$, then this guarantees that buckling is the mode of failure in bar $b_i$, and the yielding constraints are also satisfied (see lemma \ref{theo:buckling_design_2}). For the \emph{superstructure}, assuming the force $f_1$  (\ref{forces_t1_t3}) and the length $b_1$ (\ref{n1q1p1:topology}), then lemma \ref{theo:buckling_design_2} reduces to (\ref{chi_1_B}). Similarly, for the \emph{substructure}, assuming the force $f_2$ (\ref{forces_t1_t3}) and the length $b_2$ (\ref{n1q1p1:topology}), then lemma (\ref{theo:buckling_design_2}) reduces to (\ref{chi_2_B}). Buckling is the mode of failure of \emph{superstructure} and \emph{substructure} designs if both (\ref{chi_1_B}) and (\ref{chi_2_B}) hold or, equivalently, if the following holds:

\bea
\frac{F}{L^2}<min \left[ \tan \alpha \sqrt{1 + \tan^2 \alpha} , \frac{\tan^2 \beta}{2} \right] \left( \frac{4 \sigma_b^2}{\pi E_b} \right).
\label{FL_chi_max_B}
\eea

\noindent From the inequality $\bar \eta_{\alpha \beta} > 1$ we obtain conditions (\ref{chi_1_B}) and (\ref{chi_2_B}). 

The mass of the \emph{substructure} is shown to be less then the mass of the \emph{superstructure} if $\eta>\eta_{\alpha \beta}$, a result that follows by taking the ratio between the mass of the \emph{superstructure }(\ref{n1q1:muB}) and the mass of the \emph{substructure} (\ref{n1p1:muB_B2}). 

\end{proof}

The left contour plot in Fig. \ref{fig:contour_etabar_eta} shows values of the function $\bar \eta_{\alpha \beta}$ for any angles $\alpha$ and $\beta$, indicating the range of $\alpha$ and $\beta$ for which $\bar \eta_{\alpha \beta}>1$, which in turn chooses the appropriate condition (\ref{chi_1_B}) or (\ref{chi_2_B}).
The trend of the function $\eta_{\alpha \beta}$ is shown in the right contour plot of Fig. \ref{fig:contour_etabar_eta}. The physical parameter $\eta$ is a positive number and Fig. \ref{fig:contour_etabar_eta} show the region for which the quantity $\eta_{\alpha \beta}$ is a negative number. We have shown earlier (\ref{n1q1:alphastB}) that the approximated $\alpha=26.56$ degrees. Furthermore Fig \ref{fig:n1pq1_buckling_graph} illustrates that $\alpha=26.56$ degrees is very close to the actual minimum over a very large range of the physical parameter $\eta$. Therefore, from  the right plot in Fig \ref{fig:contour_etabar_eta} any $\alpha$ in the range of the optimal value (~$26$ degrees) yields $\eta>\eta(\alpha,\beta)$. Hence, the \emph{substructure} bridge has the minimal mass.

\begin{figure}[hb]
\unitlength1cm
\begin{picture}(8,5.75)
\put(0.5,0){\psfig{figure=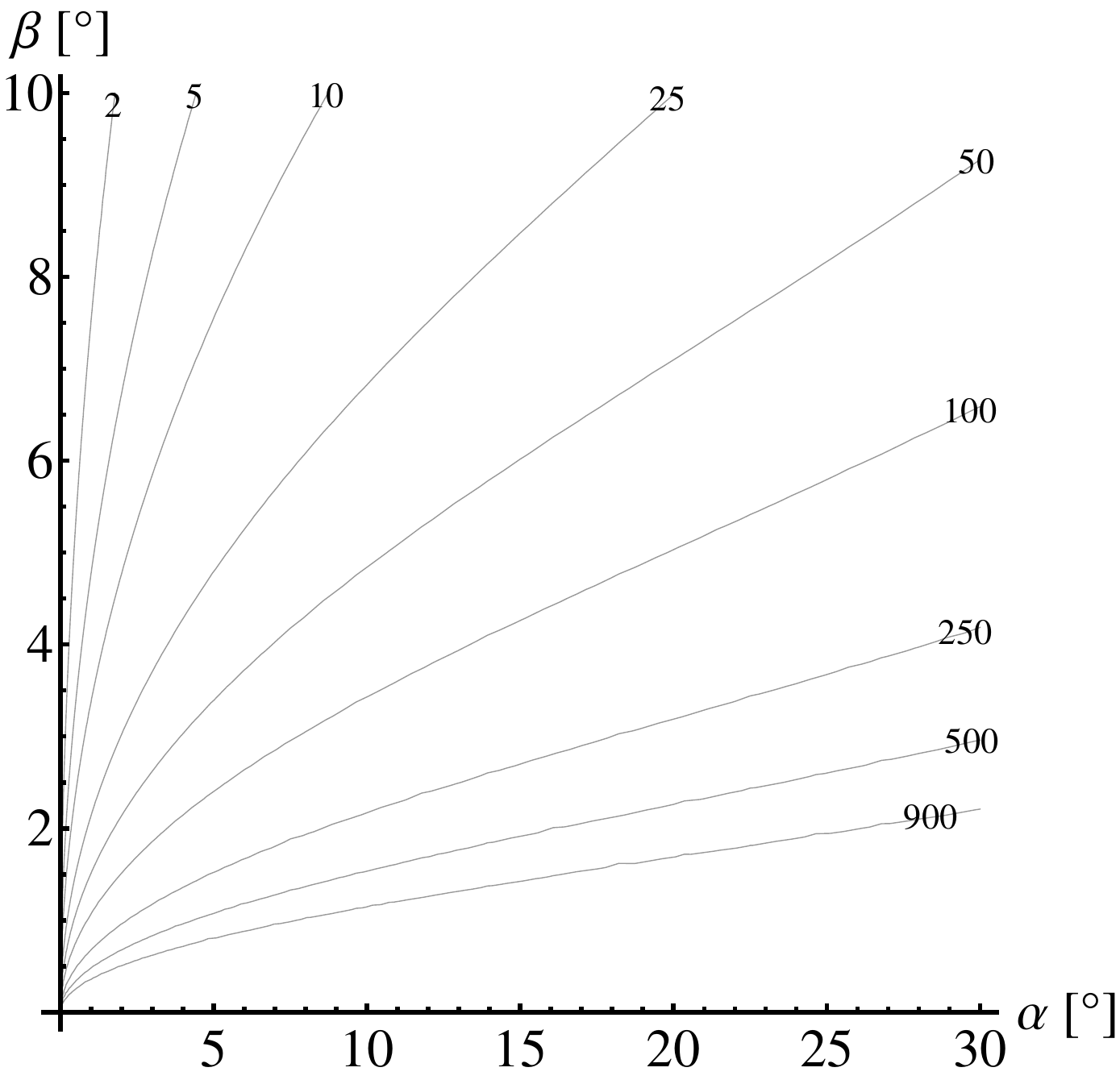,height=6cm}}
\put(8.5,0){\psfig{figure=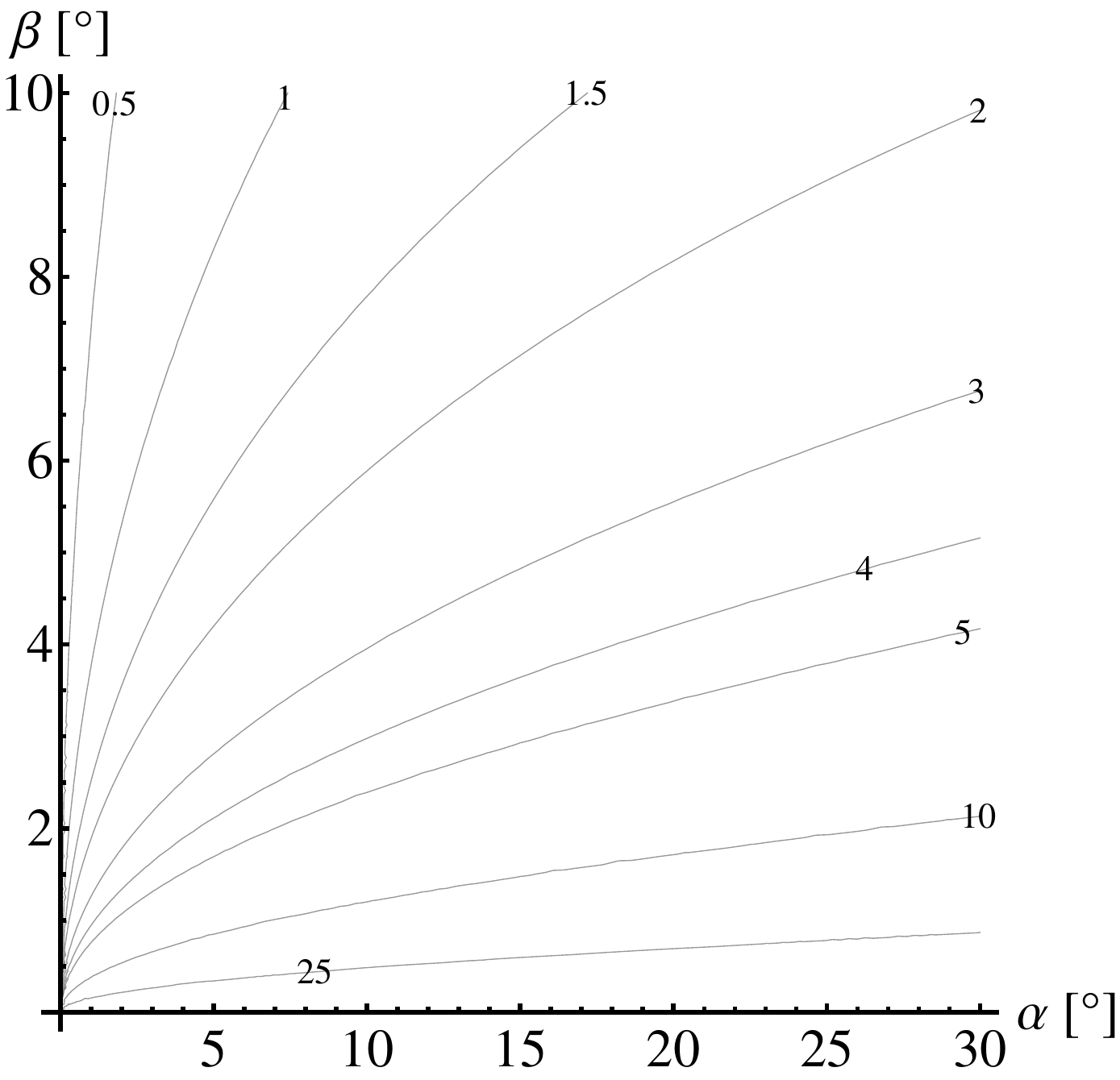,height=6cm}}
\end{picture}
\caption{Contour plots of the functions $\bar \eta_{\alpha \beta}$, (left, Eq. \ref{etabar_alpha_beta}) and $\eta_{\alpha \beta}$, (right, Eq. \ref{eta_alpha_beta}) for different values of the aspect angles $\alpha$ and $\beta$}
\label{fig:contour_etabar_eta}
\end{figure}

\section{Mass of Bridges of Complexity ($n,p,q$) = ($1,p,q$), Under Yield and Buckling Constraints} \label{Sec:n1qp}

Now we consider more complex structures by increasing $p$, $q$. This section finds the minimal mass of \emph{substructure}, and \emph{superstructure} bridges with complexity ($n,p,q$) = ($1,p,q$), for any $p$ and $q$ greater then $1$.

\begin{figure}[hb] 
\unitlength1cm
\begin{picture}(10,5)
\put(0.5,0){\psfig{figure=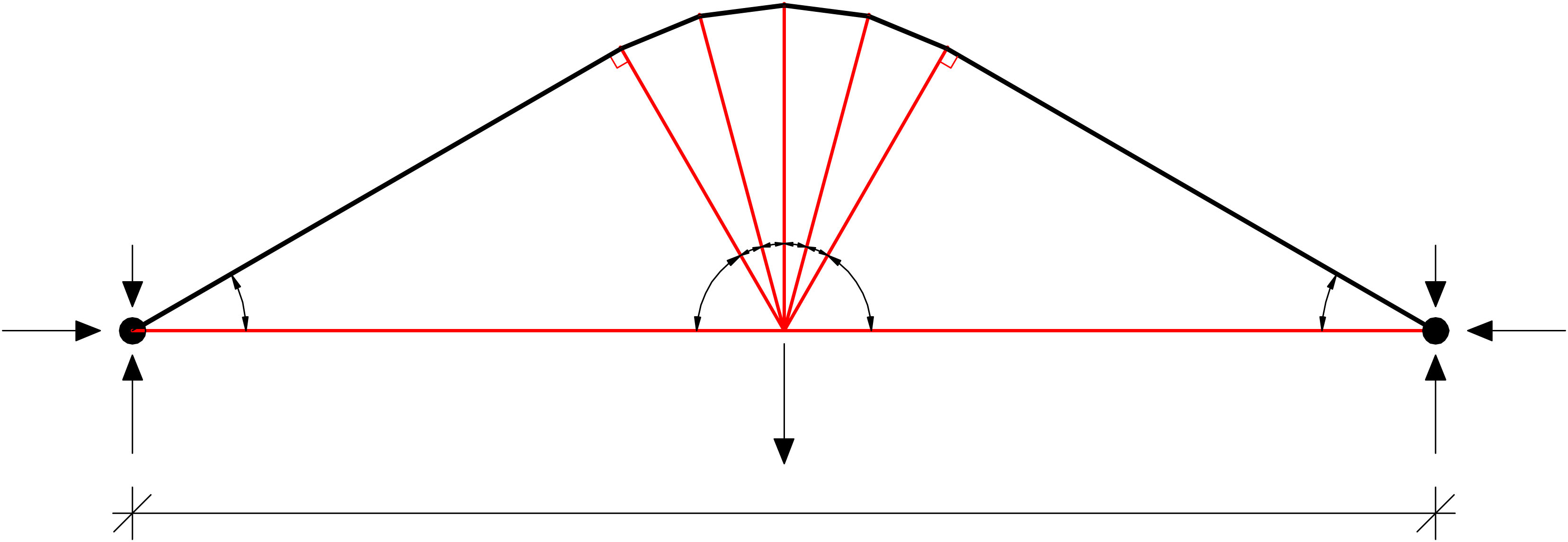,height=5cm}}
\put(7.5,-0.3){{$\mbox{$L$}$}}
\put(7.8,1.2){{$\mbox{$\frac{F}{2}$}$}}
\put(1.85,2.5){{$\mbox{$\frac{F}{4}$}$}}
\put(13.25,2.5){{$\mbox{$\frac{F}{4}$}$}}
\put(1.85,1.25){{$\mbox{$\frac{F}{2}$}$}}
\put(13.25,1.25){{$\mbox{$\frac{F}{2}$}$}}
\put(14.2,2.2){{$\mbox{$w_x$}$}}
\put(0.7,2.2){{$\mbox{$w_x$}$}}
\put(12.3,2.2){{$\mbox{$\alpha$}$}}
\put(2.9,2.2){{$\mbox{$\alpha$}$}}
\put(5.9,2.2){{$\mbox{$\frac{\pi}{2}-\alpha$}$}}
\put(8.7,2.2){{$\mbox{$\frac{\pi}{2}-\alpha$}$}}
\put(8.1,2.95){{$\mbox{$\gamma$}$}}
\put(7.75,2.95){{$\mbox{$\gamma$}$}}
\put(7.15,2.95){{$\mbox{$\gamma$}$}}
\put(7.45,2.95){{$\mbox{$\gamma$}$}}
\put(4.2,1.5){{$\mbox{$t_0,s_0$}$}}
\put(10.2,1.5){{$\mbox{$t_0,s_0$}$}}
\put(3.3,3.5){{$\mbox{$f_1,b_1$}$}}
\put(11.3,3.5){{$\mbox{$f_1,b_1$}$}}
\put(6.2,5){{$\mbox{$f_2,b_2$}$}}
\put(8.3,5){{$\mbox{$f_2,b_2$}$}}
\put(5.9,3.25){{$\mbox{$t_1,s_1$}$}}
\put(8.7,3.25){{$\mbox{$t_1,s_1$}$}}
\put(6.75,4.25){{$\mbox{$t_2,s_2$}$}}
\put(8,4.25){{$\mbox{$t_2,s_2$}$}}
\end{picture}
\caption{Notations for forces and lengths of bars and cables for a \emph{superstructure} with complexity $n = 1$ and $q > 1$.}
\label{theorems_sup_n1p>1}
\end{figure}
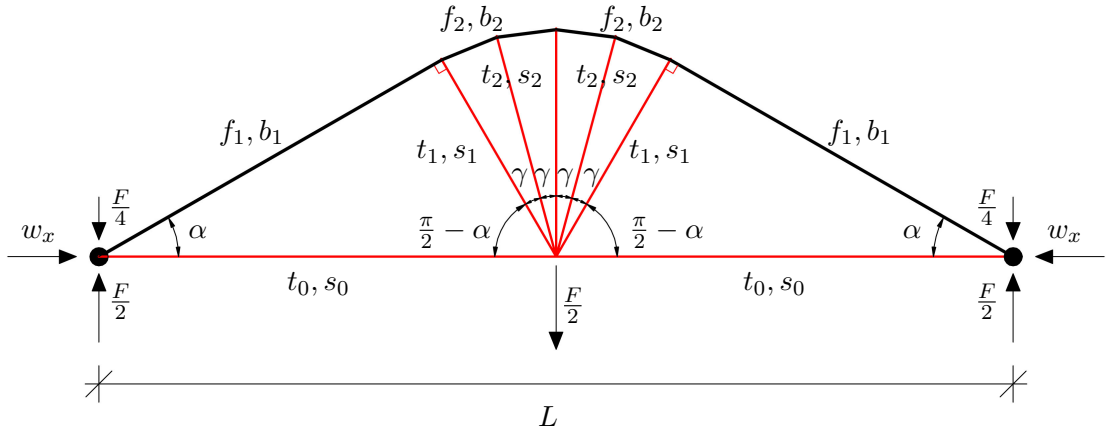

\begin{figure}[hbt]
\unitlength1cm
\scalebox{0.8}{
\begin{picture}(14,17.5)
\put(0.5,-0.5){\psfig{figure=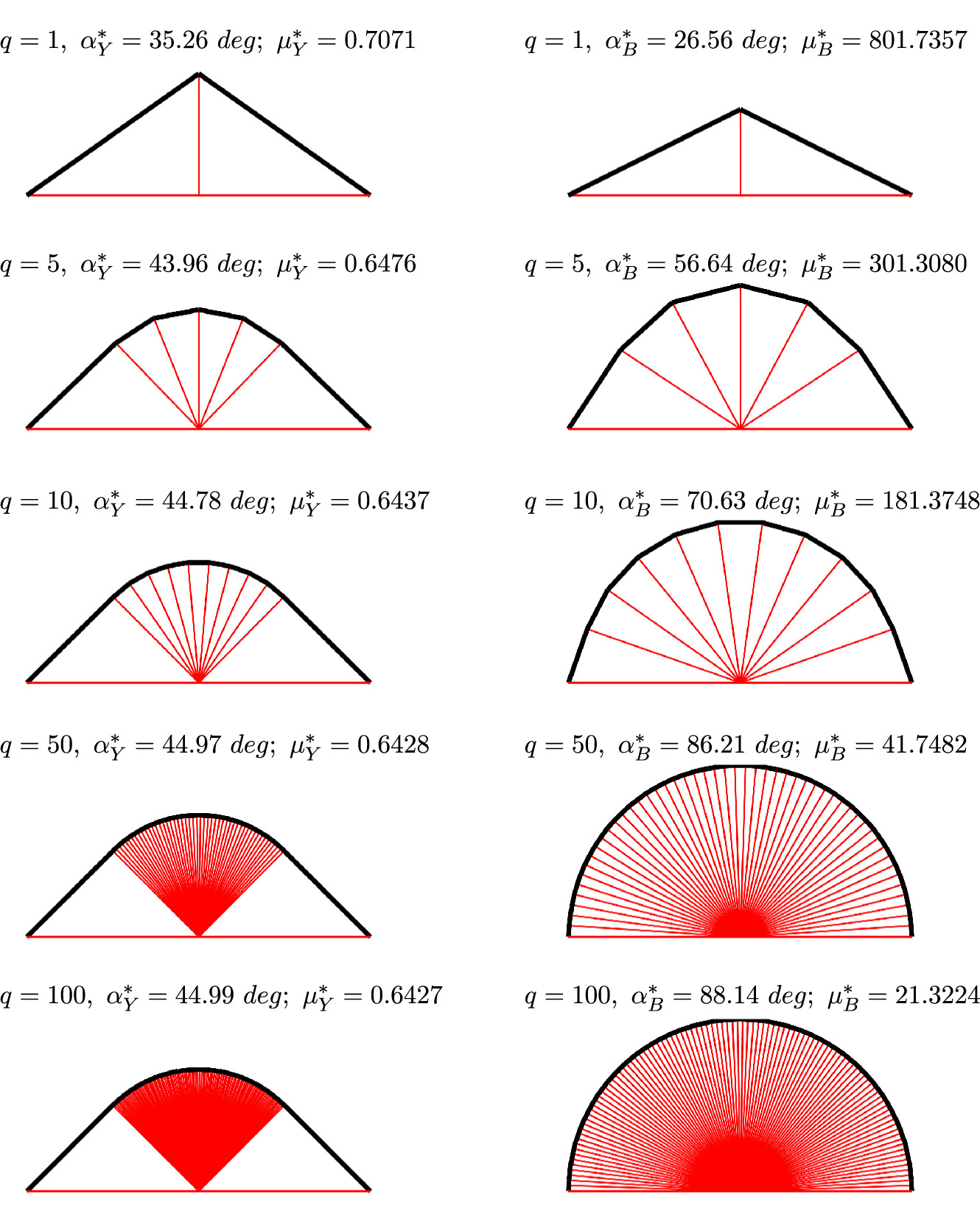,height=17.5cm}}
\end{picture}}
\caption{Optimal topologies of \emph{superstructure} bridges with complexity $(n,p,q)=(1,0,q \rightarrow \infty)$ under yielding constraints (left) and buckling constraints (right) for different $q$, (steel for bars and cables, $F = 1 \ N$, $L = 1 \ m$).}
\label{fig:PANEL1}
\end{figure}

\begin{figure}[hb] 
\unitlength1cm
\begin{picture}(10,5.35)
\put(0.5,1.55){\psfig{figure=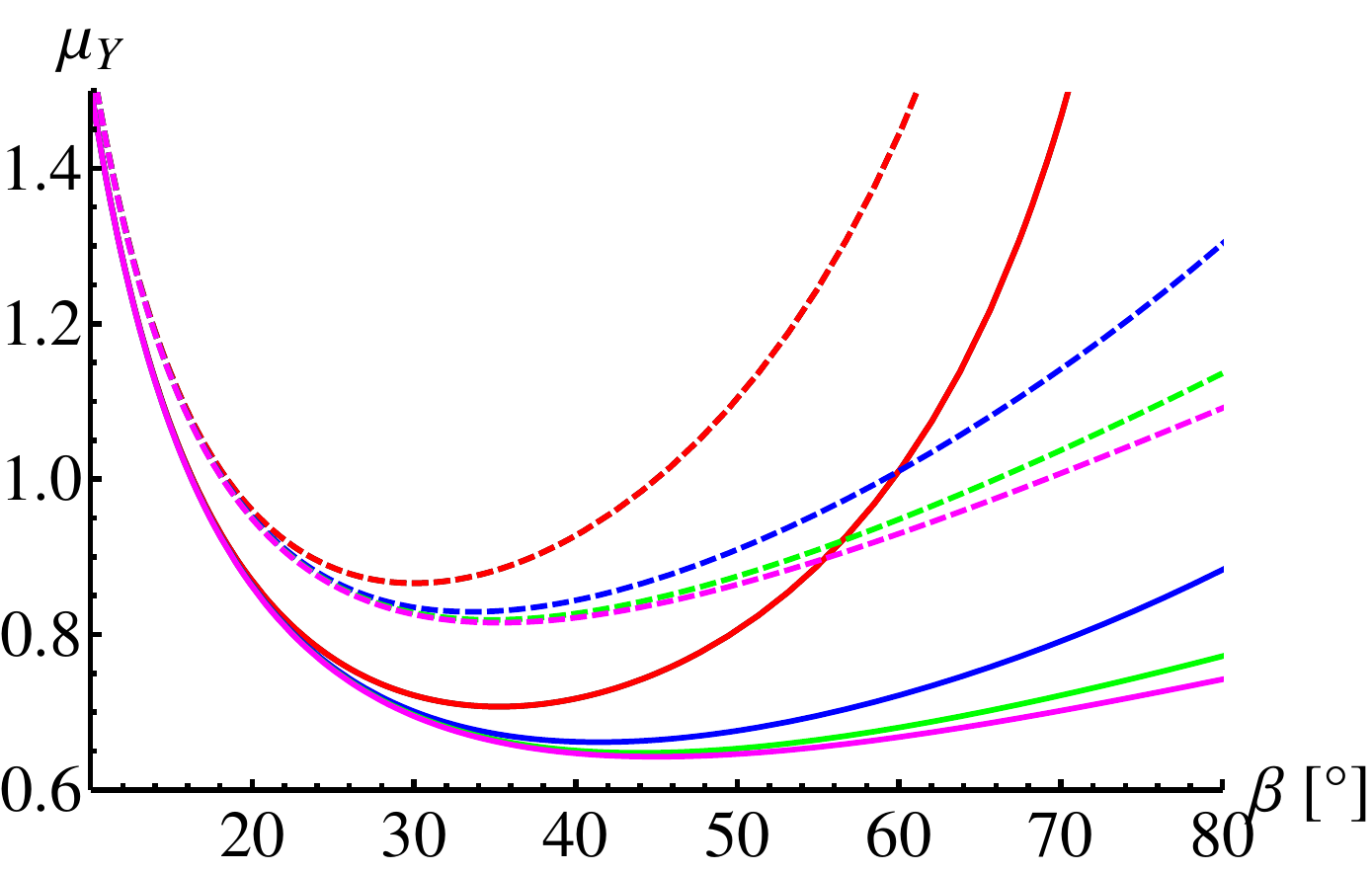,height=4.5cm}}
\put(8,1.55){\psfig{figure=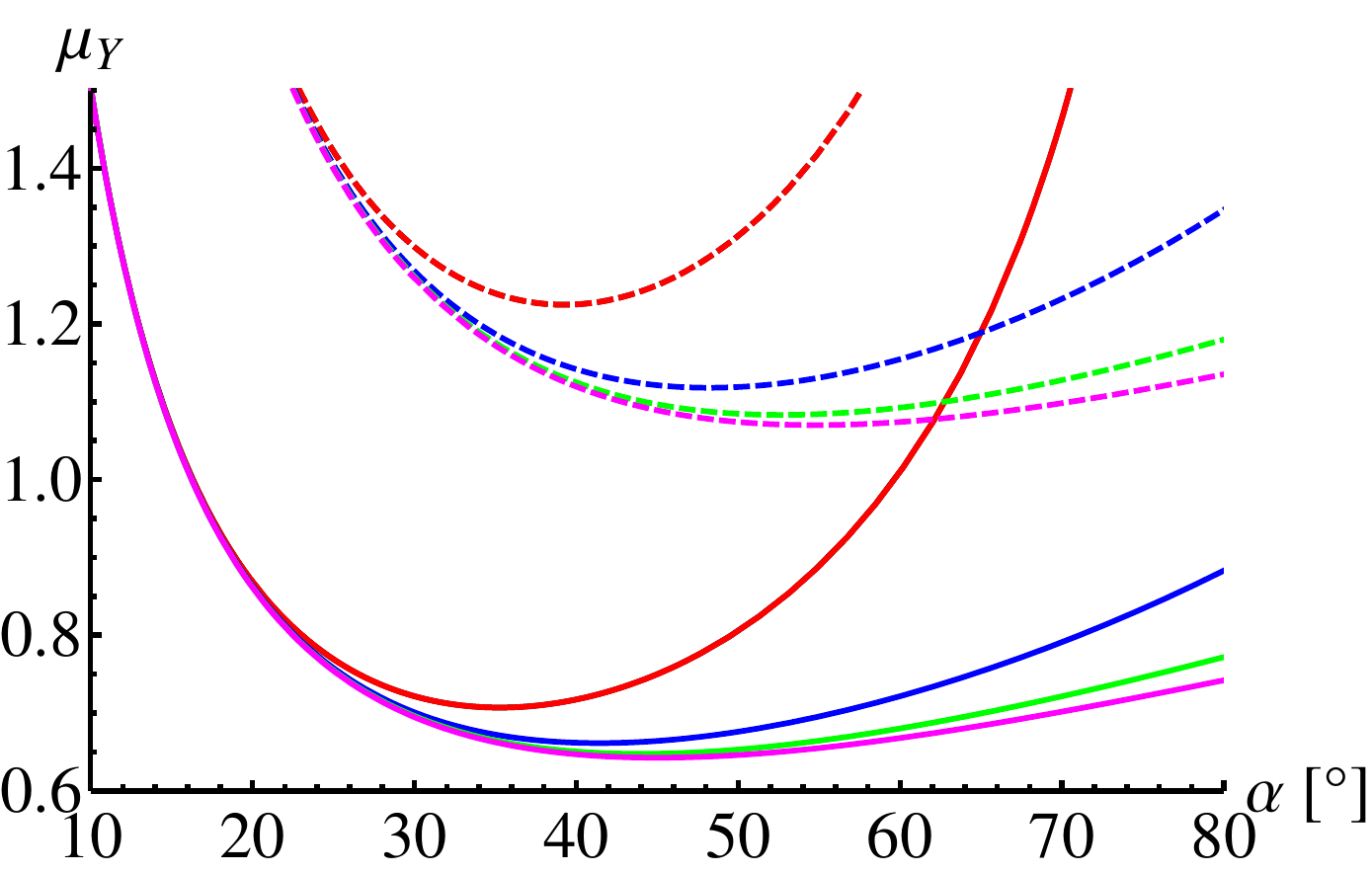,height=4.5cm}}
\put(2.5,-0.25){\psfig{figure=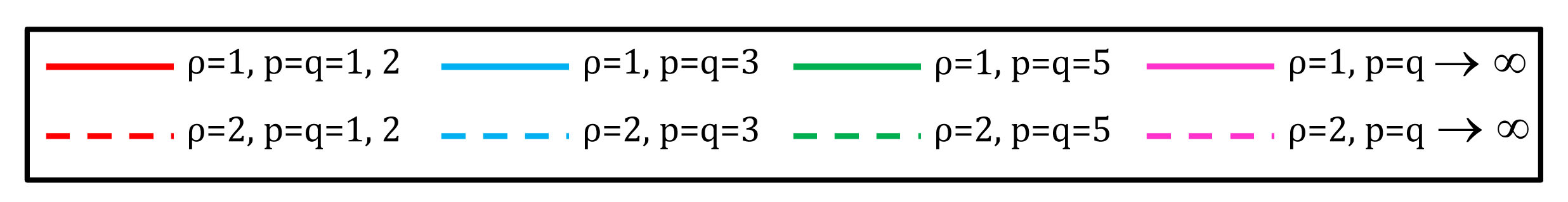,height=1.5cm}}
\end{picture}
\caption{Mass curves under yielding constraints of \emph{substructures} (left) and \emph{superstructures} (right) vs. aspect angle $\beta$ (left) and $\alpha$ (right) for different complexity $p$ (left) and $q$ (right), ($F = 1 \ N$, $L = 1 \ m$).}
\label{plot_sub_n1p>1}
\end{figure}

\begin{figure}[hb] 
\unitlength1cm
\begin{picture}(10,5.4)
\put(0.5,1.5){\psfig{figure=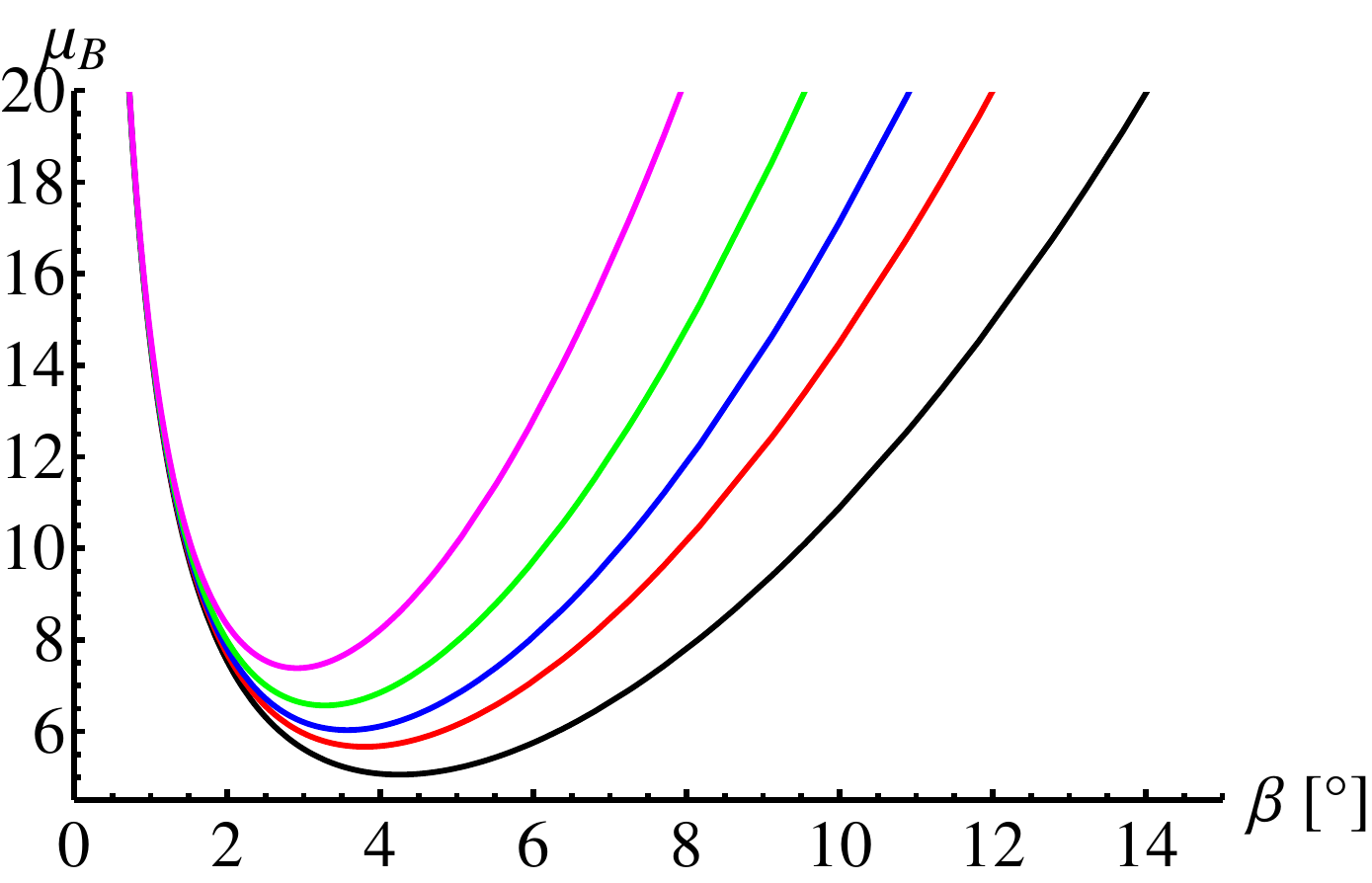,height=4.5cm}}
\put(8,1.5){\psfig{figure=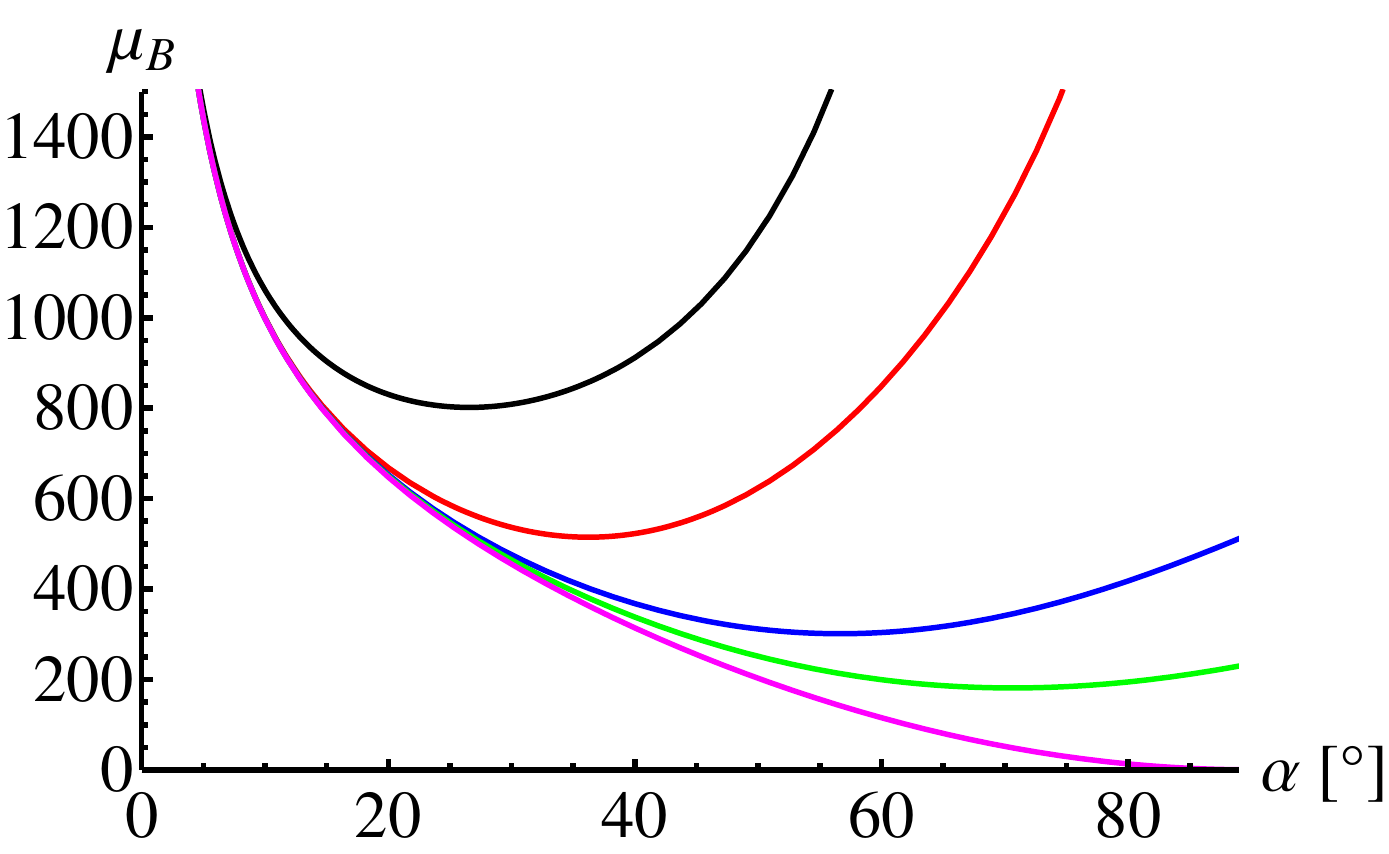,height=4.5cm}}
\put(2,-0.25){\psfig{figure=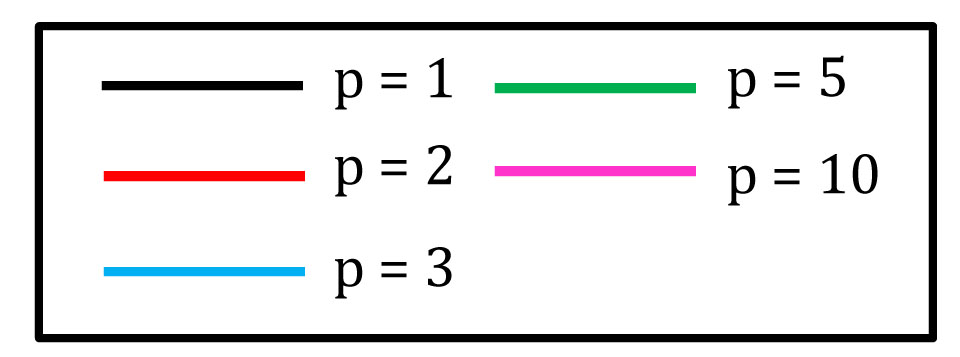,height=1.5cm}}
\put(9.5,-0.25){\psfig{figure=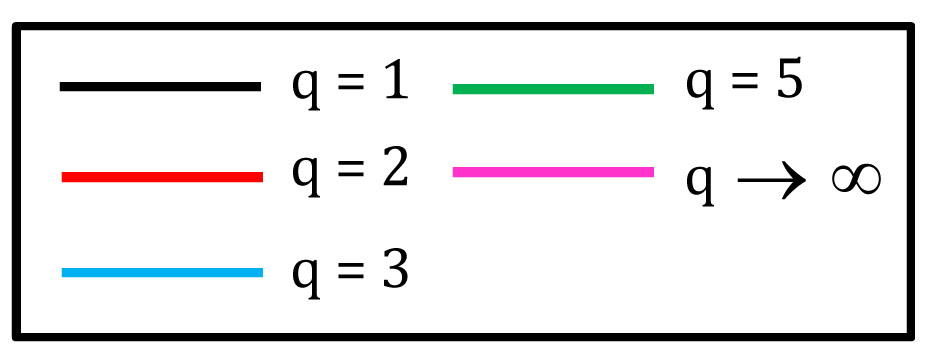,height=1.5cm}}
\end{picture}
\caption{Mass curves under buckling constraints of \emph{substructures} (left) and \emph{superstructures} (right) vs. aspect angle $\beta$ (left) and $\alpha$ (right) for different complexity $p$ (left) and $q$ (right), (steel bars and cables, $F = 1 \ N$, $L = 1 \ m$).}
\label{plot_sup_n1p>1}
\end{figure}

\begin{figure}[hb] 
\unitlength1cm
\begin{picture}(10,5)
\put(0.5,0){\psfig{figure=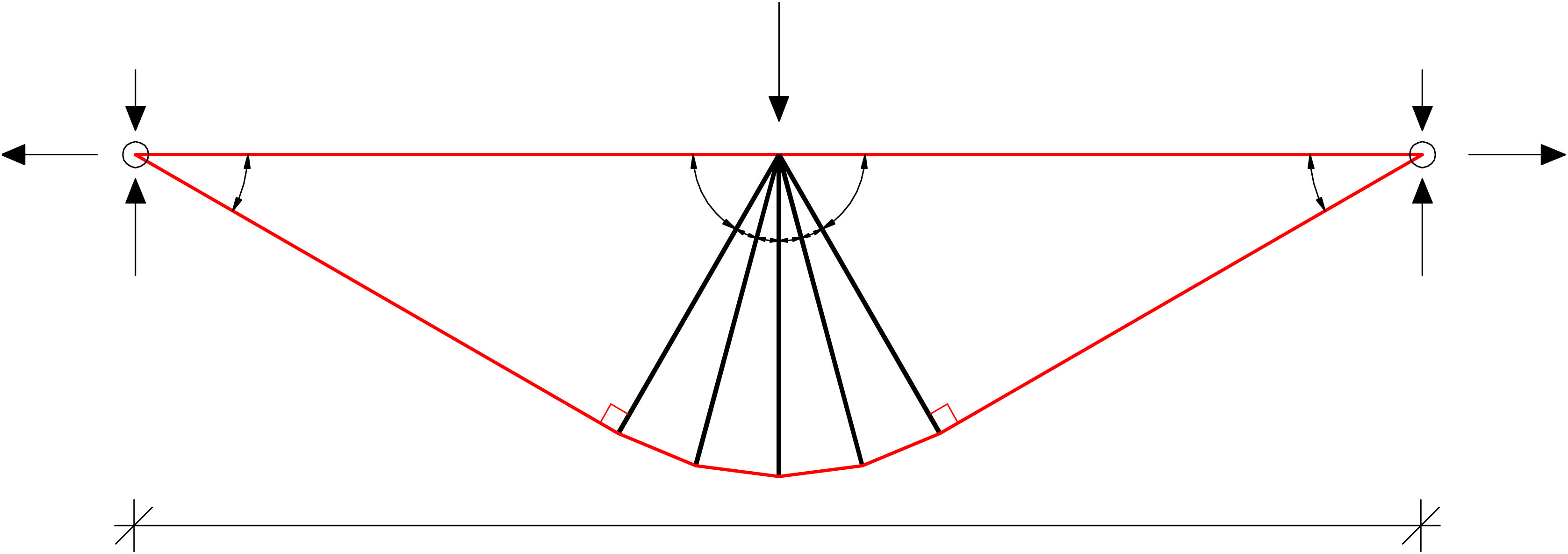,height=5cm}}
\put(7.5,-0.3){{$\mbox{$L$}$}}
\put(7.65,4.5){{$\mbox{$\frac{F}{2}$}$}}
\put(1.85,4){{$\mbox{$\frac{F}{4}$}$}}
\put(12.85,4){{$\mbox{$\frac{F}{4}$}$}}
\put(1.85,2.75){{$\mbox{$\frac{F}{2}$}$}}
\put(12.85,2.75){{$\mbox{$\frac{F}{2}$}$}}
\put(14,3.85){{$\mbox{$w_x$}$}}
\put(0.7,3.85){{$\mbox{$w_x$}$}}
\put(12,3.2){{$\mbox{$\beta$}$}}
\put(2.8,3.2){{$\mbox{$\beta$}$}}
\put(5.7,3.2){{$\mbox{$\frac{\pi}{2}-\beta$}$}}
\put(8.5,3.2){{$\mbox{$\frac{\pi}{2}-\beta$}$}}
\put(7.9,2.4){{$\mbox{$\gamma$}$}}
\put(7.55,2.4){{$\mbox{$\gamma$}$}}
\put(6.95,2.4){{$\mbox{$\gamma$}$}}
\put(7.25,2.4){{$\mbox{$\gamma$}$}}
\put(4.2,3.8){{$\mbox{$t_0,s_0$}$}}
\put(10.2,3.8){{$\mbox{$t_0,s_0$}$}}
\put(3.2,2){{$\mbox{$t_1,s_1$}$}}
\put(11.2,2){{$\mbox{$t_1,s_1$}$}}
\put(6.2,0.5){{$\mbox{$t_2,s_2$}$}}
\put(8.3,0.5){{$\mbox{$t_2,s_2$}$}}
\put(5.6,2){{$\mbox{$f_1,b_1$}$}}
\put(8.6,2){{$\mbox{$f_1,b_1$}$}}
\put(6.5,1.25){{$\mbox{$f_2,b_2$}$}}
\put(7.7,1.25){{$\mbox{$f_2,b_2$}$}}
\end{picture}
\caption{Notations for forces and lengths of bars and cables for a \emph{substructure} with complexity $n = 1$ and $p > 1$.}
\label{theorems_sub_n1p>1}
\end{figure}

\subsection{\emph{Superstructure} Bridge with Complexity $(n,p,q)=(1,0,q>1)$} \label{sup_n1p>1}

\bigskip

\noindent Refer to Fig. \ref{theorems_sup_n1p>1} for the notation. The angle between the bars is:

\beq
\label{gamma_sup}
\gamma= \frac{2 \alpha}{q - 1}.
\eeq

\noindent The lengths of the bars and cables are:

\beq
s_0 = \frac{L}{2}, \ \ \
s_1 = s_2 = \frac{L}{2} \sin \alpha, \ \ \
b_1 = \frac{L}{2} \cos \alpha, \ \ \
b_2 = L \sin \alpha \sin \left( \frac{\alpha}{q - 1} \right). 
\label{geom_n1q>1}
\eeq

\noindent From the equilibrium equations, we obtain the following relations for the forces:

\bea
t_2 = \frac{F}{2 \left[ \cos \alpha + \sin \left( \frac{ \alpha \left( q - 2 \right)}{q-1} \right) / \sin \left( \frac{\alpha}{q-1} \right) \right] }, \ \ \
t_1 = \frac{t_2}{2}, \\
f_2 = \frac{t_2}{2 \sin \left( \frac{\alpha}{q-1} \right)}, \ \ \ 
f_1 = {f_2}{\cos \left( \frac{\alpha}{q-1} \right)}.
\label{forces_n1q>1}
\eea

\begin{theorem} \label{theo:min_yielding_n1q>1}

Consider a \emph{superstructure} bridge, of total span $L$, topology defined by (\ref{geom_n1q>1}), with complexity ($n=1,q>1$), Fig. \ref{theorems_sup_n1p>1}. At the yield condition under a vertical load $F$ the dimensionless total mass is: 

\bea
\mu_{Y} \left(\alpha,q \right) = \frac{t_0}{F} + \frac{ \left( q - 1 \right) \sin \alpha }{4 \left[ \cos \alpha + \sin \left( \frac{\alpha \left( q-2\right)}{q-1} \right) / \sin \left( \frac{\alpha}{q-1} \right) \right] } + \nonumber \\  \frac{\rho}{4} \frac{ \left( q-1 \right) \sin \alpha \sin \left( \frac{\alpha}{q-1} \right) + \cos \alpha \cos \left( \frac{\alpha}{q-1} \right) }{ \sin \left( \frac{\alpha}{q-1} \right) \cos \alpha + \sin \left( \frac{\alpha \left( q-2\right)}{q-1} \right) }.
\label{n1q_muY}
\eea

\end{theorem}

\begin{proof}

\noindent The total mass of the cables is:

\bea
m_s = \frac{\rho_s}{\sigma_s} \sum_{i=1}^{n_s} t_i s_i = \frac{\rho_s}{\sigma_s} \left( 2 t_0 s_0 + 2 t_1 s_1 + \left( p - 2 \right) t_2 s_2 \right).
\eea

\noindent Substituting (\ref{geom_n1q>1}) and (\ref{forces_n1q>1}) into $m_s$ we get:

\bea
m_s = \frac{\rho_s}{\sigma_s} \left( t_0 L + \frac{F L}{4} \frac{ \left( q - 1 \right) \sin \alpha }{ \left( \cos \alpha + \sin \left( \frac{\alpha \left( q-2 \right)}{q-1} \right) / \sin \left( \frac{\alpha}{q-1} \right) \right) } \right).
\eea

\noindent The total mass of bars is:

\bea
m_b = \frac{\rho_b}{\sigma_b} \sum_{i=1}^{n_b} f_i b_i = \frac{\rho_b}{\sigma_b} \left( 2 f_1 b_1 + \left( p-1 \right) f_2 b_2 \right).
\eea

\noindent Substituting (\ref{geom_n1q>1}) and (\ref{forces_n1q>1}) into $m_b$ we get:

\bea
m_{b} = \frac{\rho_b F L}{4 \sigma_b} \frac{ \left( q-1 \right) \sin \alpha \sin \left( \frac{\alpha}{q-1} \right) + \cos \alpha \cos \left( \frac{\alpha}{q-1} \right) }{ \sin \left( \frac{\alpha}{q-1} \right) \cos \alpha + \sin \left( \frac{\alpha \left( q-2\right)}{q-1} \right) }.
\eea

\noindent Normalizing $m_s$ and $m_b$ and summing we get (\ref{n1q_muY}). 

\end{proof}

\begin{corollary}

The minimal mass in (\ref{n1q_muY}) is achieved at infinite complexity $q \rightarrow \infty$ and $t_0 = 0$. Then the minimal mass at yielding for a \emph{superstructure} bridge is:

\bea
\mu_{Y}^* (\alpha^*_Y, q^*) = \frac{1}{4} \left[ \left( 1 + \rho \right) \arctan \sqrt{\rho} + \sqrt{\rho} \right],
\label{n1q_mustY}
\eea

\noindent where $q^* \rightarrow \infty$ and the optimal angle $\alpha^*_Y$ is:

\bea
\alpha^*_Y = \arctan \sqrt{\rho}.
\label{n1q_alphastY}
\eea

\end{corollary}

The left side of Fig. \ref{fig:PANEL1} illustrates \emph{superstructure} bridges as $q \rightarrow \infty$, where masses are given for any $q$ by (\ref{n1q_muY}).

\begin{proof}

\noindent Substitute $q \rightarrow \infty$ into Eq. (\ref{n1q_muY}) to obtain:

\bea
\mu_{Y}^* (\alpha, q^* \rightarrow \infty)= \frac{\alpha}{4} \left( 1 + \rho \right) + \frac{\rho}{4 \tan \alpha}.
\label{n1q_mustYq}
\eea

The value of $\alpha$ that minimizes (\ref{n1q_mustYq}) is (\ref{n1q_alphastY}). See Fig. \ref{plot_sub_n1p>1} to see how mass (\ref{n1q_muY}) varies with $q$ and $\alpha$. The optimal $q^*$ is deduced from the plot of Fig. \ref{plot_sub_n1p>1} and the optimal angle is computed analytically in Eq. (\ref{n1q_alphastY}).   

\end{proof}

\begin{theorem} \label{theo:min_buckling_n1q>1}

Consider a \emph{superstructure} bridge with topology (\ref{geom_n1q>1}), and complexity $(n,p,q)=(1,0,q>1)$, see Fig. \ref{theorems_sup_n1p>1}. At the buckling condition the dimensionless total mass is:

\bea
\mu_{B} \left(\alpha,q \right) = \frac{t_0}{F} + \frac{ \left( q - 1 \right) \sin \alpha }{4 \left[ \cos \alpha + \sin \left( \frac{\alpha \left( q-2\right)}{q-1} \right) / \sin \left( \frac{\alpha}{q-1} \right) \right] } + \nonumber \\  \eta \left[ \frac{\cos^2 \alpha \sqrt{\cos \left( \frac{\alpha}{q-1} \right) } + 2 \left( q-1 \right) \sin^2 \alpha \sin^2 \left( \frac{\alpha}{q-1} \right) }{2 \sqrt{ \sin \left( \frac{\alpha}{p-1} \right) \cos \alpha + \sin \left( \frac{\alpha \left( q-2\right)}{q-1} \right) } } \right].
\label{n1q_muB}
\eea

\end{theorem}

\begin{proof}

\noindent The total mass of the cables has been already computed in the proof of Theorem \ref{theo:min_yielding_n1q>1}.

\noindent The total mass of bars is:

\bea
m_b = \sum_{i=1}^{n_b} 2 \rho_b b_i^2 \sqrt{\frac{f_i}{\pi E_b}} = 4 \rho_b b_1^2 \sqrt{\frac{f_1}{\pi E_b}} + 2 \left( p-1 \right) \rho_b b_2^2 \sqrt{\frac{f_2}{\pi E_b}}.
\eea

\noindent Substituting (\ref{geom_n1q>1}) and (\ref{forces_n1q>1}) into $m_b$ we get:

\bea
m_{b} = \frac{\rho_b L^2 \sqrt{F}}{\sqrt{\pi E_b}} \left[ \frac{\cos^2 \alpha \sqrt{\cos \left( \frac{\alpha}{q-1} \right) } + 2 \left( q-1 \right) \sin^2 \alpha \sin^2 \left( \frac{\alpha}{q-1} \right) }{2 \sqrt{ \sin \left( \frac{\alpha}{p-1} \right) \cos \alpha + \sin \left( \frac{\alpha \left( q-2\right)}{q-1} \right) } } \right].
\eea

\noindent Normalizing $m_s$ and $m_b$ and summing we get (\ref{n1q_muB}). 

\end{proof}

\begin{corollary}

\noindent The minimal mass \emph{superstructure} is achieved for $q \rightarrow \infty$ and $t_0 = 0$, leading to the following mass:

\bea
\mu_{B} \left(\alpha, q \rightarrow \infty \right) = \frac{\alpha}{4} + \frac{\eta \cos^2 \alpha}{2 \sqrt{\sin \alpha}}.
\label{n1qinf_muB}
\eea

\end{corollary}

\begin{proof}
The plot in Fig. \ref{plot_sup_n1p>1} vs. $\alpha$ for different $q$ shows that (\ref{n1q_muB})  has a global minimum value at $q \rightarrow \infty$.
\end{proof}

It is important to consider that, for the solution $q \rightarrow \infty$, buckling is not the mode of failure since the lengths of the bars approaches zero. Also note that at $\alpha = 90 \ deg$, $\mu_B = \pi/8$.

\begin{figure}[hbt]
\unitlength1cm
\scalebox{0.8}{
\begin{picture}(14,17.5)
\put(-0.5,0){\psfig{figure=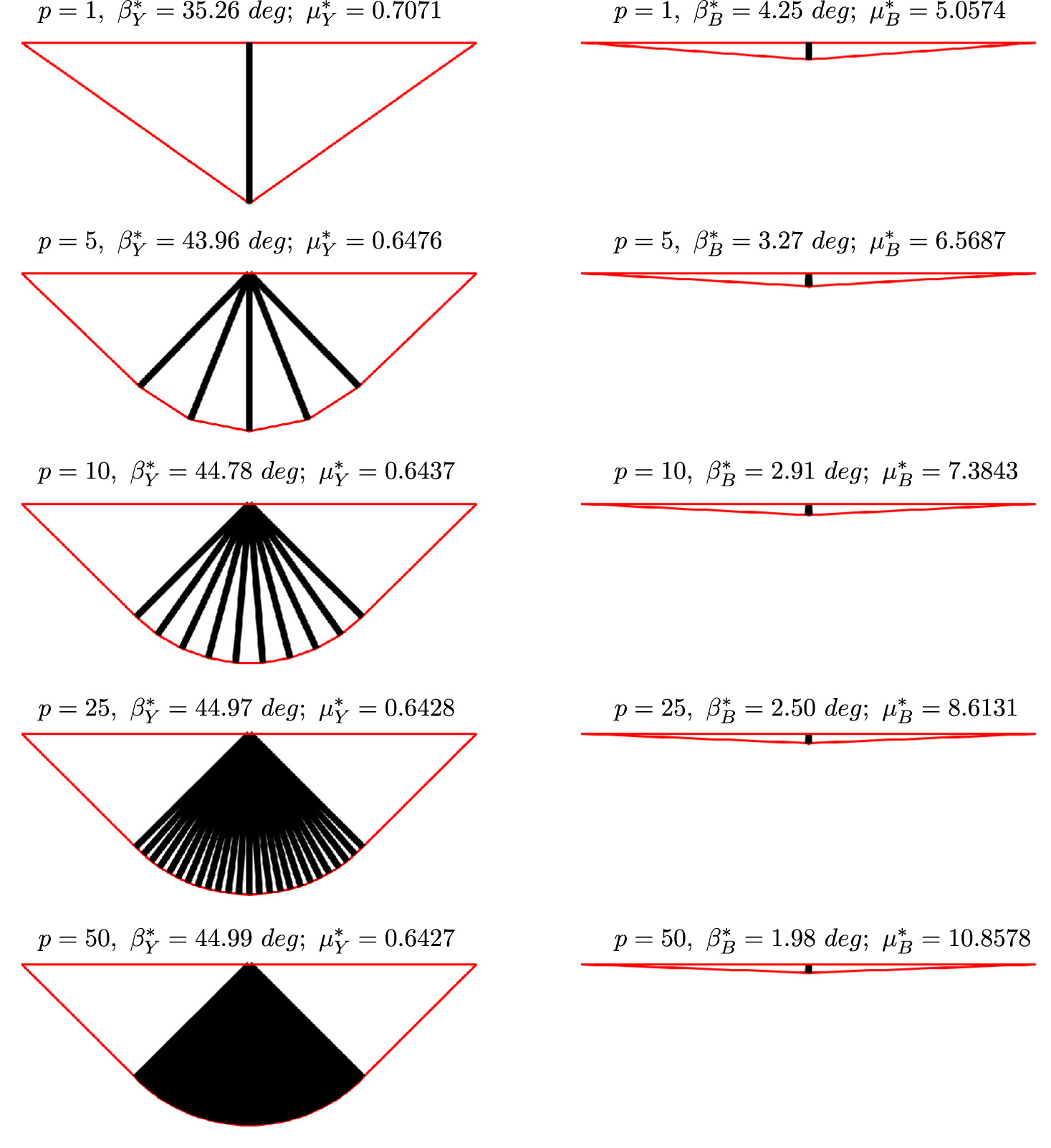,height=17.5cm}}
\end{picture}}
\caption{Optimal topologies of \emph{substructure} bridges with $n=1$ under yielding constraints (left) and buckling constraints (right) for different $p$, (steel for bars and cables, $F = 1 \ N$, $L = 1 \ m$).}
\label{fig:PANEL2}
\end{figure}

The left side of Fig. \ref{fig:PANEL1} shows a sequence of \emph{superstructures} under yielding constraints, as $q$ increases. From (\ref{n1q_muY}) the mass is minimized at $q \rightarrow \infty$ and $\alpha^*_Y = 45 \ deg \ (\rho = 1)$. 
The right side of Fig. \ref{fig:PANEL1} shows a sequence of \emph{superstructures} under buckling constraints, as $q$ increases. From plot in Fig. \ref{plot_sup_n1p>1} the mass is minimized at $\alpha=90 \ deg$ for $q=\infty$ ($\eta = 857.71$, same steel/steel material as above).

\subsection{\emph{Substructure} Bridge with Complexity $(n,p,q)=(1,p>1,0)$} \label{sub_n1p>1}

\bigskip

\noindent Refer to Fig. \ref{theorems_sub_n1p>1} for the notation. The angle between the bars is:

\beq
\label{gamma_sub}
\gamma= \frac{2 \beta}{p - 1}.
\eeq

\noindent The lengths of the bars and cables are:

\beq
s_0 = \frac{L}{2}, \ \ \
s_1 = \frac{L}{2} \cos \beta, \ \ \
s_2 = L \sin \beta \sin \left( \frac{\beta}{p - 1} \right), \ \ \
b_1 = b_2 = \frac{L}{2} \sin \beta.
\label{geom_n1p>1}
\eeq

\noindent From the equilibrium equations, we obtain the following relations for the forces:

\bea
f_1 = \frac{F}{4 \left[ \cos \beta + \sin \left( \frac{ \beta \left( p - 2 \right)}{p-1} \right) / \sin \left( \frac{\beta}{p-1} \right) \right] }, \ \ \
f_2 = 2 f_1, \\
t_2 = \frac{f_2}{2 \sin \left( \frac{\beta}{p-1} \right)}, \ \ \ 
t_1 = {t_2}{\cos \left( \frac{\beta}{p-1} \right)}. 
\eea

\begin{theorem} \label{theo:min_yielding_n1p>1}
Consider a \emph{substructure} bridge with topology described by (\ref{geom_n1p>1}), with complexity $(n,p,q)=(1,p,0)$ (Fig. \ref{theorems_sub_n1p>1}). At the yield condition the dimensionless total mass is: 

\bea
\mu_{Y} \left(\beta,p \right) = \frac{t_0}{F} + \frac{1}{4} \left[ \frac{ \left( p - 1 \right) \sin \beta \sin \left( \frac{\beta}{p-1} \right) + \cos \beta \cos \left( \frac{\beta}{p-1} \right) }{ \cos \beta \sin \left( \frac{\beta}{p-1} \right) + \sin \left( \frac{\beta \left( p - 2 \right)}{p-1} \right)} \right] \nonumber + \\ \rho \frac{\left( p - 1 \right) \sin \beta}{4 \left[ \cos \beta + \sin \left( \frac{\beta \left( p-2 \right)}{p-1} \right) / \sin \left( \frac{\beta}{p-1} \right) \right]}.
\label{n1p_muY}
\eea

\end{theorem}

\begin{proof}

Observing that the \emph{substructure} bridge of the present theorem is the \emph{dual} structure of the \emph{superstructure} bridge of Theorem \ref{theo:min_yielding_n1q>1}, we can easily obtain the proof of this theorem. 

\end{proof}

\begin{corollary}

The minimal mass in (\ref{n1p_muY}) is achieved at infinite complexity $p \rightarrow \infty$ and $t_0 = 0$. The minimal mass at yielding for a \emph{substructure} bridge is:

\bea
\mu_{Y}^* (\beta^*_Y, p^*)= \frac{1}{4} \left[ \sqrt{\rho} + \left( 1 + \rho \right) \arctan \frac{1}{\sqrt{\rho}} \right],
\label{n1p_mustY}
\eea

\noindent where $p^* \rightarrow \infty$ and the optimal angle $\beta^*_Y$ is:

\bea
\beta^*_Y = \arctan \left( \frac{1}{\sqrt{\rho}} \right).
\label{n1p_betastY}
\eea

\end{corollary}

\begin{proof}

\noindent Substitute $p \rightarrow \infty$ into Eq. (\ref{n1p_muY}) to obtain:

\bea
\mu_{Y}^* (\beta, p^* \rightarrow \infty)= \frac{\beta}{4} \left( 1 + \rho \right) + \frac{1}{4 \tan \beta}.
\label{n1p_mustYp}
\eea

The value of $\beta$ that minimizes (\ref{n1p_mustYp}) is (\ref{n1p_betastY}). Fig. \ref{plot_sub_n1p>1} shows how mass (\ref{n1p_muY}) varies with $p$ and $\beta$. The optimal $p^*$ is deduced from the plot of Fig. \ref{plot_sub_n1p>1} and the optimal angle is computed analytically in Eq. (\ref{n1p_betastY}).   

\end{proof}

\begin{theorem} \label{theo:min_buckling_n1p>1}
Consider a \emph{substructure} bridge with topology defined by (\ref{geom_n1p>1}), with complexity $(n,p,q)=(1,p,0)$, See Fig. \ref{theorems_sub_n1p>1}. At the buckling condition the dimensionless total mass is minimized at $p=2$ and $t_0 = 0$, where:

\bea
\mu_{B} \left(\beta,p=2 \right) = \frac{1+\tan^2 \beta}{4 \tan \beta} + \frac{\eta}{2} \frac{\tan^2 \beta}{\left( 1+\tan^2 \beta \right)^{3/4}}.
\label{n1p2_muB}
\eea

\end{theorem}

\begin{corollary}

\noindent The minimal mass \emph{substructure} is achieved for $p=1$.

\end{corollary}

\begin{proof}

\noindent The mass of a \emph{substructure} with topology of $n=1$ defined by (\ref{geom_n1p>1}), for a general $p>1$ is:

\bea
\mu_{B} \left(\beta,p \right) = \frac{t_0}{F} + \frac{1}{4} \left[ \frac{ \left( p - 1 \right) \sin \beta  \sin \left( \frac{\beta}{p-1} \right) + \cos \beta \cos \left( \frac{\beta}{p-1} \right) }{ \cos \beta \sin \left( \frac{\beta}{p-1} \right) + \sin \left( \frac{\beta \left( p - 2 \right)}{p-1} \right)} \right] \nonumber + \\ \frac{\eta}{2 \sqrt{2}} \frac{ \left( p-2+\sqrt{2} \right) \sin^2 \beta }{ \sqrt{ \cos \beta + \sin \left( \frac{\beta \left( p-2 \right)}{p-1} \right) / \sin \left( \frac{\beta}{p-1} \right) } }.
\label{n1p_muB}
\eea

The plot of (\ref{n1p_muB}) in Fig. \ref{plot_sup_n1p>1} vs. $\beta$ for different $p$ shows that (\ref{n1p_muB})  has a minimum value at $p=2$. However, the mass at $p = 2$, (\ref{n1p2_muB}), is larger then the mass (\ref{n1p1:mustB}) at $p =1$ from Corollary \ref{theo:n1p1:min_buckling}.

\end{proof}

The left side of Fig. \ref{fig:PANEL2} shows a sequence of \emph{substructures} under yielding constraints, as $p$ increases. From (\ref{n1p_muY}) the mass is minimized at $p \rightarrow \infty$ and $\beta^*_Y = 45 \ deg \ (\rho = 1)$. The right side of Fig. \ref{fig:PANEL2} shows a sequence of \emph{substructures} under buckling constraints, as $p$ increases. From plot in Fig. \ref{plot_sup_n1p>1} the mass is minimized at $\beta=90 \ deg$ for $p=1$ ($\eta = 857.71$, same steel/steel material as above).

\begin{theorem} \label{theo:BOB06112014_1}

A minimal mass \emph{superstructure} constrained against yielding with hinge/roller boundary conditions, has the same optimal topology as a minimal mass \emph{superstructure} constrained against buckling and hinge/hinge boundary conditions.

\end{theorem}

\begin{proof}

\cite{Michell1904} proved that the minimal mass structure constrained against yielding with hinge/roller boundary conditions has the topology of the right side of Fig. \ref{fig:PANEL1} as $q \rightarrow \infty$ and $\alpha \rightarrow 90 \ deg$. Theorem \ref{theo:min_buckling_n1q>1} provides the same topology for hinge/hinge constraints.

\end{proof}

\begin{theorem} \label{theo:BOB06112014_2}

The minimal mass \emph{nominal} bridge constrained against yielding is obtained combining the optimal \emph{superstructure} topology (Fig. \ref{fig:PANEL1}, left side as $q \rightarrow \infty$) with the optimal \emph{substructure} topology (Fig. \ref{fig:PANEL1}, left side as $p \rightarrow \infty$).

\end{theorem}

\begin{proof}

\cite{Michell1904} obtained these same results by starting with a continuum and optimizing the shape.

\end{proof}

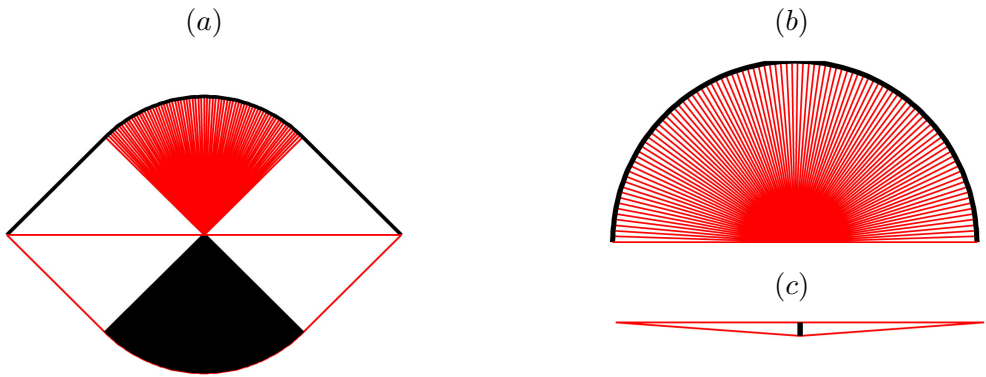
\begin{figure}[hb]
\unitlength1cm
\begin{picture}(10,5.5)
\put(0.5,0){\psfig{figure=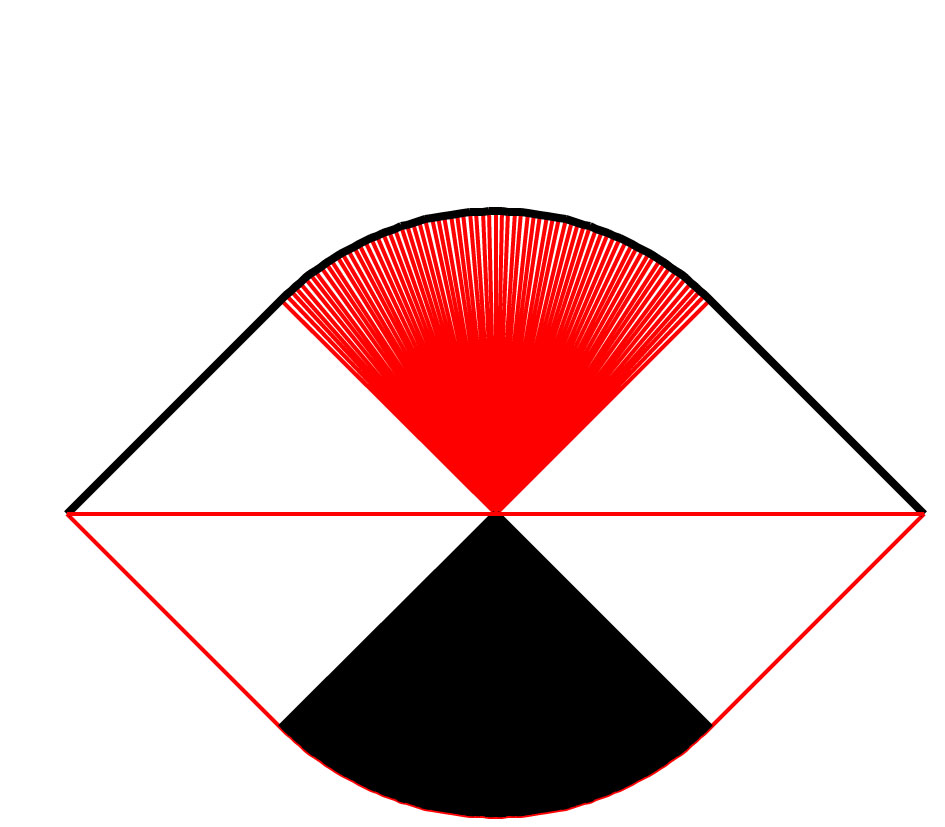,height=5cm}}
\put(8.5,1.2){\psfig{figure=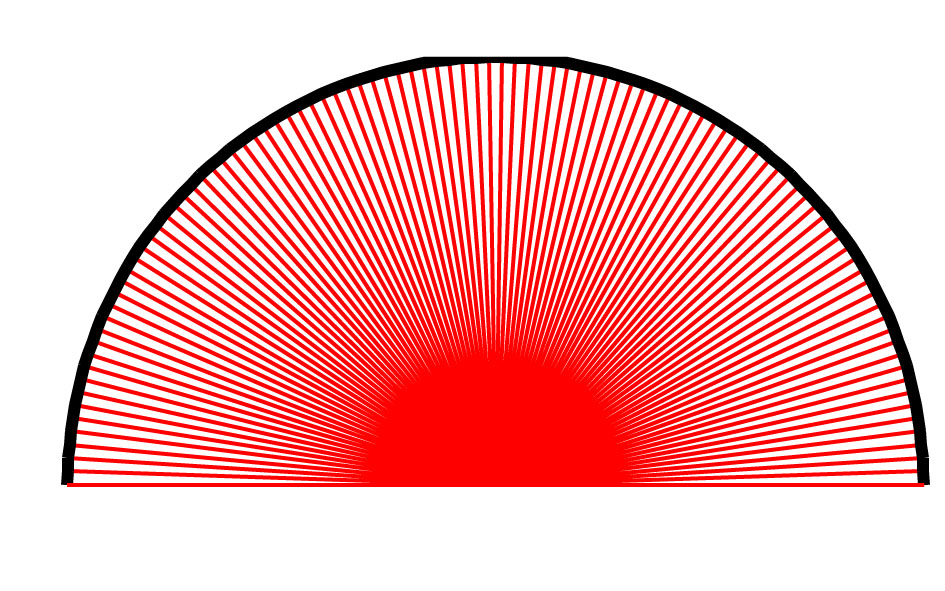,height=3.3cm}}
\put(8.4,-1.25){\psfig{figure=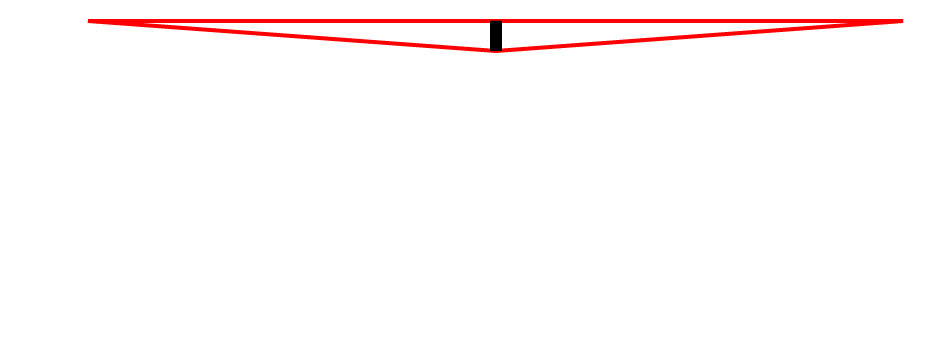,height=2.1cm}}
\put(3.25,4.6){{$\mbox{$(a) $}$}}
\put(11,4.6){{$\mbox{$(b) $}$}}
\put(11,1.1){{$\mbox{$(c) $}$}}
\end{picture}
\caption{Minimal mass bridges under (a) yielding constrained \emph{nominal} bridges, (b) buckling constrained \emph{superstructure} bridge and (c) buckling constrained \emph{substructure} bridge.}
\label{fig:optimal_n1qp}
\end{figure}

Fig. \ref{fig:optimal_n1qp}(a) illustrates the minimal mass \emph{nominal} bridge under yielding constraints (Theorem \ref{theo:BOB06112014_1}), leading to complexity $(n,p,q)=(1,\infty,\infty)$. Fig. \ref{fig:optimal_n1qp}(b) illustrates the minimal mass \emph{superstructure} bridge under buckling constraints, leading to complexity $(n,p,q)=(1,0,q \rightarrow \infty)$. Fig. \ref{fig:optimal_n1qp}(c) illustrates the minimal mass \emph{substructure} bridge under buckling constraints, leading to complexity $(n,p,q)=(1,1,0)$. 

\section{Mass of Bridges of Complexity $(n,p,q)=(n,1 \ or \ 0,1 \ or \ 0)$}
\label{Sec:nqp1}

This section finds the minimal mass of a tensegrity bridge of any complexity $n$. As in previous sections, no deck mass is yet added til the next section. The total external load is a given constant force $F$. Dividing the span into $2^n$ equal sections, creates nodes at each section that carries load $f$, given by,

\beq
\label{fnodes}
f = \frac{F}{2^n}.
\eeq

Distributing the total external load equally among the number of spans ($2^n$) of the subsections requires internal nodes to carry load $f=F/2^n$, and the external nodes of the deck to carry load $f/2$. 

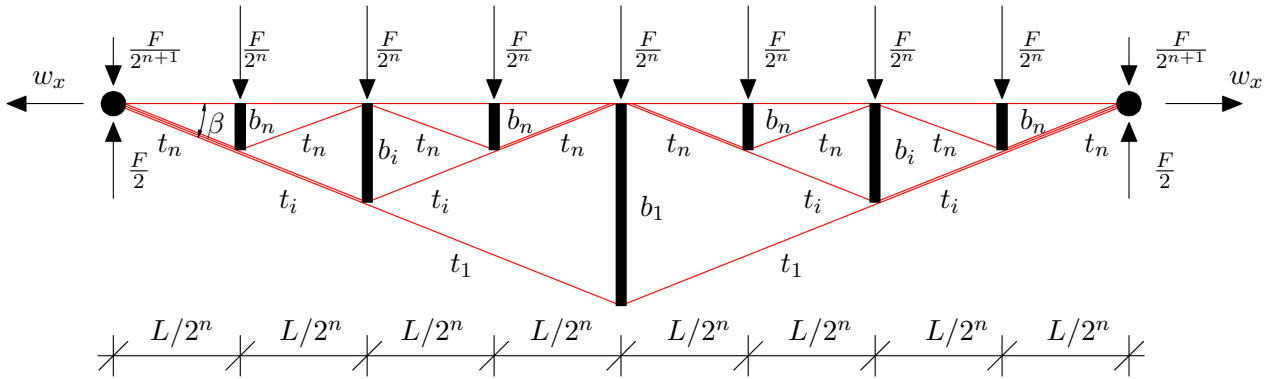
\begin{figure}[hb] 
\unitlength1cm
\begin{picture}(10,5)
\put(-0.1,0.1){\psfig{figure=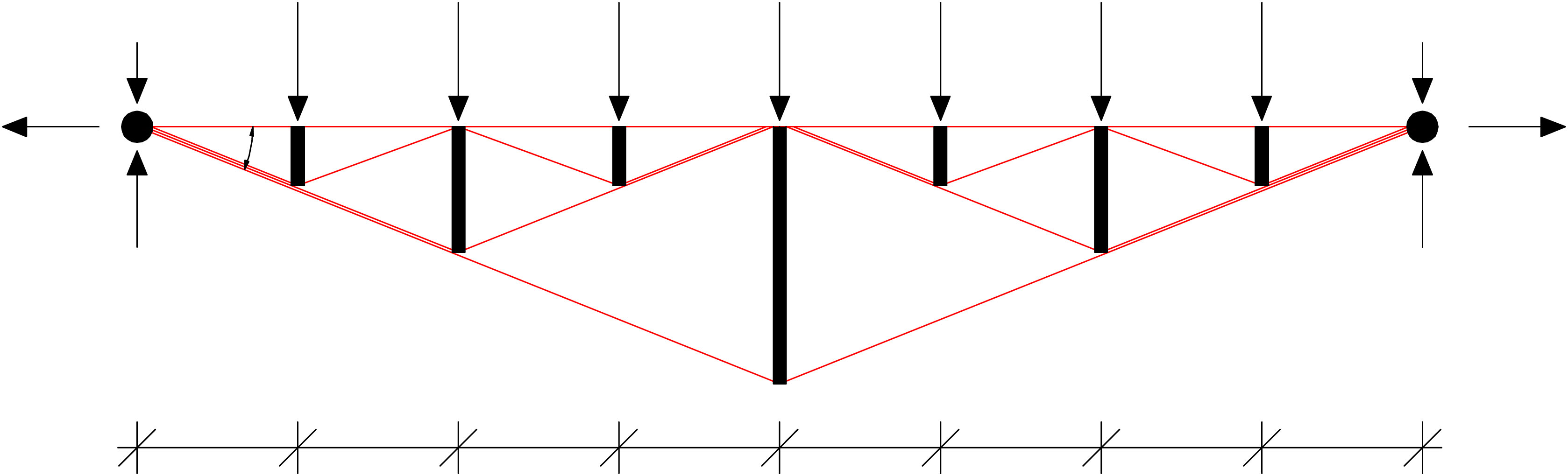,height=4.95cm}}
\put(15,4.35){{$\mbox{$\frac{F}{2^{n+1}}$}$}}
\put(13.2,4.35){{$\mbox{$\frac{F}{2^n}$}$}}
\put(11.5,4.35){{$\mbox{$\frac{F}{2^n}$}$}}
\put(9.85,4.35){{$\mbox{$\frac{F}{2^n}$}$}}
\put(8.15,4.35){{$\mbox{$\frac{F}{2^n}$}$}}
\put(6.45,4.35){{$\mbox{$\frac{F}{2^n}$}$}}
\put(4.75,4.35){{$\mbox{$\frac{F}{2^n}$}$}}
\put(3,4.35){{$\mbox{$\frac{F}{2^n}$}$}}
\put(1.5,4.35){{$\mbox{$\frac{F}{2^{n+1}}$}$}}
\put(2.55,3.35){{$\mbox{$\beta$}$}}
\put(15,2.75){{$\mbox{$\frac{F}{2}$}$}}
\put(1.5,2.75){{$\mbox{$\frac{F}{2}$}$}}
\put(8.25,2.25){{$\mbox{$b_1$}$}}
\put(10.1,1.5){{$\mbox{$t_1$}$}}
\put(5.75,1.5){{$\mbox{$t_1$}$}}
\put(11.6,3){{$\mbox{$b_i$}$}}
\put(10.4,2.35){{$\mbox{$t_i$}$}}
\put(12.2,2.35){{$\mbox{$t_i$}$}}
\put(4.8,3){{$\mbox{$b_i$}$}}
\put(3.5,2.35){{$\mbox{$t_i$}$}}
\put(5.5,2.35){{$\mbox{$t_i$}$}}
\put(6.5,3.4){{$\mbox{$b_n$}$}}
\put(1.9,3.1){{$\mbox{$t_n$}$}}
\put(3.8,3.1){{$\mbox{$t_n$}$}}
\put(5.3,3.1){{$\mbox{$t_n$}$}}
\put(7.2,3.1){{$\mbox{$t_n$}$}}
\put(8.6,3.1){{$\mbox{$t_n$}$}}
\put(10.6,3.1){{$\mbox{$t_n$}$}}
\put(12.1,3.1){{$\mbox{$t_n$}$}}
\put(14.1,3.1){{$\mbox{$t_n$}$}}
\put(3.1,3.4){{$\mbox{$b_n$}$}}
\put(9.9,3.4){{$\mbox{$b_n$}$}}
\put(13.25,3.4){{$\mbox{$b_n$}$}}
\put(13.6,0.6){{$\mbox{${L}/{2^n}$}$}}
\put(12,0.6){{$\mbox{${L}/{2^n}$}$}}
\put(10.2,0.6){{$\mbox{${L}/{2^n}$}$}}
\put(8.6,0.6){{$\mbox{${L}/{2^n}$}$}}
\put(6.8,0.6){{$\mbox{${L}/{2^n}$}$}}
\put(5.1,0.6){{$\mbox{${L}/{2^n}$}$}}
\put(3.5,0.6){{$\mbox{${L}/{2^n}$}$}}
\put(1.8,0.6){{$\mbox{${L}/{2^n}$}$}}
\put(0.25,4){{$\mbox{${w_x}$}$}}
\put(16,4){{$\mbox{${w_x}$}$}}
\end{picture}
\caption{Adopted notations for forces and lengths of bars and cables for a substructure with generic complexity $n$ and $p = 1$.}
\label{theorems_module_sub}
\end{figure}

\begin{figure}[hb] 
\unitlength1cm
\begin{picture}(10,5)
\put(0,0){\psfig{figure=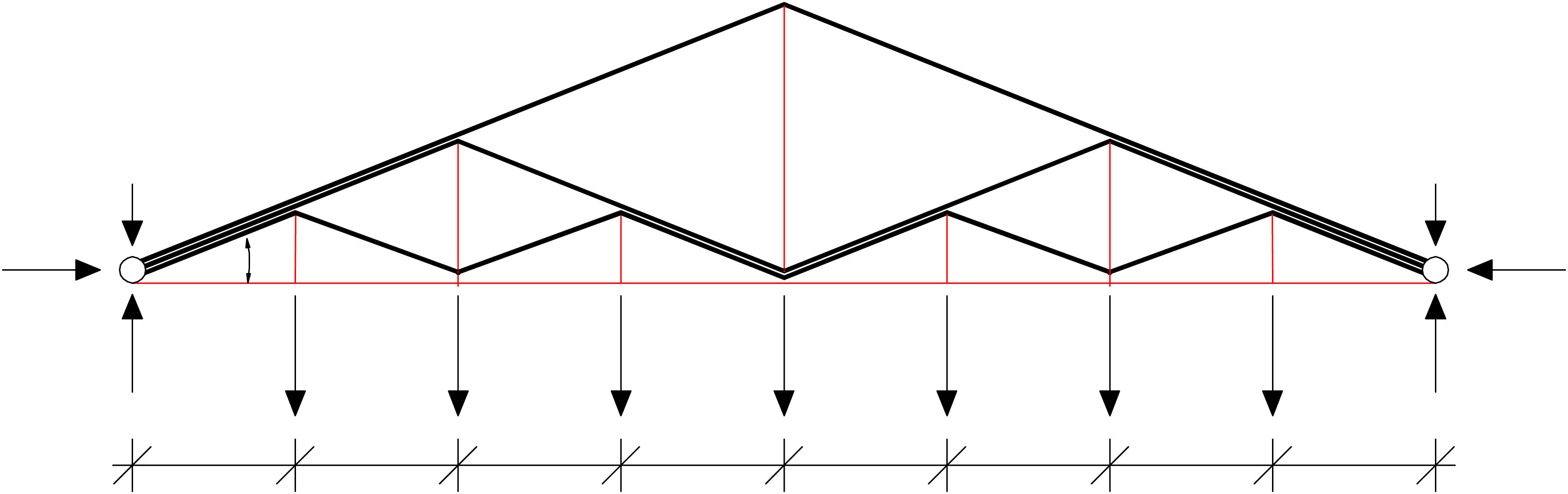,height=5cm}}
\put(14.75,2.85){{$\mbox{$\frac{F}{2^{n+1}}$}$}}
\put(13.2,1.35){{$\mbox{$\frac{F}{2^n}$}$}}
\put(11.5,1.35){{$\mbox{$\frac{F}{2^n}$}$}}
\put(9.85,1.35){{$\mbox{$\frac{F}{2^n}$}$}}
\put(8.15,1.35){{$\mbox{$\frac{F}{2^n}$}$}}
\put(6.45,1.35){{$\mbox{$\frac{F}{2^n}$}$}}
\put(4.75,1.35){{$\mbox{$\frac{F}{2^n}$}$}}
\put(3,1.35){{$\mbox{$\frac{F}{2^n}$}$}}
\put(0.25,2.85){{$\mbox{$\frac{F}{2^{n+1}}$}$}}
\put(0.1,1.8){{$\mbox{${w_x}$}$}}
\put(15.4,1.8){{$\mbox{${w_x}$}$}}
\put(2.6,2.3){{$\mbox{$\alpha$}$}}
\put(15,1.25){{$\mbox{$\frac{F}{2}$}$}}
\put(1.5,1.25){{$\mbox{$\frac{F}{2}$}$}}
\put(8,3.75){{$\mbox{$b_1$}$}}
\put(10.25,4.25){{$\mbox{$t_1$}$}}
\put(5.5,4.25){{$\mbox{$t_1$}$}}
\put(11.3,3){{$\mbox{$b_i$}$}}
\put(10.3,2.9){{$\mbox{$t_i$}$}}
\put(12.1,2.9){{$\mbox{$t_i$}$}}
\put(4.7,3){{$\mbox{$b_i$}$}}
\put(3.6,2.9){{$\mbox{$t_i$}$}}
\put(5.4,2.9){{$\mbox{$t_i$}$}}
\put(2,2.2){{$\mbox{$t_n$}$}}
\put(3.7,2.2){{$\mbox{$t_n$}$}}
\put(5.4,2.2){{$\mbox{$t_n$}$}}
\put(7.1,2.2){{$\mbox{$t_n$}$}}
\put(8.7,2.2){{$\mbox{$t_n$}$}}
\put(10.5,2.2){{$\mbox{$t_n$}$}}
\put(12,2.2){{$\mbox{$t_n$}$}}
\put(13.6,2.2){{$\mbox{$t_n$}$}}
\put(6.4,2.3){{$\mbox{$b_n$}$}}
\put(3.1,2.3){{$\mbox{$b_n$}$}}
\put(9.7,2.3){{$\mbox{$b_n$}$}}
\put(13,2.3){{$\mbox{$b_n$}$}}
\put(13.4,-0.2){{$\mbox{${L}/{2^n}$}$}}
\put(11.8,-0.2){{$\mbox{${L}/{2^n}$}$}}
\put(10.2,-0.2){{$\mbox{${L}/{2^n}$}$}}
\put(8.5,-0.2){{$\mbox{${L}/{2^n}$}$}}
\put(6.7,-0.2){{$\mbox{${L}/{2^n}$}$}}
\put(5,-0.2){{$\mbox{${L}/{2^n}$}$}}
\put(3.4,-0.2){{$\mbox{${L}/{2^n}$}$}}
\put(1.7,-0.2){{$\mbox{${L}/{2^n}$}$}}
\end{picture}
\caption{Adopted notations for forces and lengths of bars and cables for a \emph{superstructure} with complexity $(n,p)=(n,1)$}
\label{theorems_module_super}
\end{figure}

\subsection{\emph{Substructure} Bridge with Complexity $(n,p,q)=(n,1,0)$} \label{sub_p1}

In this case, we make use of the notation illustrated in Fig. \ref{theorems_module_sub} in which complexity $p$ is fixed to be one. Since $n$ is the number of self-similar iteration of the basic module of Fig. \ref{basic_module}c at different scales, it can be defined $n$ orders of bars and cables. The length of the generic $i^{th}$ bar and the length of the generic $i^{th}$ cables are,

\bea
b_i =\frac{L}{2^{i} } \tan \beta,\qquad i=1-n,
\label{lbars}
\eea

\bea
s_i =\frac{L}{2^{i} \cos \beta },\qquad i=1-n.
\label{lstrings}
\eea

\noindent From the equilibrium conditions, the axial force in each bar and the axial force in each cable are given by,

\bea
f_{bi} =\frac{F}{2^{i} },
\label{fbars}
\eea

\bea
t_{si} =\frac{F}{2^{\left(1+i\right)} \sin \beta }.
\label{tstrings}
\eea

\begin{theorem} \label{theo:min_yielding_np1}
Consider a \emph{substructure} bridge with topology defined by (\ref{nnodes}), (\ref{nconn_sub}), (\ref{lbars}) and (\ref{lstrings}), with complexity $(n,p,q)=(n,1,0)$, see Fig. \ref{theorems_module_sub}. The minimal mass design under only yielding constraints is given by the following aspect angle:

\bea
\beta^*_{Y} = \arctan{ \left(\frac{1}{\sqrt{1+\rho } } \right)},
\label{np1_betastY}
\eea

\noindent which corresponds to the following dimensionless minimal mass 

\bea
\mu_{Y}^{*} = \left(1-\frac{1}{2^{n} } \right)\sqrt{1+\rho }.
\label{np1_mustY}
\eea

\end{theorem}

\begin{proof}

\noindent Observing the multiscale structure of Fig. \ref{theorems_module_sub} it's clear that the number of bars and the number of cables of $i^{th}$ order are

\bea
n_{si} = 2^i, \ \ \ n_{bi} = 2^{i-1}.
\label{n_bars_n_strings_ith}
\eea

\noindent The total mass of the structure is:

\bea
m_Y = \frac{\rho_s}{\sigma_s} \sum_{i=1}^n {n_{si} t_{si} s_i} + \frac{\rho_b}{\sigma_b} \sum_{i=1}^n {n_{bi} f_{bi} b_i}.
\label{np1_mass}
\eea

\noindent Substituting (\ref{lbars}), (\ref{lstrings}), (\ref{fbars}), (\ref{tstrings}) and (\ref{n_bars_n_strings_ith}) into (\ref{np1_mass}) yields,

\bea
m_Y = \frac{F L}{2} \left( \sum_{i=1}^n {\frac{1}{2^i}} \right) \left( \frac{\rho_s}{\sigma_s} \frac{1}{\sin \beta \cos \beta} + \frac{\rho_b}{\sigma_b} \tan \beta \right).
\label{np1_mass_2}
\eea

\noindent Using the following identities in (\ref{np1_mass_2}),

\bea
\sum_{i=1}^n {\frac{1}{2^i}} = \left( 1 - \frac{1}{2^n} \right), \ \ \ \frac{1}{\sin \beta \cos \beta} = \frac{1+\tan^2 \beta}{\tan \beta},
\label{np1_positions}
\eea

\noindent we obtain:

\bea
m_Y = \frac{F L}{2} \left( 1 - \frac{1}{2^n} \right) \left[ \frac{\rho_s}{\rho_b} \frac{\left( 1+\tan^2 \beta \right)}{\tan \beta} + \frac{\rho_b}{\sigma_b} \tan \beta \right].
\label{np1_mass_3}
\eea

\noindent Switching to the dimensionless mass defined in (\ref{mu}) we have:

\bea
\mu_Y = \frac{1}{2} \left( 1 - \frac{1}{2^n} \right) \left[  \frac{\left( 1+\tan^2 \beta \right)}{\tan \beta}  + \rho \tan \beta \right].
\label{np1_mu}
\eea

\noindent The solution for minimal mass can be achieved from,

\bea
\frac{\partial \mu_{Y}}{\partial \tan \beta}= \frac{1}{2} \left( 1 - \frac{1}{2^n} \right) \left[ - \frac{\left( \tan^2 \beta +1 \right) }{\tan^2 \beta} + 2 + \rho \right] = 0,
\label{np1:dmuy}
\eea

\noindent yielding the optimal angle of (\ref{np1_betastY}). Substituting it into (\ref{np1_mu}) concludes the proof.

\end{proof}

Note from (\ref{n1p1:betastY}) and (\ref{np1_betastY}) that the optimal angle $\beta$ does not depend upon the choice of $n$. The minimal mass solution under yielding constraints depends only on the material choice $\rho$ (\ref{rho}), and the complexity parameter $n$. Note that, since the total external force $F$ is a specified constant, the optimum complexity is $n = 1$.  However if the total vertical force depends upon $n$ as it will in the next section dealing with massive decks, or with massive joints, then the optimal complexity will be shown to be $n>1$.

\begin{theorem} \label{theo:mim_buckling_p1}

Consider a \emph{substructure} bridge with topology defined by (\ref{nnodes}), (\ref{nconn_sub}), (\ref{lbars}) and (\ref{lstrings}), with complexity $(n,p,q)=(n,1,0)$, see Fig. \ref{theorems_module_sub}. The minimal mass design under yielding  and buckling constraints, is given by the following aspect angle:

\bea
\beta _{B}^{*} = \arctan{ \left\{ \frac{1}{12 \alpha_2 \eta} \left[\alpha_3 + \alpha_1 \left( \frac{\alpha_1}{\alpha_3} - 1 \right) \right] \right\} }.
\label{np1_betastB}
\eea

\noindent which corresponds to the following dimensionless minimal mass: 

\bea
\mu_{B}^{*} = \alpha_1 \frac{1+\tan^2 \beta_B^*}{2 \tan \beta_B^*} + \eta \alpha_2 \tan^2 \beta_B^*,
\label{np1_mustB}
\eea

\noindent where:

\bea
\alpha_1 = \left( 1 - \frac{1}{2^n} \right),
\label{np1_alpha_1}
\eea

\bea
\alpha_2 = \left( \frac{1 + 2 \sqrt{2}}{7} \right) \left( 1 - \frac{1}{2^{3 n / 2}} \right),
\label{np1_alpha_2}
\eea

\bea
\alpha_3 = \left( 216 \alpha_1 \alpha_2^2 \eta^2 - \alpha_1^3 + 12 \sqrt{324 \alpha_1^2 \alpha_2^4 \eta^4 - 3 \alpha_1^4 \alpha_2^2 \eta^2} \right)^{1/3}.
\label{np1_alpha_3}
\eea

\end{theorem}

\begin{proof}

\noindent The total mass of cables is given by,

\bea
m_s = \sum_{i=1}^n \frac{\rho_s}{\sigma_s} n_{si} t_{si} s_i.
\label{np1_ms_1}
\eea

\noindent Substituting (\ref{lstrings}), (\ref{tstrings}), (\ref{n_bars_n_strings_ith}) into (\ref{np1_ms_1}) and making use of identities (\ref{np1_positions}),

\bea
m_s = \frac{F L}{2} \frac{\rho_s}{\sigma_s} \left( \frac{1 + \tan^2 \beta}{\tan \beta} \right) \left( 1 - \frac{1}{2^n} \right).
\label{np1_ms_2}
\eea

\noindent This corresponds to the following normalized mass

\bea
\mu_s = \left( \frac{1 + \tan^2 \beta}{2 \tan \beta} \right) \left( 1 - \frac{1}{2^n} \right).
\label{np1_mus}
\eea

\noindent The total mass of bars, making use of (\ref{m_b}), is

\bea
m_b = \sum_{i=1}^n n_{bi} m_{bi} = \sum_{i=1}^n n_{bi} \frac{2 \rho_b}{\sqrt{\pi E_b}} b_i^2 \sqrt{f_i}.
\label{np1_mb_1}
\eea

\noindent Substituting (\ref{lbars}), (\ref{n_bars_n_strings_ith}) and (\ref{fbars}) into (\ref{np1_mb_1}) yields

\bea
m_b = \frac{\rho_b L^2 \sqrt{F}}{\sqrt{\pi E_b}} \tan^2 \beta \sum_{i=1}^n {\frac{1}{2^{3 i/2}}}. 
\label{np1_mb_2}
\eea

\noindent Since

\bea
\sum_{i=1}^n {\frac{1}{2^{3 i/2}}} = \left( \frac{1 + 2 \sqrt{2}}{7} \right) \left( 1 - \frac{1}{2^{3 n / 2}} \right),
\label{np1_position_2}
\eea

\noindent normalizing we get the following dimensionless mass of bars,

\bea
\mu_b = \eta \tan^2 \beta \left( \frac{1 + 2 \sqrt{2}}{7} \right) \left( 1 - \frac{1}{2^{3 n / 2}} \right).
\label{np1_mub_1}
\eea

\noindent The total mass is then the sum of (\ref{np1_mus}) and (\ref{np1_mub_1}) and introducing constants $\alpha_1$ and $\alpha_2$ given in (\ref{np1_alpha_1}) and (\ref{np1_alpha_2}):

\bea
\mu_B = \mu_s + \mu_b= \alpha_1 \frac{1+\tan^2 \beta}{2 \tan \beta} + \eta \alpha_2 \tan^2 \beta.
\label{np1_muB}
\eea

\noindent The solution for minimal mass can be achieved from,

\bea
\frac{\partial \mu_{B}}{\partial \tan \beta}= \alpha_1 \left( 1 - \frac{1 - \tan^2 \beta}{2 \tan^2 \beta} \right) + 2 \eta \alpha_2 \tan \beta = 0,
\label{np1:dmuB}
\eea

\noindent yielding the optimal angle (\ref{np1_betastB}) by solving the following cubic equation:

\bea
4 \frac{\alpha_2}{\alpha_1} \eta \tan^3 \beta + \tan^2 \beta - 1 = 0.
\label{np1:dmuB_2}
\eea

Note that the optimal angle given in (\ref{np1_betastB}) reduces to the optimal angle given in (\ref{n1p1:betastB}) for the particular case $n=1$. Then, substituting (\ref{np1_betastB}) into (\ref{np1_muB}) concludes the proof.

\end{proof}

\subsection{\emph{Superstructure} Bridge with Complexity $(n,p,q)=(n,0,1)$} \label{super_nq1}

In this case, we make use of the notation illustrated in Fig. \ref{theorems_module_super} in which complexity $q$ is fixed to be one. Since $n$ is the number of self-similar iteration of the basic module of Fig. \ref{basic_module}b at different scales, it can be defined $n$ orders of bars and cables. The length of the generic $i^{th}$ bar and the length of the generic $i^{th}$ cable, for $i$ ranging from $1$ to $n$, are:

\bea
b_i =\frac{L}{2^{i} \cos \alpha}, \ \ \ 
s_i = \frac{L}{2^i} \tan \alpha.
\label{nq1_lbars_lstrings}
\eea

Moreover, looking at the equilibrium of each node of the structure, we found that the axial force in each bar and the axial force in each cable are given by:

\bea
f_{bi} = \frac{F}{2^{\left(1+i\right)} \sin \alpha }, \ \ \
t_{si} =\frac{F}{2^{i} }.
\label{nq1_fbars_tstrings}
\eea

\noindent Observing the multiscale structure of Fig. \ref{theorems_module_super} it's clear that the number of bars and the number of cables of $i^{th}$ order are:

\bea
n_{si} = 2^{i-1}, \ \ \ n_{bi} = 2^{i}.
\label{nq1_n_bars_n_strings}
\eea

\begin{theorem} \label{theo:min_yielding_nq1}

Consider a \emph{superstructure} bridge with topology defined by (\ref{nnodes}), (\ref{nconn_super}), (\ref{nq1_lbars_lstrings}), with complexity $(n,p,q)=(n,0,1)$, see Fig. \ref{theorems_module_super}. The minimal mass design under yielding constraints is given by the following aspect angle:

\bea
\alpha^*_{Y} = \arctan{ \left(\sqrt{\frac{\rho}{{1+\rho } }} \right)},
\label{nq1_alphastY}
\eea

\noindent which corresponds to the following dimensionless minimal mass: 

\bea
\mu_{Y}^{*} = \left(1-\frac{1}{2^{n} } \right)\sqrt{ \rho \left(1+\rho \right) }.
\label{nq1_mustY}
\eea

\end{theorem}

\begin{proof}

\noindent The total mass of the structure is:

\bea
m_Y = \frac{\rho_s}{\sigma_s} \sum_{i=1}^n {n_{si} t_{si} s_i} + \frac{\rho_b}{\sigma_b} \sum_{i=1}^n {n_{bi} f_{bi} b_i}.
\label{nq1_mass}
\eea

\noindent Substituting (\ref{nq1_lbars_lstrings}), (\ref{nq1_fbars_tstrings}), and (\ref{nq1_n_bars_n_strings}) into (\ref{nq1_mass}) and considering positions (\ref{np1_positions}) we get:

\bea
m_Y = \frac{F L}{2} \left( 1 - \frac{1}{2^n} \right) \left[ \frac{\rho_s}{\sigma_s} \tan \alpha + \frac{\rho_b}{\sigma_b} \frac{\left( 1 + \tan^2 \alpha \right)}{\tan \alpha} \right].
\label{nq1_mass_2}
\eea

\noindent Switching to the dimensionless mass defined in (\ref{mu}) we have:

\bea
\mu_Y = \frac{1}{2} \left( 1 - \frac{1}{2^n} \right) \left[ \tan \alpha + \rho \frac{\left( 1 + \tan^2 \alpha \right)}{\tan \alpha} \right].
\label{nq1_mu}
\eea

\noindent The solution for minimal mass can be achieved from,

\bea
\frac{\partial \mu_{Y}}{\partial \tan \alpha}= \frac{1}{2} \left( 1 - \frac{1}{2^n} \right) \left[ 1 + \rho \left( 2 - \frac{1 + \tan^2 \alpha}{\tan^2 \alpha} \right) \right] = 0,
\label{nq1:dmuy}
\eea

\noindent yielding the optimal angle of (\ref{nq1_alphastY}). Substituting it into (\ref{nq1_mu}) concludes the proof.

\end{proof}

\begin{theorem} \label{theo:min_buckling_nq1}

Consider a \emph{superstructure} bridge with topology defined by (\ref{nnodes}), (\ref{nconn_super}), (\ref{nq1_lbars_lstrings}), with complexity $(n,p,q)=(n,0,1)$, see Fig. \ref{theorems_module_super}. The minimal mass design under yielding and buckling constraints is given by the following aspect angle:

\bea
\alpha _{B}^{*} = \arctan{ \frac{1}{2} },
\label{nq1_alphastB}
\eea

\noindent which corresponds to the following dimensionless minimal mass: 

\bea
\mu_{B}^{*} = \frac{\gamma_1}{2} + \eta \gamma_2 \frac{5^{5/4}}{4},
\label{nq1_mustB}
\eea

\noindent where:

\bea
\gamma_1 = \frac{1}{2} \left( 1 - \frac{1}{2^n} \right),
\label{nq1_gamma_1}
\eea

\bea
\gamma_2 = \sqrt{2} \left( \frac{1 + 2 \sqrt{2}}{7} \right) \left( 1 - \frac{1}{2^{3 n / 2}} \right).
\label{nq1_gamma_2}
\eea

\end{theorem}

\begin{proof}

\noindent The total mass of cables is given by:

\bea
m_s = \sum_{i=1}^n \frac{\rho_s}{\sigma_s} n_{si} t_{si} s_i.
\label{nq1_ms_1}
\eea

\noindent Substituting (\ref{nq1_lbars_lstrings}), (\ref{nq1_fbars_tstrings}) and (\ref{nq1_n_bars_n_strings}) into (\ref{nq1_ms_1}) and making use of position (\ref{np1_positions}):

\bea
m_s = \frac{F L}{2} \frac{\rho_s}{\sigma_s} \left( 1 - \frac{1}{2^n} \right) \tan \alpha.
\label{nq1_ms_2}
\eea

\noindent That corresponds to the following normalized mass:

\bea
\mu_s = \frac{1}{2} \left( 1 - \frac{1}{2^n} \right) \tan \alpha.
\label{nq1_mus}
\eea

\noindent The total mass of bars, making use of (\ref{m_b}), is:

\bea
m_b = \sum_{i=1}^n n_{bi} m_{bi} = \sum_{i=1}^n n_{bi} \frac{2 \rho_b}{\sqrt{\pi E_b}} b_i^2 \sqrt{f_i}.
\label{nq1_mb_1}
\eea

\noindent Substituting (\ref{nq1_lbars_lstrings}), (\ref{nq1_fbars_tstrings}) and (\ref{nq1_n_bars_n_strings}) into (\ref{nq1_mb_1}):

\bea
m_b = \frac{\sqrt{2} \rho_b L^2 \sqrt{F}}{\sqrt{\pi E_b}} \frac{1}{\cos^2 \alpha \sqrt{\sin \alpha}} \sum_{i=1}^n {\frac{1}{2^{3 i/2}}}.
\label{nq1_mb_2}
\eea

\noindent Since:

\bea
\sum_{i=1}^n {\frac{1}{2^{3 i/2}}} = \left( \frac{1 + 2 \sqrt{2}}{7} \right) \left( 1 - \frac{1}{2^{3 n / 2}} \right), \nonumber \\
\frac{1}{\cos^2 \alpha} = 1 + \tan^2 \alpha, \ \ \ 
\frac{1}{\sqrt{\sin \alpha}} = \frac{\left( 1 + \tan^2 \alpha\right)^{1/4}}{\sqrt{\tan \alpha}},
\label{nq1_position_1}
\eea

\noindent and normalizing we get the following dimensionless mass of bars:

\bea
\mu_b = \sqrt{2} \eta \left( \frac{1 + 2 \sqrt{2}}{7} \right) \left( 1 - \frac{1}{2^{3 n / 2}} \right) \frac{\left( 1 + \tan^2 \alpha \right)^{5/4}}{\sqrt{\tan \alpha}}.
\label{nq1_mub}
\eea

\noindent The total mass is then the sum of (\ref{nq1_mus}) and (\ref{nq1_mub}) and introducing constants $\gamma_1$ and $\gamma_2$ given in (\ref{nq1_gamma_1}) and (\ref{nq1_gamma_2}):

\bea
\mu_B = \mu_s + \mu_b = \gamma_1 \tan \alpha + \eta \gamma_2 \frac{\left( 1 + \tan^2 \alpha \right)^{5/4}}{\sqrt{\tan \alpha}}
\label{nq1_muB}.
\eea

\noindent The solution for minimal mass can be achieved assuming that:

\bea
\gamma_1 \tan \alpha \ll \eta \gamma_2 \frac{\left( 1 + \tan^2 \alpha \right)^{5/4}}{\sqrt{\tan \alpha}}.
\label{nq1_hyp}
\eea

\noindent So that the (\ref{nq1_muB}) becomes:

\bea
\bar \mu_B = \eta \gamma_2 \frac{\left( 1 + \tan^2 \alpha \right)^{5/4}}{\sqrt{\tan \alpha}}.
\label{nq1_muB_bar}
\eea

\noindent The optimal angle can be obtained from:

\bea
\frac{\partial \bar \mu_{B}}{\partial \tan \alpha}= \frac{\eta}{2} \gamma_2 \left( 1 + \tan^2 \alpha \right)^{\left(1/4\right)} \left( \frac{4 \tan^2 \alpha - 1}{\tan \alpha \sqrt{\tan \alpha}} \right) = 0,
\label{np1:dmuB}
\eea

\noindent yielding the optimal angle of (\ref{nq1_alphastB}). Substituting it into (\ref{nq1_muB}) concludes the proof.

\end{proof}

\section{Introducing Deck and Joint Masses}

In previous sections, complexity $n$ was restricted to $1$. This is appropriate only when the external loads are all applied at the midspan. Real bridges cannot tolerate such an assumption. So in this section we consider a distributed load. Part of the load is the mass of the deck that must span the distance between adjacent support structures (complexity $n$ will add $2^n-1$ supports). In the section \ref{joint_mass} we will consider adding mass to make the joints, where high precision joints have less mass then rudely constructed joints.

\subsection{Including Deck Mass}
\label{deck_design}

The total load that the structure must support includes the mass of the deck, which increases with the distance that must be spanned between support points of the structure design (which is determined by the choice of complexity $n$). We therefore consider bridges with increasing complexity $n$. We will show that the smallest $n=1$ yields smallest structural mass and the largest deck mass. The required deck mass obviously approaches zero as the required deck span approaches zero, which occurs as $n \rightarrow \infty$. We will show that the mass of the deck plus the mass of the structure is minimized at a finite value of $n$.

The deck, as illustrated in Fig. \ref{fig:deck}, is composed by $2^n$ simply supported beams connecting the nodes on the deck. Let the deck parameters be labeled as: mass $m_d$, mass density $\rho_d$, yielding strength $\sigma_{d}$, width $w_d$, thickness $t_d$ and length equal to:

\bea
\ell_d = \frac{L}{2^n}.
\label{ell_deck}
\eea

The cross sectional of the deck beam has a moment of inertia equal to: $I_d = w_d t_d^3 / 12$.
\noindent Each beam is assumed to be loaded by a uniformly distributed vertical load summing to the total value $F$ and the total self weight of the deck ($\mathcal{F}$) ($g = 9.81 m s^{-2}$):

\bea
f_d = \frac{F}{L} + \frac{\mathcal{F}}{L} = \frac{F}{L} + \frac{m_d \ g \ 2^n}{L}.
\label{f_deck}
\eea

\begin{figure}[hbt]
\unitlength1cm
\begin{picture}(10,8)
\put(2,4){\psfig{figure=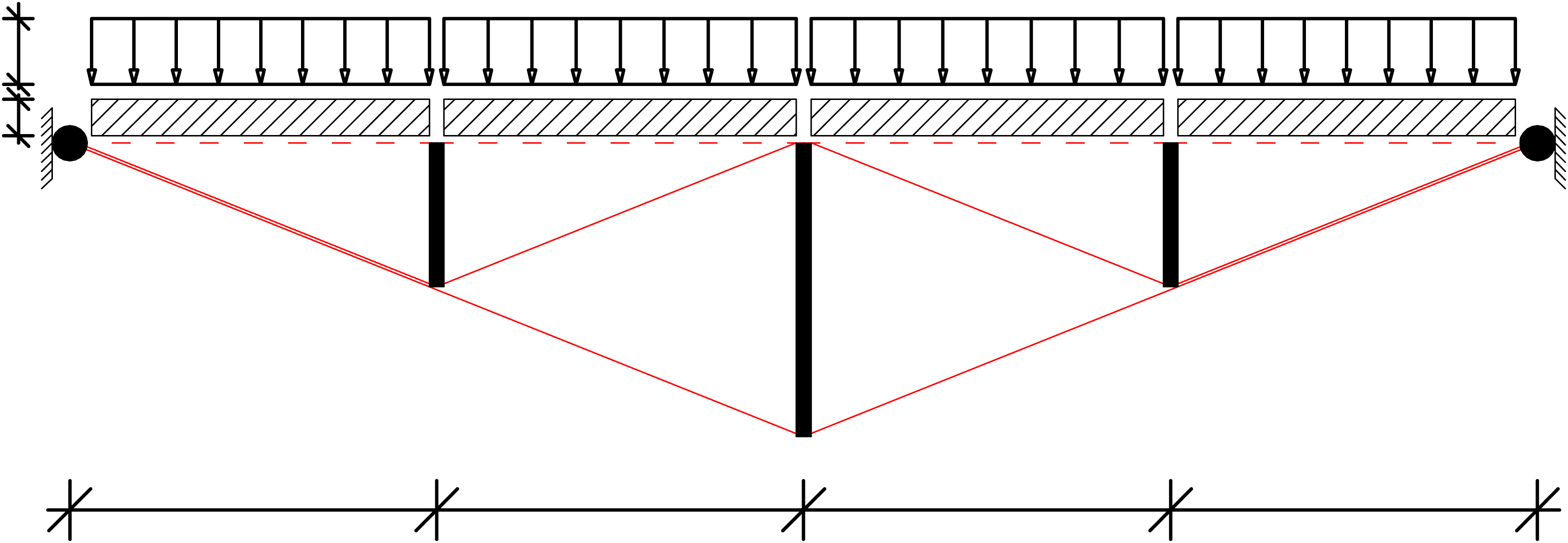,width=12cm}}
\put(5,0){\psfig{figure=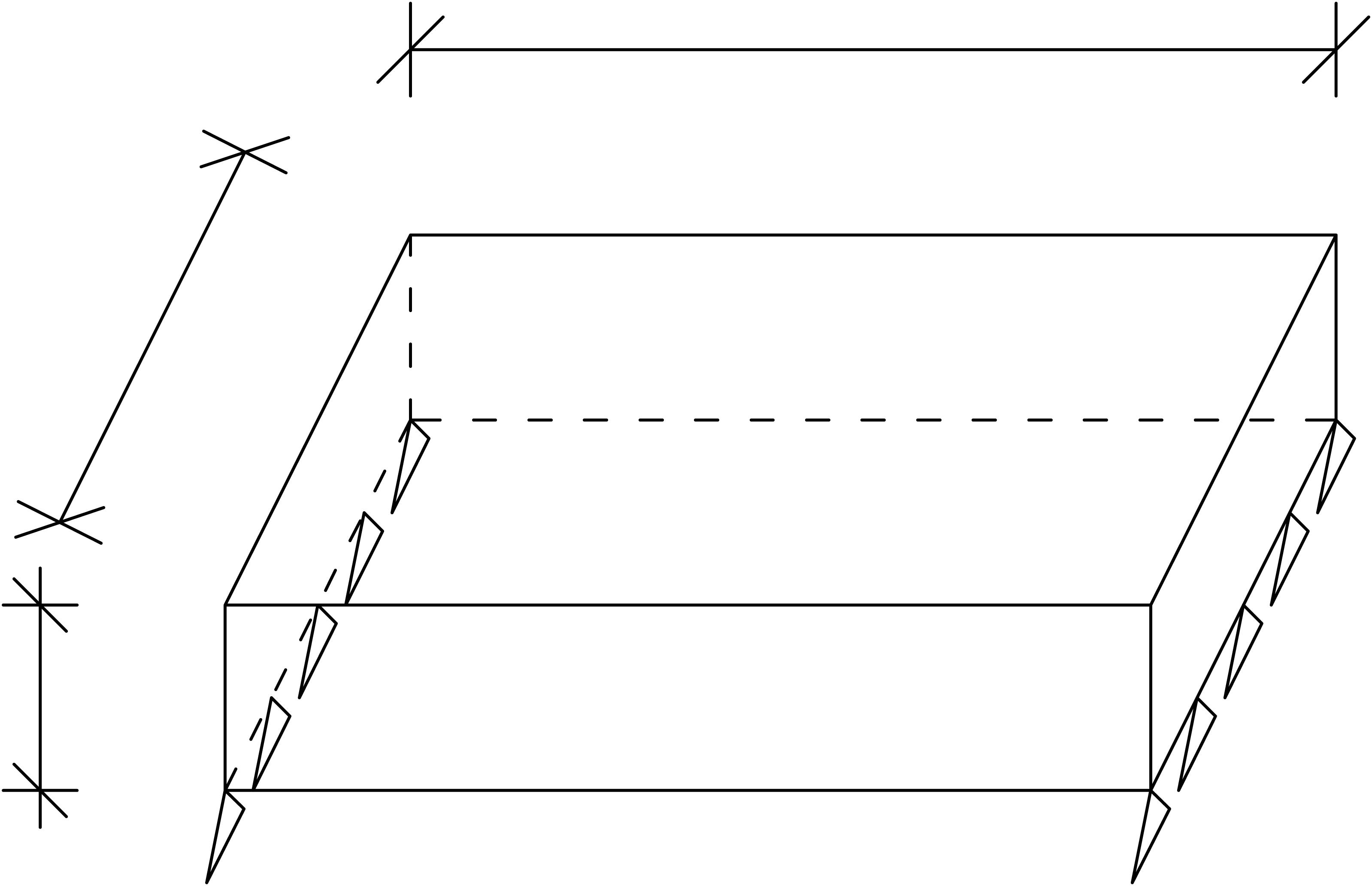,width=5cm}}
\put(1,8.5){{$\mbox{$a)$}$}}
\put(1,3.5){{$\mbox{$b)$}$}}
\put(1.5,7.75){{$\mbox{$f_d$}$}}
\put(1.5,7.15){{$\mbox{$t_d$}$}}
\put(3.75,4.45){{$\mbox{$\ell_d$}$}}
\put(6.75,4.45){{$\mbox{$\ell_d$}$}}
\put(9.5,4.45){{$\mbox{$\ell_d$}$}}
\put(12.25,4.45){{$\mbox{$\ell_d$}$}}
\put(8,3.25){{$\mbox{$\ell_d$}$}}
\put(4.5,0.65){{$\mbox{$t_d$}$}}
\put(4.9,2){{$\mbox{$w_d$}$}}
\put(7.4,1.25){{$\mbox{$m_d,$ $\rho_d$}$}}
\end{picture}
\caption{a) schematic deck system for a \emph{substructure} with complexity $n = 3$ and $p = 1$. b) detail of a single deck module.}
\label{fig:deck}
\end{figure}
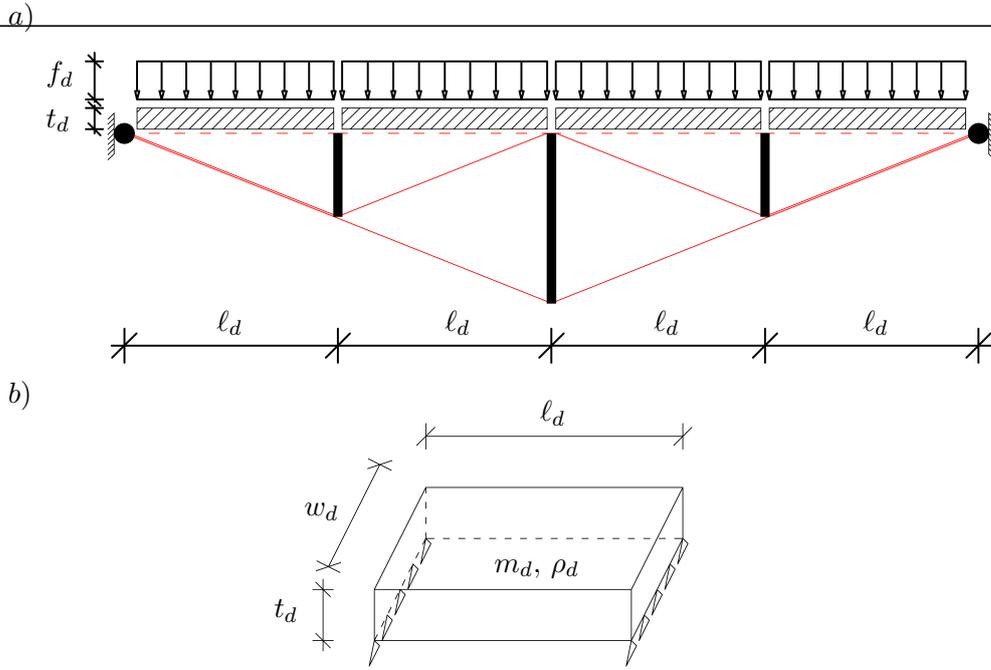

\noindent Assuming that the beam of a single deck section is simply supported between two consecutive nodes of the bridge, the maximum bending moment is equal to $f_d \ell_d^2 / 8$ and the maximum stress is given by Navier's equation \citep{Gere1997}:

\bea
\sigma_d = \frac{3}{4} \frac{f_d \ \ell_d^2}{w_d \ t_d^2}.
\label{sigma_deck}
\eea

\noindent The thickness of the deck beam is:

\bea
t_d = \frac{m_d}{\rho_d \ w_d \ \ell_d}.
\label{t_deck}
\eea

Substituting (\ref{ell_deck}), (\ref{f_deck}) and (\ref{t_deck}) into (\ref{sigma_deck}) we get the following equation for the mass of one deck section:

\bea
m_d = \frac{c_1}{2^{3 n}} + \frac{c_1}{2^{2 n}} \sqrt{c_2 + \frac{1}{2^{2 n}}},
\label{m_deck}
\eea

\noindent where:

\bea
c_1 = \frac{3 \ w_d \ g \ \rho_d^2 \ L^3}{8 \ \sigma_d}, \ \ \ 
c_2 = \frac{16 \ \sigma_d \ F}{3 \ w_d \ g^2 \ L^3 \ \rho_d^2}.
\label{c_1_2}
\eea

\noindent Then, the normalized total mass of the deck structure is:

 \bea
\mu_d^* = \frac{2^n \ m_d}{\left( \rho_s / \sigma_s \right) F L}.
\label{mu_d}
\eea

\noindent The total force acting on each internal node on the deck is then the sum of the force due to the external loads and the force due to the deck:

\bea
F_{tot} = F + 2^n \ m_d \ g.
\label{F_tot_deck}
\eea

\subsection{Adding Deck Mass for A \emph{Substructure} Bridge with Complexity $(n,p,q)=(n,1,0)$} \label{np1_deck_sub}

In this case, we make use of the notation illustrated in Fig. \ref{theorems_module_sub} in which complexity $p$ is fixed to be one. Complexity $n$ is defined to be the number of self-similar iterations of the basic module of Fig. \ref{basic_module}c.
Each iteration $n=1,2,...$ generates different lengths of bars and cables. The lengths at the $i^{th}$ iteration are:

\bea
b_i =\frac{L}{2^{i} } \tan \beta,\qquad i=1-n,
\label{lbars}
\eea

\bea
s_i =\frac{L}{2^{i} \cos \beta },\qquad i=1-n.
\label{lstrings}
\eea

\noindent Observing the multiscale structure of Fig. \ref{theorems_module_sub} it's clear that the number of bars and the number of cables at the $i^{th}$ self-similar iteration are

\bea
n_{si} = 2^i, \ \ \ n_{bi} = 2^{i-1}.
\label{n_bars_n_strings_ith}
\eea

\noindent In this case the total force applied to the bridge structure is given by (\ref{F_tot_deck}) and then the forces in each member become:

\bea
f_{bi} =\frac{F + 2^n m_d g}{2^{i} }, \ \ \ t_{si} =\frac{F  + 2^n m_d g}{2^{\left(1+i\right)} \sin \beta }.
\label{np1_forces_deck}
\eea

\begin{theorem} \label{theo:min_yielding_p1_deck}

Consider a \emph{substructure} bridge with deck mass $m_d$ and topology defined by (\ref{nnodes}), (\ref{nconn_sub}), (\ref{lbars}) and (\ref{lstrings}), with complexity $(n,p,q)=(n,1,0)$, see Fig. \ref{theorems_module_sub}. The minimal mass design under yielding constraints is given by:

\bea
\mu_{Y}^{*} = \left(1-\frac{1}{2^{n} } \right) \left( 1 + 2^n g \frac{m_d}{F} \right) \sqrt{1+\rho},
\label{mustY_np1_deck}
\eea

\noindent using the optimal angle:

\bea
\beta^*_{Y} = \arctan{ \left(\frac{1}{\sqrt{1+\rho } } \right)}.
\label{np1_betastY}
\eea

\end{theorem}

\begin{proof}

\noindent Assuming (\ref{lbars}) and (\ref{lstrings}) for the length of each member, (\ref{np1_forces_deck}) for the forces of each member, and (\ref{n_bars_n_strings_ith}) for the number of members, the dimensionless minimal mass becomes:

\bea
\mu_Y = \frac{1}{2} \left( 1 + 2^n g \frac{m_d}{F} \right) \left[ \frac{1}{\sin \beta \cos \beta} + \rho \tan \beta \right] \left( \sum_{i=1}^n {\frac{1}{2^i}} \right).
\label{np1_muY_deck_1}
\eea

\noindent yielding,

\bea
\mu_Y = \frac{1}{2} \left( 1 - \frac{1}{2^{n}} \right) \left( 1 + 2^n g \frac{m_d}{F} \right) \left[ \frac{\left(1 + \tan^2 \beta\right)}{\tan \beta} + \rho \tan \beta \right].
\label{np1_muY_deck_2}
\eea

\noindent The solution for minimal mass can be achieved from,

\bea
\frac{\partial \mu_{Y}}{\partial \tan \beta}= \frac{1}{2} \left( 1 - \frac{1}{2^n} \right) \left( 1 + 2^n g \frac{m_d}{F} \right) \left[ - \frac{\left( \tan^2 \beta +1 \right) }{\tan^2 \beta} + 2 + \rho \right] = 0,
\label{np1:dmuy}
\eea

\noindent yielding the optimal angle of (\ref{np1_betastY}). Substituting it into (\ref{np1_muY_deck_1})  concludes the proof.

\end{proof}

Observe that (\ref{mustY_np1_deck}) yields mass $\sqrt{1+\rho}/2$ for complexity $n=1$ and mass $\sqrt{1+\rho}$ for complexity $n=\infty$.
Note from (\ref{np1_betastY}), which is the same as (\ref{n1p1:betastY}), that the optimal angle $\beta^*_Y$ does not depend upon the choice of $n$. Indeed, the minimal mass solution under yielding constraints (\ref{mustY_np1_deck}) depends on the material choice $\rho$ (\ref{rho}), the complexity parameter $n$ and the deck properties. Note that, since the total external force $F$ is a specified constant, the mass is minimized by the complexity $n = 1$ if $m_d = 0$.  However since $m_d$ depends upon $n$, the total vertical force including deck mass depends upon $n$, and the optimal complexity will be shown to be $n>1$ in that case.

\begin{theorem} \label{theo:mim_buckling_np1_deck}

Consider a \emph{substructure} bridge with topology defined by (\ref{nnodes}), (\ref{nconn_sub}), (\ref{lbars}) and (\ref{lstrings}), with complexity $(n,p,q)=(n,1,0)$. The minimal mass design under yielding and buckling constraints is given by:

\bea
\mu_{B}^{*} = \beta_1 \frac{\left( 1+\tan^2 \beta_B^* \right)}{2 \tan \beta_B^*} + \eta \beta_2 \tan^2 \beta_B^*,
\label{np1_mustB_deck}
\eea

\noindent using the aspect angle:

\bea
\beta _{B}^{*} = \arctan{ \left\{ \frac{1}{12 \beta_2 \eta} \left[\beta_3 + \beta_1 \left( \frac{\beta_1}{\beta_3} - 1 \right) \right] \right\} },
\label{np1_betastB_deck}
\eea

\noindent where:

\bea
\beta_1 = \left( 1 - \frac{1}{2^n} \right) \left( 1 + 2^n g \frac{m_d}{F} \right),
\label{np1_beta_1}
\eea

\bea
\beta_2 = \left( \frac{1 + 2 \sqrt{2}}{7} \right) \left( 1 - \frac{1}{2^{3 n / 2}} \right) \sqrt{1 + 2^n g \frac{m_d}{F}},
\label{np1_beta_2}
\eea

\bea
\beta_3 = \left( 216 \beta_1 \beta_2^2 \eta^2 - \beta_1^3 + 12 \sqrt{324 \beta_1^2 \beta_2^4 \eta^4 - 3 \beta_1^4 \beta_2^2 \eta^2} \right)^{1/3}.
\label{np1_beta_3}
\eea

\end{theorem}

\begin{proof}

\noindent The total mass of the cables, using (\ref{lstrings}), (\ref{np1_forces_deck}) and (\ref{n_bars_n_strings_ith}), is given by:

\bea
\mu_s = \left( \frac{1 + \tan^2 \beta}{2 \tan \beta} \right) \left( 1 - \frac{1}{2^n} \right) \left( 1 + 2^n g \frac{m_d}{F} \right).
\label{np1_mus}
\eea

\noindent Similarly, making use of (\ref{m_b}), the total mass of bars is:

\bea
\mu_b = \eta \tan^2 \beta \left( \frac{1 + 2 \sqrt{2}}{7} \right) \left( 1 - \frac{1}{2^{3 n / 2}} \right) \sqrt{1 + 2^n g \frac{m_d}{F}} .
\label{np1_mub_1}
\eea

\noindent Introducing constants $\beta_1$ and $\beta_2$ given in (\ref{np1_beta_1}) and (\ref{np1_beta_2}), the total mass is:

\bea
\mu_B = \mu_s + \mu_b = \beta_1 \frac{\left( 1+\tan^2 \beta \right) }{2 \tan \beta} + \eta \beta_2 \tan^2 \beta.
\label{np1_muB_deck}
\eea

\noindent The solution for minimal mass can be achieved from,

\bea
\frac{\partial \mu_{B}}{\partial \tan \beta}= \beta_1 \left( 1 - \frac{1 - \tan^2 \beta}{2 \tan^2 \beta} \right) + 2 \eta \beta_2 \tan \beta = 0,
\label{np1:dmuB_deck}
\eea

\noindent yielding the optimal angle of (\ref{np1_betastB_deck}) by solving the following cubic equation:

\bea
4 \frac{\beta_2}{\beta_1} \eta \tan^3 \beta + \tan^2 \beta - 1 = 0.
\label{np1:dmuB_2_deck}
\eea

\noindent Substituting (\ref{np1_betastB_deck}) into (\ref{np1_muB_deck}) concludes the proof.

\end{proof}

\begin{figure}[hb]
\unitlength1cm
\begin{picture}(10,5.25)
\put(0.25,1.25){\psfig{figure=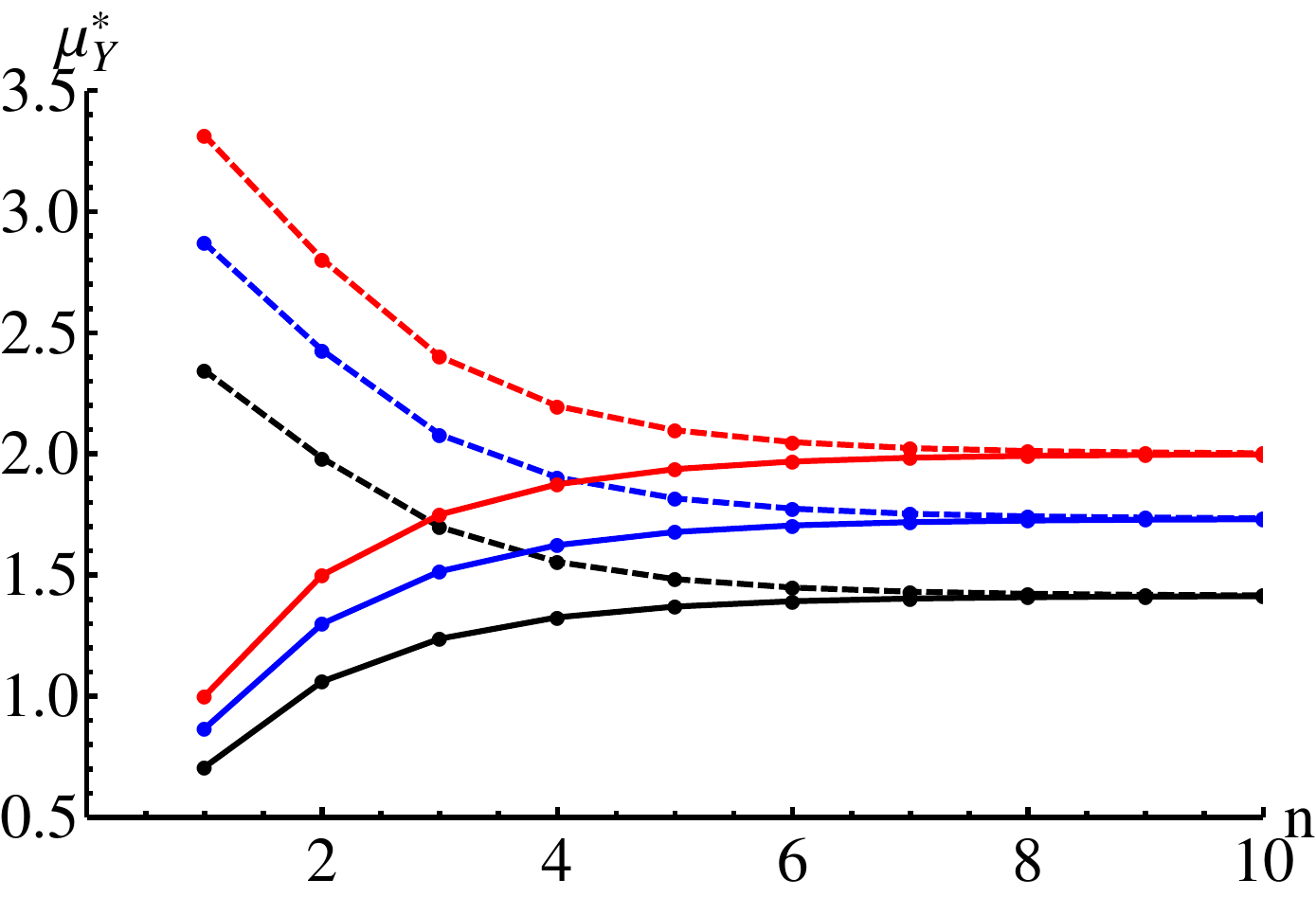,height=4.5cm}}
\put(8.5,1.25){\psfig{figure=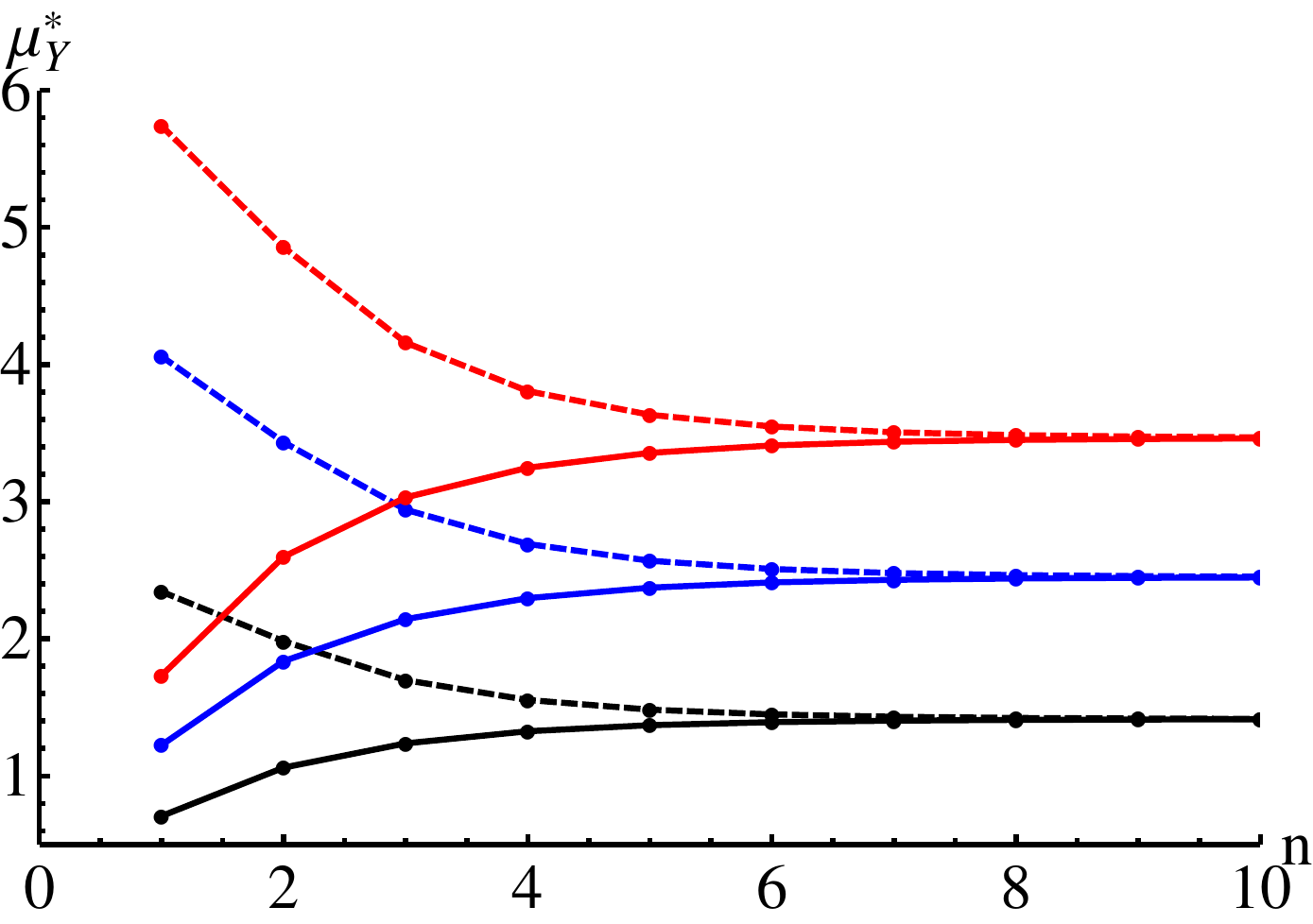,height=4.5cm}}
\put(3.5,-0.4){\psfig{figure=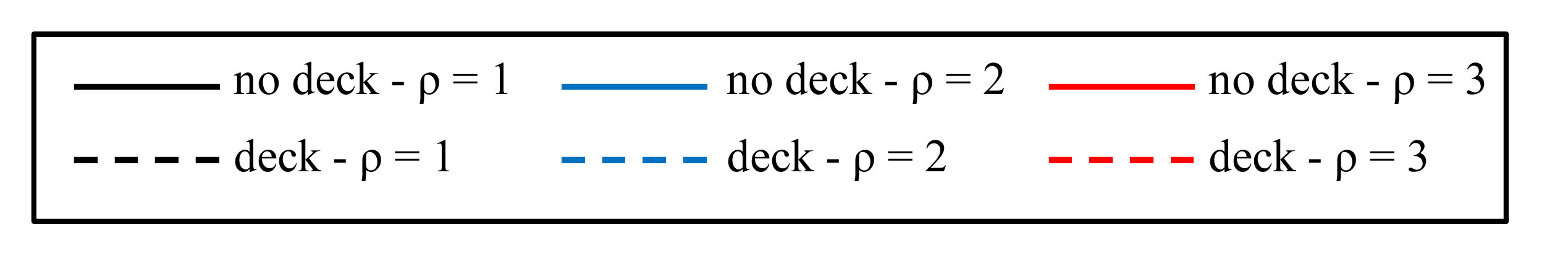,height=1.5cm}}
\end{picture}
\caption{Optimal masses under yielding of the \emph{substructures} (left) and \emph{superstructure} (right) without deck (solid curves) and with deck (dashed curves) for different values of the complexity $n$ and for different values of $\rho$, ($F = 1 \ N$, $w_d = 1 \ m$, steel deck).}
\label{fig:must_n_yielding}
\end{figure}

\begin{figure}[hb]
\unitlength1cm
\begin{picture}(10,5.25)
\put(0.25,1.25){\psfig{figure=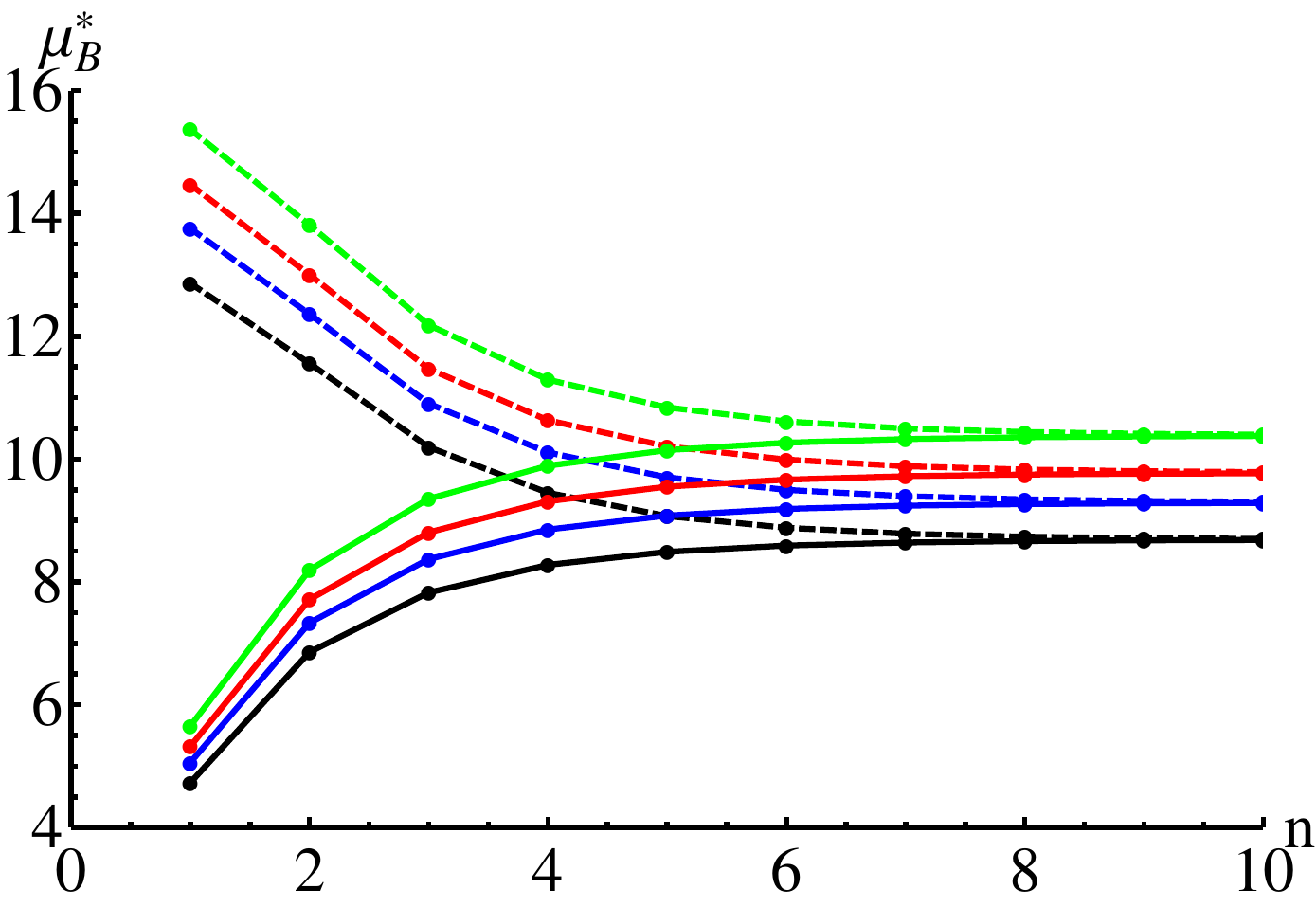,height=4.5cm}}
\put(8.5,1.25){\psfig{figure=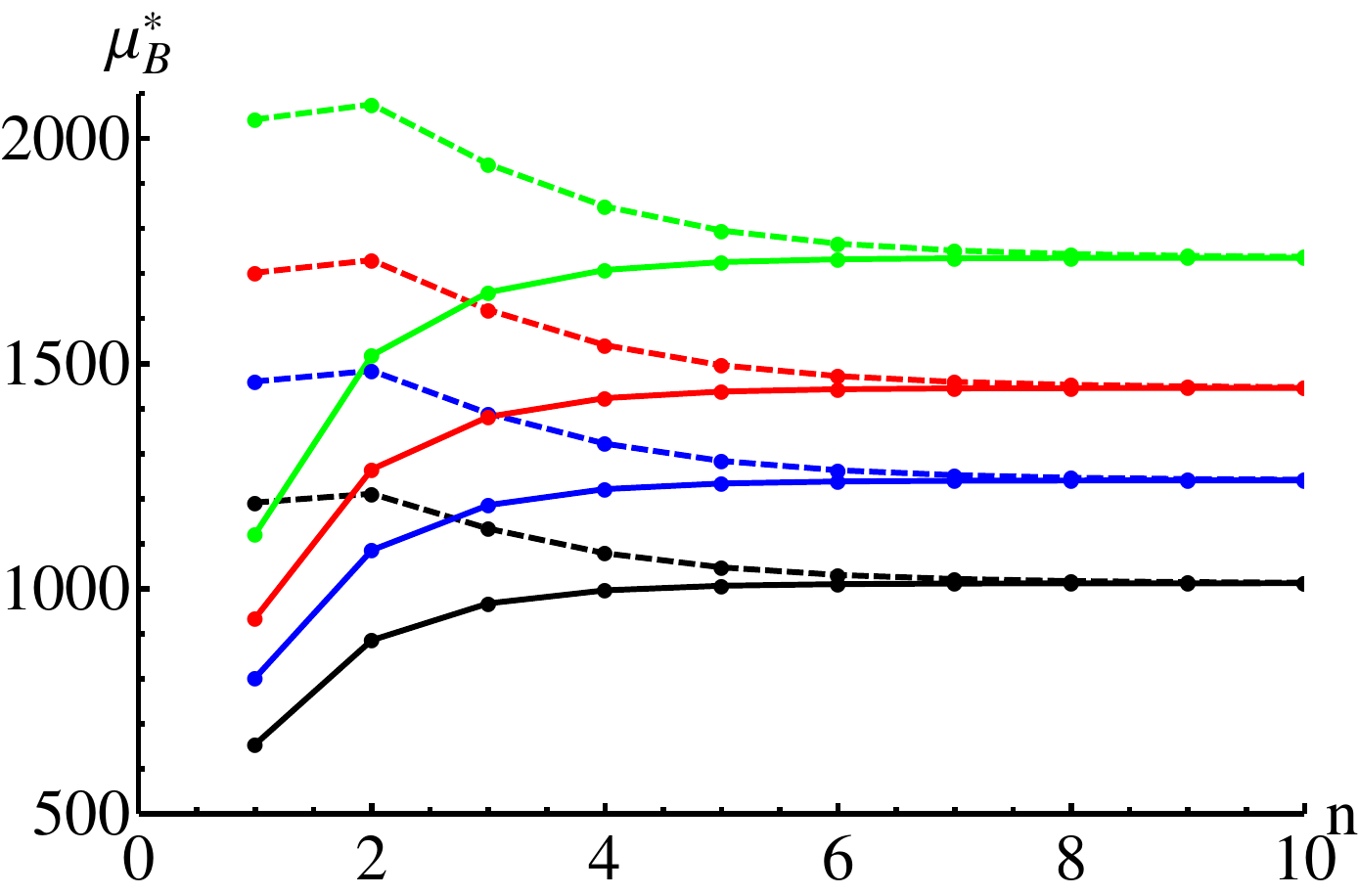,height=4.5cm}}
\put(2,-0.4){\psfig{figure=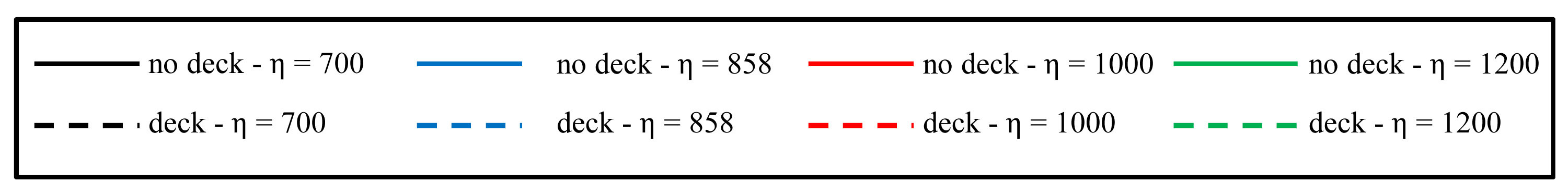,height=1.5cm}}
\end{picture}
\caption{Optimal masses under buckling of the \emph{substructures} (left) and \emph{superstructure} (right) without deck (solid curves) and with deck (dashed curves) for different values of the complexity $n$ and for different values of $\eta$, ($F = 1 \ N$, $L = w_d = 1 \ m$, steel deck).}
\label{fig:must_n_buckling}
\end{figure}

\subsection{Adding Deck Mass for A \emph{Superstructure} Bridge with Complexity $(n,p,q)=(n,0,1)$} \label{super_nq1_deck}

In this case, we make use of the notation illustrated in Fig. \ref{theorems_module_super} in which complexity $q$ is fixed to be one. Complexity $n$ is the number of self-similar iterations of the basic module of Fig. \ref{basic_module}b at different scales. After the $i^{th}$ self-similar iterations, the length of the bars and cables for $i$ ranging from $1$ to $n$, are:

\bea
b_i =\frac{L}{2^{i} \cos \alpha}, \ \ \ 
s_i = \frac{L}{2^i} \tan \alpha.
\label{nq1_lbars_lstrings}
\eea

\noindent Observing the multiscale structure of Fig. \ref{theorems_module_super} it's clear that the number of bars and the number of cables after the $i^{th}$ self-similar iterations are:

\bea
n_{si} = 2^{i-1}, \ \ \ n_{bi} = 2^{i}.
\label{nq1_n_bars_n_strings}
\eea

\noindent In this case the total force applied to the bridge structure is given by (\ref{F_tot_deck}) and then the forces in each member become:

\bea
f_{bi} =\frac{F + 2^n m_d g}{2^{\left(i + 1\right)} \sin \alpha }, \ \ \ t_{si} =\frac{F  + 2^n m_d g}{2^{i} }.
\label{nq1_forces_deck}
\eea

\begin{theorem} \label{theo:min_yielding_nq1_deck}
Consider a \emph{superstructure} bridge with topology defined by (\ref{nnodes}), (\ref{nconn_super}), (\ref{nq1_lbars_lstrings}), with complexity $(n,p,q)=(n,0,1)$, Fig. \ref{theorems_module_super}. Under a given total vertical force (\ref{F_tot_deck}), the minimal mass design under yielding constraints is given by:

\bea
\mu_{Y}^{*} = \left(1-\frac{1}{2^{n} } \right) \left( 1 + 2^n g \frac{m_d}{F} \right) \sqrt{ \rho \left(1+\rho \right) },
\label{nq1_mustY_deck}
\eea

\noindent using the aspect angle:

\bea
\alpha^*_{Y} = \arctan{ \left(\sqrt{\frac{\rho}{{1+\rho } }} \right)}.
\label{nq1_alphastY_deck}
\eea

\end{theorem}

\begin{proof}

\noindent Substituting (\ref{nq1_lbars_lstrings}), (\ref{nq1_forces_deck}), and (\ref{nq1_n_bars_n_strings}) into (\ref{nq1_mass}) and considering positions (\ref{np1_positions}) we get:

\bea
m_Y = \frac{\left( F + 2^n m_d g \right) L}{2} \left( 1 - \frac{1}{2^n} \right) \left[ \frac{\rho_s}{\sigma_s} \tan \alpha + \frac{\rho_b}{\sigma_b} \frac{\left( 1 + \tan^2 \alpha \right)}{\tan \alpha} \right].
\label{nq1_mass_deck}
\eea

\noindent Switching to the dimensionless mass defined in (\ref{mu}) we have:

\bea
\mu_Y = \frac{1}{2} \left( 1 + 2^n g \frac{m_d}{F} \right) \left( 1 - \frac{1}{2^n} \right) \left[ \tan \alpha + \rho \frac{\left( 1 + \tan^2 \alpha \right)}{\tan \alpha} \right].
\label{nq1_mu_deck}
\eea

\noindent The solution for minimal mass can be achieved from,

\bea
\frac{\partial \mu_{Y}}{\partial \tan \alpha}= \frac{1}{2} \left( 1 + 2^n g \frac{m_d}{F} \right) \left( 1 - \frac{1}{2^n} \right) \left[ 1 + \rho \left( 2 - \frac{1 + \tan^2 \alpha}{\tan^2 \alpha} \right) \right] = 0,
\label{nq1:dmuy_deck}
\eea

\noindent yielding the optimal angle of (\ref{nq1_alphastY_deck}). Substituting it into (\ref{nq1_mu_deck}) concludes the proof.

\end{proof}

\begin{theorem} \label{theo:min_buckling_nq1_deck}

Consider a \emph{superstructure} bridge with topology defined by (\ref{nnodes}), (\ref{nconn_super}), (\ref{nq1_lbars_lstrings}), and complexity $(n,p,q)=(n,0,1)$, see Fig. \ref{theorems_module_super}. The structure is loaded with a given total vertical force (\ref{F_tot_deck}) and the minimal bar mass, subject to yield constraints is given by:

\bea
\mu_{B}^{*} = \frac{\delta_1}{2} + \eta \delta_2 \frac{5^{5/4}}{4},
\label{nq1_mustB}
\eea

\noindent using the aspect angle:

\bea
\alpha _{B}^{*} = \arctan{ \frac{1}{2} },
\label{nq1_alphastB_deck}
\eea

\noindent where:

\bea
\delta_1 = \frac{1}{2} \left( 1 + 2^n g \frac{m_d}{F} \right) \left( 1 - \frac{1}{2^n} \right),
\label{nq1_delta_1}
\eea

\bea
\delta_2 = \sqrt{2} \left( \frac{1 + 2 \sqrt{2}}{7} \right) \sqrt{ 1 + 2^n g \frac{m_d}{F} } \left( 1 - \frac{1}{2^{3 n / 2}} \right).
\label{nq1_delta_2}
\eea

\end{theorem}

\begin{proof}

Substituting (\ref{nq1_lbars_lstrings}), (\ref{nq1_forces_deck}) and (\ref{nq1_n_bars_n_strings}) into (\ref{nq1_ms_1}) and making use of position (\ref{np1_positions}):

\bea
m_s = \frac{\left( F + 2^n m_d g \right) L}{2} \frac{\rho_s}{\sigma_s} \left( 1 - \frac{1}{2^n} \right) \tan \alpha.
\label{nq1_ms_deck}
\eea

\noindent That corresponds to the following normalized mass:

\bea
\mu_s = \frac{1}{2} \left( 1 + 2^n g \frac{m_d}{F} \right) \left( 1 - \frac{1}{2^n} \right) \tan \alpha.
\label{nq1_mus_deck}
\eea

\noindent Substituting (\ref{nq1_lbars_lstrings}), (\ref{nq1_forces_deck}) and (\ref{nq1_n_bars_n_strings}) into (\ref{nq1_mb_1}):

\bea
m_b = \frac{\sqrt{2} \rho_b L^2}{\sqrt{\pi E_b}} \frac{\sqrt{F + 2^n m_d g}}{\cos^2 \alpha \sqrt{\sin \alpha}} \sum_{i=1}^n {\frac{1}{2^{3 i/2}}}.
\label{nq1_mb_deck}
\eea

\noindent Using positions (\ref{nq1_position_1}) into (\ref{nq1_mb_deck}) and normalizing we get the following dimensionless mass of bars:

\bea
\mu_b = \sqrt{2} \eta \sqrt{1 + 2^n g \frac{m_d}{F}} \left( \frac{1 + 2 \sqrt{2}}{7} \right) \left( 1 - \frac{1}{2^{3 n / 2}} \right) \frac{\left( 1 + \tan^2 \alpha \right)^{5/4}}{\sqrt{\tan \alpha}}.
\label{nq1_mub_deck}
\eea

\noindent The total mass is then the sum of (\ref{nq1_mus_deck}) and (\ref{nq1_mub_deck}) and introducing constants $\delta_1$ and $\delta_2$ given in (\ref{nq1_delta_1}) and (\ref{nq1_delta_2}):

\bea
\mu_B = \mu_s + \mu_b = \delta_1 \tan \alpha + \eta \delta_2 \frac{\left( 1 + \tan^2 \alpha \right)^{5/4}}{\sqrt{\tan \alpha}}.
\label{nq1_muB_deck}
\eea

\noindent The solution for minimal mass can be achieved assuming that:

\bea
\delta_1 \tan \alpha \ll \eta \delta_2 \frac{\left( 1 + \tan^2 \alpha \right)^{5/4}}{\sqrt{\tan \alpha}}.
\label{nq1_hyp_deck}
\eea

\noindent So that the (\ref{nq1_muB_deck}) becomes:

\bea
\bar \mu_B = \eta \delta_2 \frac{\left( 1 + \tan^2 \alpha \right)^{5/4}}{\sqrt{\tan \alpha}}.
\label{nq1_muB_bar_deck}
\eea

\noindent The optimal angle can be obtained from:

\bea
\frac{\partial \bar \mu_{B}}{\partial \tan \alpha}= \frac{\eta}{2} \delta_2 \left( 1 + \tan^2 \alpha \right)^{\left(1/4\right)} \left( \frac{4 \tan^2 \alpha - 1}{\tan \alpha \sqrt{\tan \alpha}} \right) = 0,
\label{np1:dmuB_deck}
\eea

\noindent yielding the optimal angle of (\ref{nq1_alphastB_deck}). Substituting it into (\ref{nq1_muB_deck}) concludes the proof.

\end{proof}

Fig. \ref{fig:must_n_yielding} for yielding and Fig. \ref{fig:must_n_buckling} for buckling show as the theorems obtained in this section can be applied to compute the optimal mass of \emph{substructure} or \emph{superstructure} for any choice of the parameter $\rho$ (for yielding) or $\eta$ (for buckling) . We obtained that, with the addition of deck mass to the design, the optimal complexity $n$ becomes greater then 1. In the next section we will show the effect of the addition of joint mass.

\begin{figure}[hb]
\unitlength1cm
\begin{picture}(10,4.5)
\put(3.5,0){\psfig{figure=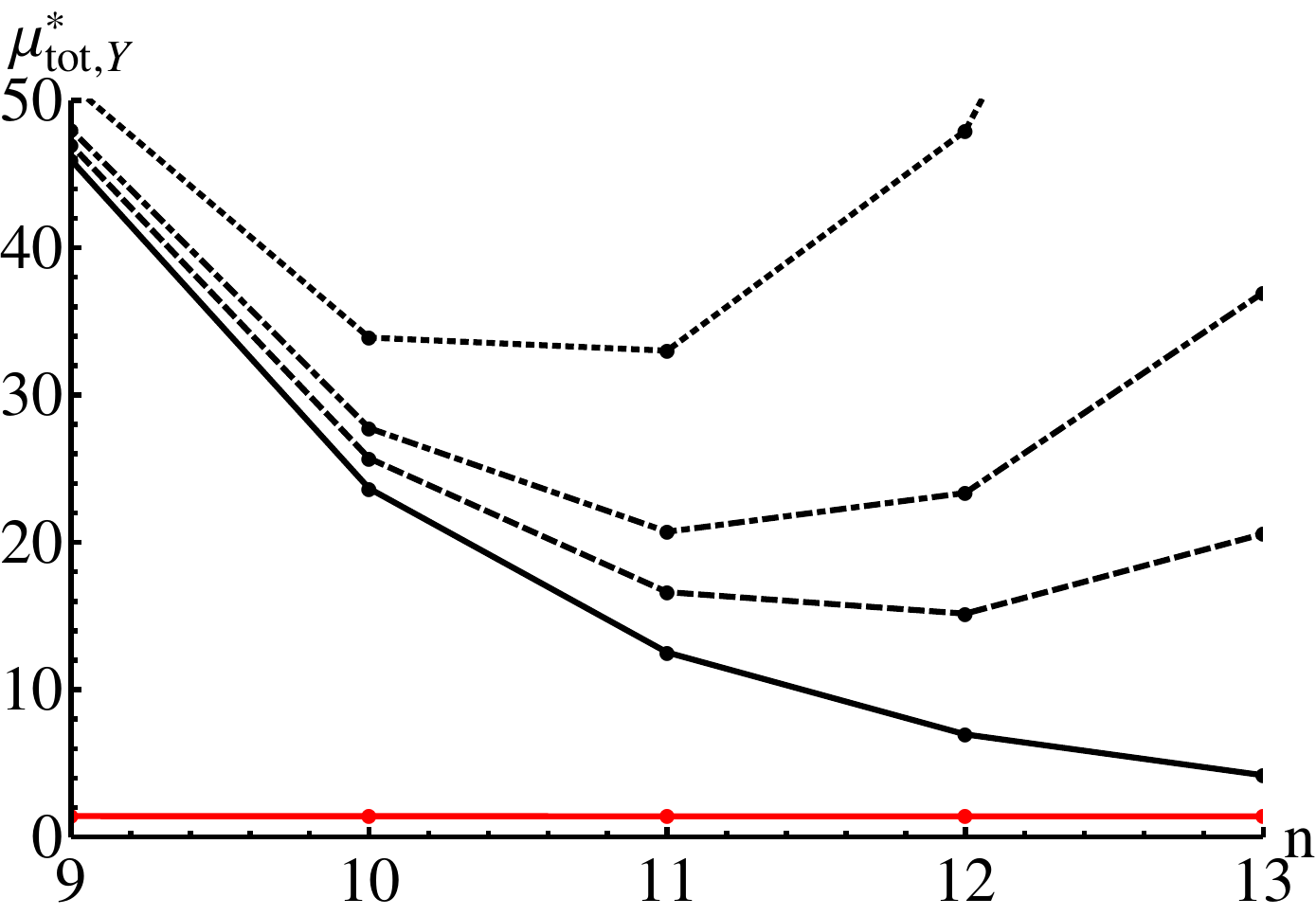,height=4.5cm}}
\put(10.25,1.5){\psfig{figure=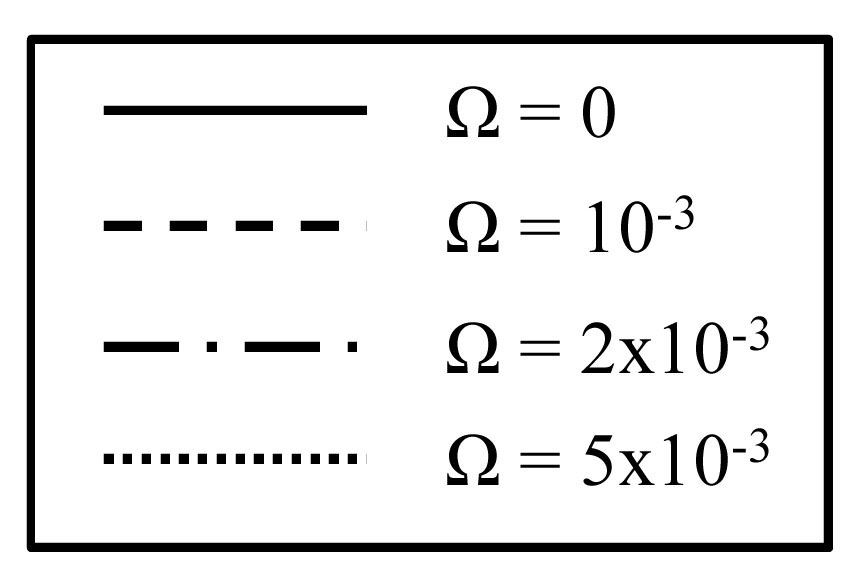,height=2cm}}
\end{picture}
\caption{Optimal masses under yielding of the substructures (left) and superstructure (right) (red curve) and total optimal mass with deck and different joint factors (dashed and dottled curves) for different values of the complexity $n$ (steel for bars, cables, deck, $F = 1 \ N$, $L = w_d = 1 \ m$).}
\label{fig:musttot_n_yielding}
\end{figure}

\begin{figure}[hb]
\unitlength1cm
\begin{picture}(10,5)
\put(0.25,0.5){\psfig{figure=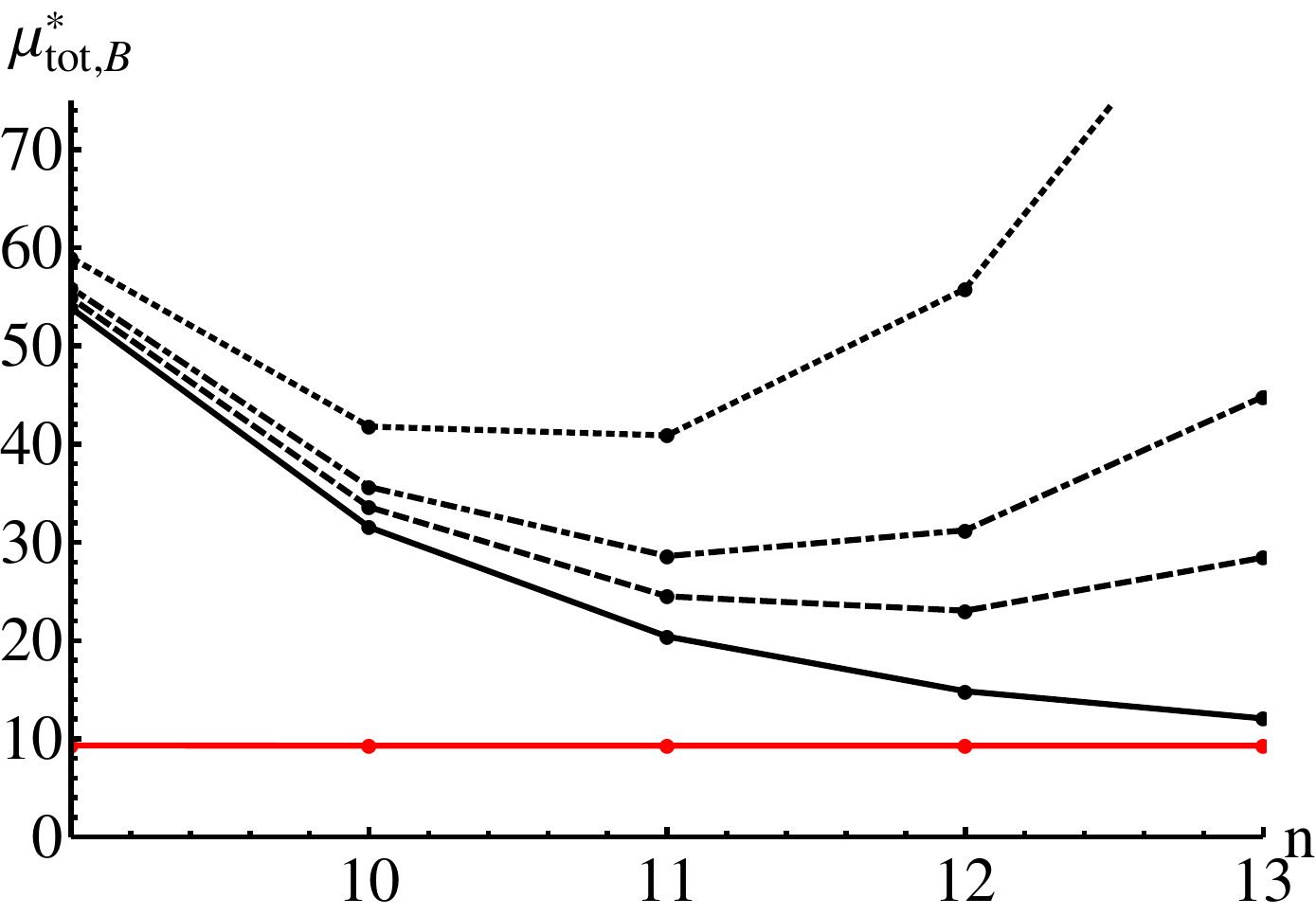,height=4.5cm}}
\put(8.5,0.5){\psfig{figure=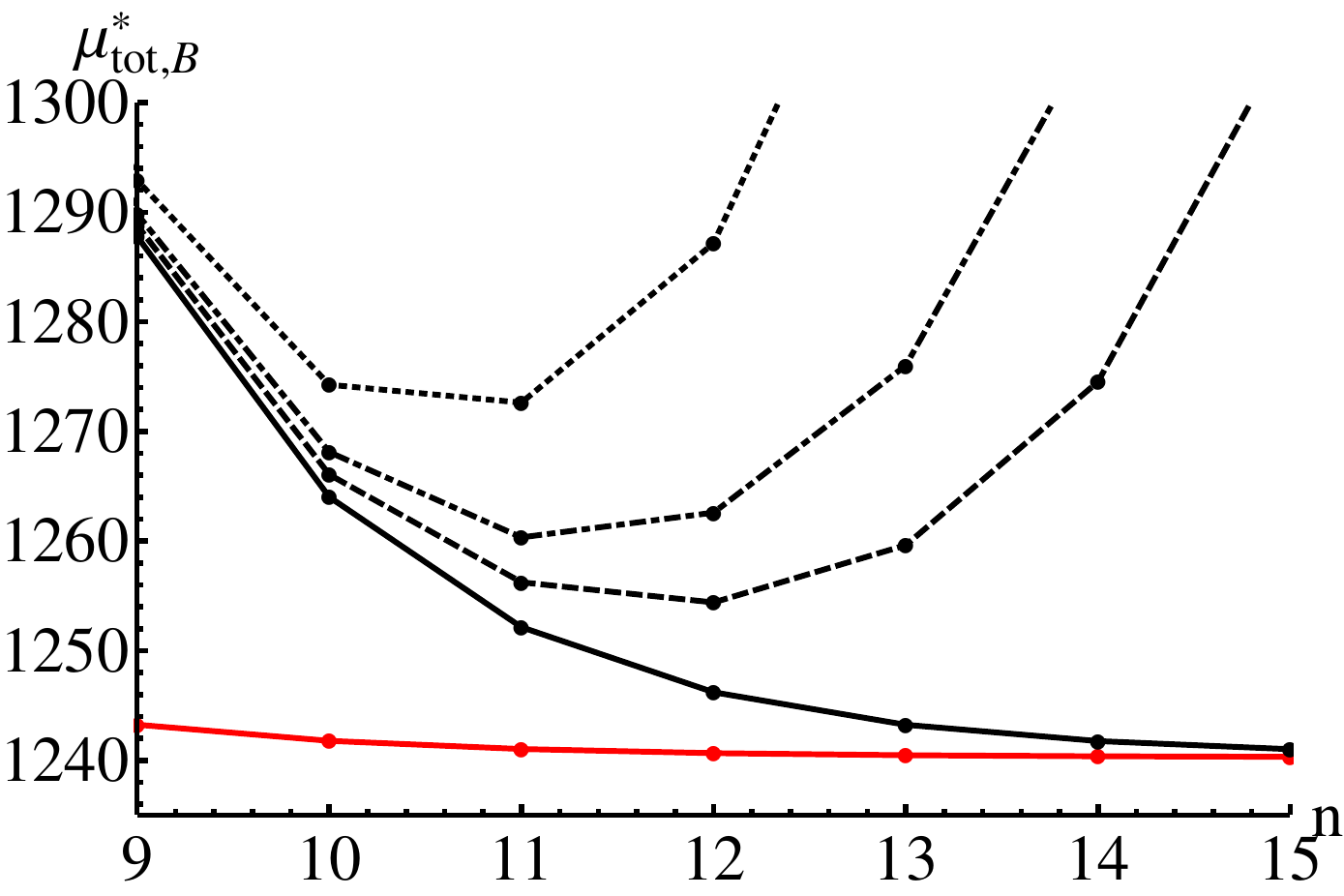,height=4.5cm}}
\put(3.5,-0.35){\psfig{figure=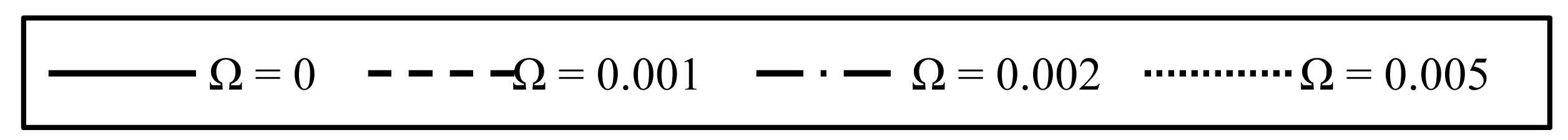,height=0.8cm}}
\end{picture}
\caption{Optimal masses under buckling of the substructures (left) and superstructure (right) (red curves) and total optimal masses with deck and different joint factors (dashed and dotted curves) for different values of the complexity $n$ (steel for bars, cables, deck, $F = 1 \ N$, $L = w_d = 1 \ m$).}
\label{fig:musttot_n_buckling}
\end{figure}

\subsection{Penalizing Complexity with cost considerations: Adding Joint Mass} 
\label{joint_mass}

Theorem \ref{theo:min_yielding_p1_deck}, for $m_d = 0$, leads to an optimal complexity $n = 1$ which corresponds to a minimal mass equal to $\sqrt{1+\rho}/2$. As complexity $n$ approaches infinity, instead, the mass given in (\ref{mustY_np1_deck}), for $m_d = 0$, go to a limit equal to $\sqrt{1+\rho}$. However, the addition of the deck mass in Theorem \ref{theo:min_yielding_p1_deck} switches the optimal complexity from $n=1$ to $n=\infty$, so small complexities $n$ are penalized by massive decks. Also in this latter case, the resulting optimal minimal mass is then $\sqrt{1+\rho}$, as can be verified looking the (\ref{mustY_np1_deck}) or considering that as $n$ goes to infinity the deck mass given in (\ref{m_deck}) approaches zero. As a matter of fact, neither $n = 1$ or $n = \infty$ are believable solutions due to practical reasons: the first solution leads only to a single force at the middle of the span, the second solution leads to an infinite number of joints and connections. The minimal masses obtained from (\ref{mustY_np1_deck}) with or without deck correspond to perfect massless joints. The addition of the joint masses to a tensegrity structure with $n_n$ nodes, as illustrated in \cite{Skelton2010c}, leads to the following total normalized mass:

\bea
\mu^*_{Y,tot} = \mu^*_Y + \mu^*_d + \Omega n_n,
\label{mu_tot_deck_joint}
\eea

Let $\$_j$ be the cost per $kg$ of making joints and let $\$_b$ be the cost per $kg$ of making bars. Then define $\Omega=\$_b/\$_j$. For perfect joints $\Omega=0$, for rudely made low cost joints $\$_j$ is small and $\Omega$ is larger. Hence $\Omega$ is also approximatively the ratio of material cost per joint divided by material cost per structural member being joined.

Consider the minimal masses of the \emph{substructure} bridge ($\mu^*_Y$) constrained against yielding, for the cases with or without deck, see Eq. (\ref{mustY_np1_deck}). Assume steel material for cables, bars and deck beams and set $F = 1 \ N$, $L = w_d = 1 \ m$. Without deck the optimal aspect angle $\beta_Y^*$ (\ref{np1_betastY}) is $35.26$ $deg$. For the case with neither deck nor joint mass, the optimum complexity $n$ is $1$, which corresponds to an optimal mass $\mu^*_Y = \sqrt{2}/2$. As $n$ approaches infinity the mass tends to a limit equal to $\sqrt{2}$, which is also the optimal mass for the case with deck mass and perfectly manufactured joints, since $\mu^*_d$ approaches zero for $n \rightarrow \infty$. Note that with the addition of joint masses as illustrated in (\ref{mu_tot_deck_joint}), the optimal complexity $n^*$ can become a finite value. The above procedure can be also used for the design under buckling constraints.

Figs. \ref{fig:musttot_n_yielding} (for yielding) and Fig. \ref{fig:musttot_n_buckling} (for buckling) show the total minimal masses obtained by using (\ref{mu_tot_deck_joint}). In both Figs. \ref{fig:musttot_n_yielding} and \ref{fig:musttot_n_buckling} we also show with red curves the minimal mass of \emph{substructures} or \emph{superstructures} only. In either case, the total mass of the structure with deck (but no joint mass), is shown by black continuous lines in Figs. \ref{fig:musttot_n_yielding} and \ref{fig:musttot_n_buckling}, reaching minimum for an infinite complexity $n$. It is worth nothing that, for infinite $n$, the mass of the deck is zero and the total minimum mass is just the mass of the bridge structure. Then, with the dotted and dashed lines, we show that a finite optimal complexity can be achieved if the joint's masses are considered.

From Fig. \ref{fig:musttot_n_yielding} note that the minimal mass ($\mu \cong 21$) bridge has complexity $n = 11$ for $\Omega=0.002$, and has minimal mass $\mu \cong 15$ with complexity $n=12$ for $\Omega=0.001$. Economic costs would decide if saving $25 \ \%$ structural mass is worth the extra cost of improving the joint precision by a factor of $2$.

\section{Numerical results without deck} \label{results}

In this section we show the minimal masses and the optimal angles of tensegrity bridges with several complexities $n$, $p$ and $q$. The numerical results are presented in terms of $\mu^*_B$, $\alpha^*_B$ and $\beta^*_B$ denoting respectively the minimal masses and the optimal aspect angles under combined yielding and buckling constraints for each bar and yielding constraints for each cable; and in terms of $\mu^*_Y$, $\alpha^*_Y$ and $\beta^*_Y$ denoting respectively the minimal masses and the optimal aspect angles under yielding constraints only for each member. The results are obtained numerically through a MatLab\textsuperscript \textregistered \ program written employing the algorithm illustrated in Sect. 3 of \cite{Skelton2014}. The optimization problems presented in Tables \ref{nr_combined_n1p}, \ref{nr_combined_np}, \ref{nr_sub_n1p}, \ref{nr_sub_np}, \ref{nr_super_n1q}, \ref{nr_super_nq},  are solved assuming $L = 1 \ m$, $F = 1 \ N$, no deck mass, and steel for both cables and bars (refer to Table \ref{tab_materials} for the material properties; $\rho = 1$; $\eta = 857.71$). The examined topologies are distinguished in three categories: 1) \emph{nominal} bridges with both structure above and below the roadway (Sect. \ref{numerical_nominal}; Fig \ref{fig:dual_exemplaries}); 2) \emph{substructure} only bridges (Sect. \ref{numerical_sub}; Fig. \ref{fig:superstructures_exemplaries}) and 3) \emph{superstructure} only bridges (Sect. \ref{numerical_super}; Fig. \ref{fig:substructures_exemplaries}). In all the optimized cases, we set step increments of complexities $n$, $p$ and $q$ to $1$ and step increments of $0.01$ $deg$ for the aspect angles $\alpha$ and $\beta$. It is worth noting that, as showed for the basic module (Fig. \ref{basic_module}) analyzed in Sect. \ref{n1q1p1_model}, the cables placed on the deck have zero mass at the solution for minimal mass basically thanks to the adopted constraints ($HH$: double fixed hinges). Appendix A reports some numerical results for rolling hinge at one end of the bridge and fixed hinge at the other end ($HR$). We also report, for each optimized structure, the masses of cables under buckling constraints ($\mu^*_{B,s}$) to show, as will be more clear in the following, as their order of magnitude with respect to the mass total mass of the structure ($\mu^*_{B}$) increase towards the global optimum. In other words, the principal source of mass savings of a tensegrity structure for buckling is placed in the mass of bars.

\begin{table}[tb]
\begin{center}
\begin{tabular}{cc} \hline \hline 
\multicolumn{2}{c}{\textit{steel}} \\ \hline \hline 
$\rho$ $[kg/m^{3}]$ & 7862 \\ \hline 
$\sigma$ $[N/m^{2}]$ & 6.9 x 10${}^{8}$ \\ \hline 
$E$ $[N/m^{2}]$ & 2.06 x 10${}^{11}$ \\ \hline \hline 
\multicolumn{2}{c}{\textit{Spectra\textsuperscript \textregistered - UHMWPE}} \\ \hline \hline
$\rho$ $[kg/m^{3}]$ & $970$ \\ \hline 
$\sigma$ $[N/m^{2}]$ & $2.7x10^9$ \\ \hline 
$E$ $[N/m^{2}]$ & $120x10^9$ \\ \hline \hline
\end{tabular}
\caption{Material properties.}
\label{tab_materials}
\end{center}
\end{table}

\subsection{Nominal Bridges} \label{numerical_nominal}

We have performed several numerical results for the \emph{nominal bridges}, illustrated in Fig. \ref{fig:dual_exemplaries}, in which both structure above and below the roadway are allowed. Starting from the basic unit in Fig. \ref{basic_module}a, we have considered different complexities $n$, $p$ and $q$  and different aspect angles $\alpha$ and $\beta$, in order to get the combination of such parameters that ensures the minimal mass solution. First of all, we start fixing parameter $n$ to unity and let parameters $p$, $\alpha$ and $\beta$ ranging in the following intervals:

\beq
(p, \ q) \in [1, 100],  \ \ \
\alpha \in (0, 90) \ \mbox{deg}, \ \ \
\beta \in (0, 90) \ \mbox{deg}.
\label{domain1}
\eeq

The results of this first design are presented in Table \ref{nr_combined_n1p}. For what concern the design under only yielding constraints for all members, the global minimum is achieved for a complexity $q^*_Y = p^*_Y \rightarrow  \infty$ and for aspect angles $\alpha^*_Y = \beta^*_Y \rightarrow 45 \ deg$, which corresponds to a minimal mass $\mu^*_Y \rightarrow 0.6427$ (Table \ref{nr_combined_n1p}). Such a result confirms the minimal mass solution for a centrally loaded loaded beam reported in Fig. 2 by \cite{Michell1904}. In particular, for a beam of total span $2a_M$ loaded in the middle with a force $F_M$ and made of tensile and compressive members with allowable yielding stresses equal respectively to $P$  and $Q$; \cite{Michell1904} predicted a limit volume equal to:

\beq
v_M = F_M a_M \left( \frac{1}{2} + \frac{\pi}{4} \right) \left( \frac{1}{P} + \frac{1}{Q} \right).
\label{Michell_limit}
\eeq

Substituting in the (\ref{Michell_limit}), as in the present case, $F_M = F/2$, $a_M = L/2 \ m$, $P = Q = 6.9x10^8 \ N/m^2$, we obtain a volume $v_M = 9.31448x10^{-10} \ m^3$. On the other hand, the minimal mass $\mu^*_Y \rightarrow 0.6427$ corresponds to a volume $v^*_Y = 0.6427 {F L}/{\sigma_s} =  9.31449x10^{-10} m^3$. We will show in the next Sects. \ref{numerical_sub} and \ref{numerical_super} that the same minimal mass can be achieved also starting from \emph{superstructure only} bridges showed in Fig. \ref{fig:substructures_exemplaries} and \emph{superstructure} bridges showed in Fig. \ref{fig:superstructures_exemplaries}. The equivalence between \emph{substructure} and \emph{superstructure} under yielding constraints can be justified by the assumption of bars and cables made of the same materials ($\rho = 1$). An example of this equivalence can be obtained assuming, eg., $\rho = 1$ in the Eq. (\ref{np1_mustY}) of Theorem \ref{theo:min_yielding_np1} for \emph{substructures} with complexity $(n,p)= (n,1)$ and in the Eq. (\ref{nq1_mustY}) of Theorem \ref{theo:min_yielding_nq1} for \emph{superstructures} with  complexity $(n,q)=(n,1)$.

Under buckling constraints, the global optimum in the domain (\ref{domain1}) is achieved for a finite complexity $p^*_B = 1$, which corresponds to a minimal mass of $\mu^*_B = 5.0574$ and an aspect angle of the substructure equal to $\beta^*_B = 4.25 \ deg$. In all the combined cases under buckling, we have obtained that the optimal solutions keep only the substructure and the total mass of the superstructure is negligible if compared with the total mass. In fact, for the global optimum with $q^*_B = p^*_B = 1$ and $\beta^*_B = 4.25 \ deg$, we have obtained a mass of the superstructure equal to $4.3983x10^{-10}$. For the other cases, we have obtained similar negligible values of the mass of superstructures ranging from a minimum of $1.9019x10^{-12}$ for the case with $q = p = 45$ and a maximum of $3.1483x10^{-6}$ for the case with $p = q = 10$. The analyzed cases of domain (\ref{domain1}), then, reduce to the substructure only cases for buckling. As a matter of fact, the angles $\beta^*_B$ decrease from $4.25 \ deg$ to $1.98 \ deg$ as the complexities $q = p$ increase from $1$ to $100$. The reduction of $\beta^*_B$ corresponds to an increase of the tensile forces of the cables constituting the substructure and, consequently, also the total mass of cables (indicated with $\mu^*_{B,s}$ in Table \ref{nr_combined_n1p}) increases. In other words, for the combined bridges under buckling or, equivalently, for the substructures bridges under buckling, the total mass of the cables is the big part of the total mass of the structure.

\begin{table}[tb]
\begin{center}
\begin{tabular}{ccccccccc} \hline \hline
$n$ & $p$ & $\alpha^{*}_{Y}$ [deg] & $\beta^{*}_{Y}$ [deg] & $\mu^{*}_{Y}$ & $\alpha^{*}_{B}$ [deg] & $\beta^{*}_{B}$ [deg] & $\mu^{*}_{B}$ & $\mu^{*}_{B,s}$ \\ \hline \hline
1 & 1 & 35.26 & 35.26 & 0.7071 & - & 4.25 & 5.0574 & 3.3827 \\ \hline 
1 & 2 & 35.26 & 35.26 & 0.7071 & - & 3.80 & 5.6662 & 3.7805 \\ \hline
1 & 3 & 41.41 & 41.41 & 0.6614 & - & 3.57 & 6.0309 & 4.0227 \\ \hline
1 & 4 & 43.23 & 43.23 & 0.6514 & - & 3.40 & 6.3260 & 4.2228 \\ \hline
1 & 5 & 43.96 & 43.96 & 0.6476 & - & 3.27 & 6.5683 & 4.3899 \\ \hline
1 & 10 & 44.78 & 44.78 & 0.6437 & - & 2.91 & 7.3839 & 4.9308 \\ \hline
1 & 15 & 44.91 & 44.91 & 0.6431 & - & 2.72 & 7.9054 & 5.2741 \\ \hline
1 & 20 & 44.95 & 44.95 & 0.6429 & - & 2.59 & 8.2960 & 5.5380 \\ \hline
1 & 25 & 44.97 & 44.97 & 0.6428 & - & 2.50 & 8.6127 & 5.7369 \\ \hline
1 & 30 & 44.98 & 44.98 & 0.6428 & - & 2.42 & 8.8796 & 5.9260 \\ \hline
1 & 35 & 44.98 & 44.98 & 0.6428 & - & 2.36 & 9.1115 & 6.0763 \\ \hline
1 & 40 & 44.99 & 44.99 & 0.6428 & - & 2.31 & 9.3173 & 6.2076 \\ \hline 
1 & 45 & 44.99 & 44.99 & 0.6427 & - & 2.26 & 9.5025 & 6.3446 \\ \hline
1 & 50 & 44.99 & 44.99 & 0.6427 & - & 2.22 & 9.6712 & 6.4587 \\ \hline
1 & 100 & 45.00 & 45.00 & 0.6427 & - & 1.98 & 10.8574 & 7.2401 \\ \hline \hline
\end{tabular}
\caption{Numerical results of nominal bridge with complexities $n = 1$ and different $p = q$ under yielding ($_Y$) and combined yielding and buckling constraints ($_B$), ($F = 1 \ N$; $L = 1 \ m$; steel bars and steel strings).}
\label{nr_combined_n1p}
\end{center}
\end{table}

The second optimization domain let the parameters $n$, $p$, $\alpha$ and $\beta$ ranging in the following intervals: 

\beq
n \in [1, 5],  \ \ \
p \in [1, 3],  \ \ \
\alpha \in (0, 90) \ \mbox{deg}, \ \ \
\beta \in (0, 90) \ \mbox{deg}.
\label{domain2}
\eeq
 
The results of this second design are presented in Table \ref{nr_combined_np}. For what concern the yielding case, we can observe that, for fixed values of the complexity $p$, the global optimum in the domain (\ref{domain2}) is achieved for $n^*_Y = 1$ and $p^*_Y = 3$. Moreover, the optimal aspect angles $\alpha^*_Y$ and $\beta^*_Y$ appear not depending on the complexity $n$. Merging the optimization carried out in both the domains (\ref{domain1}) and (\ref{domain2}), we can conclude that the global optimum for yielding is for $n^*_Y = 1$ and $q^*_Y = p^*_Y \rightarrow  \infty$ and for aspect angles $\alpha^*_Y = \beta^*_Y \rightarrow 45 \ deg$, which corresponds to a minimal mass $\mu^*_Y \rightarrow 0.6427$ (Table \ref{nr_combined_n1p}). It is worth noting that such a solution brings to a mass reduction, from the case with $n = 1$ and $(p,q) = (1,1)$ ($\mu^*_Y = 0.7071$), of only $9.1\%$. Moreover the above minimum doesn't take care of manufacture processes that becomes relevant for structures with numerous joints and members. Then, a finite optimal complexity $p = q$ could be achieved a posteriori by adding, eg., joint masses as illustrated in Sect. \ref{joint_mass}.

The optimizations under buckling constraints reported in Table \ref{nr_combined_np} show that the global optimum in the domain (\ref{domain2}) is the same obtained in the domain (\ref{domain1}), i.e. for $p^*_B = 1$, $\beta^*_B = 4.25 \ deg$ and $\mu^*_B = 5.0574$. Also in each complexities ranging in the intervals (\ref{domain2}) we obtain, for buckling, local minimal masses solutions that keeps only the substructures. In fact, also in domain (\ref{domain2}), we have obtained negligible values of the mass of superstructures under buckling ranging from a minimum of $3.7436x10^{-13}$ for the case with $n = 5, q = p = 1$ and a maximum of $2.3257x10^{-6}$ for the case with $n = 1, p = q = 3$. Such a results for buckling are also confirmed in the next Sect. \ref{numerical_sub}. In Table \ref{nr_combined_np}, we observe that the optimal aspect angles for buckling $\beta^*_B$ increase as complexity $n$ increase and decrease as complexity $q = p$ increase.

\begin{table}[tb]
\begin{center}
\begin{tabular}{ccccccccc} \hline \hline
$n$ & $p$ & $\alpha^{*}_{Y}$ [deg] & $\beta^{*}_{Y}$ [deg] & $\mu^{*}_{Y}$ & $\alpha^{*}_{B}$ [deg] & $\beta^{*}_{B}$ [deg] & $\mu^{*}_{B}$ & $\mu^{*}_{B,s}$ \\ \hline \hline
1 & 1 & 35.26 & 35.26 & 0.7071 & - & 4.25 & 5.0574 & 3.3827 \\ \hline 
1 & 2 & 35.26 & 35.26 & 0.7071 & - & 3.80 & 5.6662 & 3.7805 \\ \hline
1 & 3 & 41.41 & 41.41 & 0.6614 & - & 3.57 & 6.0309 & 4.0227 \\ \hline \hline
2 & 1 & 35.26 & 35.26 & 1.0607 & - & 4.40 & 7.3326 & 4.9024 \\ \hline
2 & 2 & 35.26 & 35.26 & 1.0607 & - & 3.93 & 8.2143 & 5.4843 \\ \hline
2 & 3 & 41.41 & 41.41 & 0.9922 & - & 3.69 & 8.7426 & 5.8388 \\ \hline \hline
3 & 1 & 35.26 & 35.26 & 1.2374 & - & 4.49 & 8.3705 & 5.6058 \\ \hline
3 & 2 & 35.26 & 35.26 & 1.2374 & - & 4.02 & 9.3762 & 6.2561 \\ \hline
3 & 3 & 41.41 & 41.41 & 1.1575 & - & 3.77 & 9.9791 & 6.6682 \\ \hline \hline
4 & 1 & 35.26 & 35.26 & 1.3258 & - & 4.55 & 8.8520 & 5.9276 \\ \hline
4 & 2 & 35.26 & 35.26 & 1.3258 & - & 4.07 & 9.9149 & 6.6211 \\ \hline
4 & 3 & 41.41 & 41.41 & 1.2402 & - & 3.82 & 10.5523 & 7.0516 \\ \hline \hline
5 & 1 & 35.26 & 35.26 & 1.3700 & - & 4.59 & 9.0790 & 6.0723 \\ \hline
5 & 2 & 35.26 & 35.26 & 1.3701 & - & 4.10 & 10.1689 & 6.7921 \\ \hline
5 & 3 & 41.41 & 41.41 & 1.2816 & - & 3.85 & 10.8226 & 7.2302 \\ \hline \hline
\end{tabular}
\caption{Numerical results of nominal bridges with different complexities $n$ and $p$ under yielding ($_Y$) and combined yielding and buckling constraints ($_B$), ($F = 1 \ N$; $L = 1 \ m$; steel bars and steel strings).}
\label{nr_combined_np}
\end{center}
\end{table}

\subsection{Substructures} \label{numerical_sub}

In this Section, we show the results obtained for the optimizations of the \emph{substructure} bridges showed in Fig. \ref{fig:substructures_exemplaries} in which only structure below the roadway is allowed. Starting from the basic module illustrated in Fig. \ref{basic_module}c, we have considered different complexities $n$, $p$ and different aspect angles $\alpha$ and $\beta$ ranging in two domains with the aims to get the global minimum mass design both under yielding constraints and under buckling constraints.

First of all, we start fixing parameter $n$ to unity and let parameters $p$, $\alpha$ and $\beta$ ranging in the following intervals:

\beq
p \in [1, 500],  \ \ \
\alpha \in (0, 90) \ \mbox{deg}, \ \ \
\beta \in (0, 90) \ \mbox{deg}.
\label{domain1_sub}
\eeq

The results of this first design are presented in Table \ref{nr_sub_n1p}. For what concern the design under only yielding constraints the global minimum is achieved for a complexity $p^*_Y \rightarrow  \infty$ and for an aspect angle $\beta^*_Y \rightarrow 45 \ deg$, which corresponds to a minimal mass $\mu^*_Y \rightarrow 0.6427$ (Table \ref{nr_sub_n1p}). The optimizations for buckling constraints, instead, allow to identify a global minimum for complexity $p^*_B = 1$ and for an aspect angle $\beta^*_B = 4.25 \ deg$, which corresponds to a minimal mass $\mu^*_B = 5.0574$ (Table \ref{nr_sub_n1p}).

\begin{table}[tb]
\begin{center}
\begin{tabular}{ccccccc} \hline \hline
$n$ & $p$ & $\beta^{*}_{Y}$ [deg] & $\mu^{*}_{Y}$ & $\beta^{*}_{B}$ [deg] & $\mu^{*}_{B}$ & $\mu^{*}_{B,s}$ \\ \hline \hline
1 & 1 & 35.26 & 0.7071 & 4.25 & 5.0574 & 3.3827 \\ \hline 
1 & 2 & 35.26 & 0.7071 & 3.80 & 5.6666 & 3.7805 \\ \hline
1 & 3 & 41.41 & 0.6614 & 3.57 & 6.0312 & 4.0227 \\ \hline 
1 & 4 & 43.23 & 0.6514 & 3.40 & 6.3265 & 4.2228 \\ \hline
1 & 5 & 43.96 & 0.6476 & 3.27 & 6.5687 & 4.3899 \\ \hline
1 & 10 & 44.78 & 0.6437 & 2.91 & 7.3843 & 4.9308 \\ \hline 
1 & 15 & 44.91 & 0.6431 & 2.72 & 7.9058 & 5.2741 \\ \hline
1 & 20 & 44.95 & 0.6429 & 2.59 & 8.2969 & 5.5380 \\ \hline
1 & 25 & 44.97 & 0.6428 & 2.50 & 8.6131 & 5.7368 \\ \hline 
1 & 30 & 44.98 & 0.6428 & 2.42 & 8.8800 & 5.9260 \\ \hline
1 & 35 & 44.98 & 0.6428 & 2.36 & 9.1120 & 6.0763 \\ \hline
1 & 40 & 44.99 & 0.6428 & 2.31 & 9.3177 & 6.2076 \\ \hline 
1 & 45 & 44.99 & 0.6427 & 2.26 & 9.5029 & 6.3446 \\ \hline
1 & 50 & 44.99 & 0.6427 & 2.22 & 9.6716 & 6.4587 \\ \hline
1 & 100 & 45.00 & 0.6427 & 1.98 & 10.8578 & 7.2401 \\ \hline
1 & 200 & 45.00 & 0.6427 & 1.76 & 12.1881 & 8.1437 \\ \hline
1 & 300 & 45.00 & 0.6427 & 1.65 & 13.0402 & 8.6860 \\ \hline
1 & 400 & 45.00 & 0.6427 & 1.57 & 13.6806 & 9.1281 \\ \hline
1 & 500 & 45.00 & 0.6427 & 1.51 & 14.1994 & 9.4904 \\ \hline \hline
\end{tabular}
\caption{Numerical results of substructures with complexities $n = 1$ and different $p$ under yielding ($_Y$) and combined yielding and buckling constraints ($_B$), ($F = 1 \ N$; $L = 1 \ m$; steel bars and steel cables).}
\label{nr_sub_n1p}
\end{center}
\end{table}

Then, we let parameters $n$, $p$, $\alpha$ and $\beta$ ranging in the following intervals:

\beq
n \in [1, 5],  \ \ \
p \in [1, 3],  \ \ \
\alpha \in (0, 90) \ \mbox{deg}, \ \ \
\beta \in (0, 90) \ \mbox{deg}.
\label{domain2_sub}
\eeq

The results of this second design are presented in Table \ref{nr_sub_np}. For what concern the design under only yielding constraints, the global minimum in domain (\ref{domain2_sub}), is achieved for complexities $n^*_Y = 1, p^*_Y = 3$ and for an aspect angle $\beta^*_Y \rightarrow 41.41 \ deg$, which corresponds to a minimal mass $\mu^*_Y = 0.6614$ (Table \ref{nr_sub_np}). For the optimizations under buckling constraints, instead, we have obtained a global minimum for complexities $n^*_B = p^*_B = 1$ and for an aspect angle $\beta^*_B = 4.25 \ deg$, which corresponds to a minimal mass $\mu^*_B = 5.0574$ (Table \ref{nr_sub_np}). 

\begin{table}[tb]
\begin{center}
\begin{tabular}{ccccccc} \hline \hline
$n$ & $p$ & $\beta^{*}_{Y}$ [deg] & $\mu^{*}_{Y}$ & $\beta^{*}_{B}$ [deg] & $\mu^{*}_{B}$ & $\mu^{*}_{B,s}$ \\ \hline \hline
1 & 1 & 35.26 & 0.7071 & 4.25 & 5.0574 & 3.3827 \\ \hline 
1 & 2 & 35.26 & 0.7071 & 3.80 & 5.6662 & 3.7805 \\ \hline
1 & 3 & 41.41 & 0.6614 & 3.57 & 6.0308 & 4.0227 \\ \hline \hline
2 & 1 & 35.26 & 1.0607 & 4.40 & 7.3326 & 4.9024 \\ \hline
2 & 2 & 35.26 & 1.0607 & 3.93 & 8.2143 & 5.4843 \\ \hline
2 & 3 & 41.41 & 0.9922 & 3.69 & 8.7426 & 5.8388 \\ \hline \hline
3 & 1 & 35.26 & 1.2374 & 4.49 & 8.3705 & 5.6058 \\ \hline
3 & 2 & 35.26 & 1.2374 & 4.02 & 9.3763 & 6.2562 \\ \hline
3 & 3 & 41.41 & 1.1575 & 3.77 & 9.9791 & 6.6682 \\ \hline \hline
4 & 1 & 35.26 & 1.3258 & 4.55 & 8.8531 & 5.9929 \\ \hline
4 & 2 & 35.26 & 1.3258 & 4.07 & 9.9149 & 6.6212 \\ \hline
4 & 3 & 41.41 & 1.2402 & 3.82 & 10.5523 & 7.0517 \\ \hline \hline
5 & 1 & 35.26 & 1.3700 & 4.59 & 9.0790 & 6.0723 \\ \hline
5 & 2 & 35.26 & 1.3701 & 4.10 & 10.1690 & 6.7921 \\ \hline
5 & 3 & 41.41 & 1.2816 & 3.85 & 10.8226 & 7.2302 \\ \hline \hline
\end{tabular}
\caption{Numerical results of substructures with different complexities $n$ and $p$ under yielding ($_Y$) and combined yielding and buckling constraints ($_B$), ($F = 1 \ N$; $L = 1 \ m$; steel bars and steel cables).}
\label{nr_sub_np}
\end{center}
\end{table}

Such a results retrace the results in Table \ref{nr_combined_n1p} already obtained for the nominal bridges. This can be explained considering the assumption $\rho = 1$ and the symmetry of the constraints of the bridge (double fixed hinges, $HH$). It is shown, eg., in the numerical results of Tables \ref{nr_nom_HR}, \ref{nr_sub_HR}, \ref{nr_sup_HR} in Appendix \ref{App_A} that changing constraints from double fixed hinges ($HH$) to fixed hinge and rolling hinge ($HR$), the equivalence between \emph{nominal} bridge, \emph{substructure} and \emph{superstructure} bridge never subsists.

\subsection{Superstructures} \label{numerical_super}

We end the numerical results without deck showing the optimizations of the superstructure bridges showed in Fig. \ref{fig:superstructures_exemplaries} in which only structure above the roadway is allowed. Starting from the basic module illustrated in Fig. \ref{basic_module}b, we have considered different complexities $n$, $q$ and different aspect angles $\alpha$ and $\beta$ ranging in two domains with the aims to get the global minimum mass design both under only yielding constraints and under buckling constraints.

First of all, we start fixing parameter $n$ to unity and let parameters $q$, $\alpha$ and $\beta$ ranging in the following intervals:

\beq
q \in [1, 500],  \ \ \
\alpha \in (0, 90) \ \mbox{deg}, \ \ \
\beta \in (0, 90) \ \mbox{deg}.
\label{domain1_sup}
\eeq

Table \ref{nr_super_n1q} shows the results obtained considering parameters ranging in the domain (\ref{domain1_sup}). For what concern the design under only yielding constraints, the numerical results in Table \ref{nr_super_n1q} show that the global minimum is achieved for a complexity $q^*_Y \rightarrow  \infty$ and for an aspect angle $\alpha^*_Y \rightarrow 45 \ deg$, which corresponds to a minimal mass $\mu^*_Y \rightarrow 0.6427$ (Table \ref{nr_super_n1q}). The optimizations for buckling constraints identify a global minimum for complexity $q^*_B \rightarrow \infty$ and for an aspect angle $\alpha^*_B \rightarrow 90 \ deg$, which corresponds to a minimal mass $\mu^*_B \rightarrow 4.6151$ (Table \ref{nr_super_n1q}).

\begin{table}[tb]
\begin{center}
\begin{tabular}{ccccccc} \hline \hline
$n$ & $q$ & $\alpha^{*}_{Y}$ [deg] & $\mu^{*}_{Y}$ & $\alpha^{*}_{B}$ [deg] & $\mu^{*}_{B}$ & $\mu^{*}_{B,s}$ \\ \hline \hline
1 & 1 & 35.26 & 0.7071 & 26.56 & 801.7357 & 0.1250 \\ \hline 
1 & 2 & 35.26 & 0.7071 & 36.22 & 514.7336 & 0.1231 \\ \hline
1 & 3 & 41.41 & 0.6614 & 45.31 & 410.5778 & 0.2087 \\ \hline 
1 & 4 & 43.23 & 0.6514 & 51.70 & 346.8507 & 0.2326 \\ \hline
1 & 5 & 43.96 & 0.6476 & 56.64 & 301.3080 & 0.2523 \\ \hline
1 & 10 & 44.78 & 0.6437 & 70.63 & 181.3748 & 0.3101 \\ \hline 
1 & 15 & 44.91 & 0.6431 & 76.86 & 128.6606 & 0.3364 \\ \hline
1 & 20 & 44.95 & 0.6429 & 80.19 & 99.3742 & 0.3505 \\ \hline
1 & 25 & 44.97 & 0.6428 & 82.25 & 80.8126 & 0.3593 \\ \hline 
1 & 30 & 44.98 & 0.6428 & 83.56 & 68.0759 & 0.3649 \\ \hline
1 & 35 & 44.98 & 0.6428 & 84.51 & 58.8000 & 0.3690 \\ \hline
1 & 40 & 44.99 & 0.6428 & 85.23 & 51.7491 & 0.3721 \\ \hline 
1 & 45 & 44.99 & 0.6427 & 85.78 & 46.2113 & 0.3744 \\ \hline
1 & 50 & 44.99 & 0.6427 & 86.21 & 41.7482 & 0.3763 \\ \hline
1 & 100 & 45.00 & 0.6427 & 88.14 & 21.3224 & 0.3846 \\ \hline
1 & 200 & 45.00 & 0.6427 & 89.07 & 10.9156 & 0.3886 \\ \hline
1 & 300 & 45.00 & 0.6427 & 89.38 & 7.4204 & 0.3900 \\ \hline
1 & 400 & 45.00 & 0.6427 & 89.53 & 5.6680 & 0.3907 \\ \hline
1 & 500 & 45.00 & 0.6427 & 89.62 & 4.6151 & 0.3910 \\ \hline \hline
\end{tabular}
\caption{Numerical results of superstructures with complexities $n = 1$ and different $q$ under yielding ($_Y$) and combined yielding and buckling constraints ($_B$), ($F = 1 \ N$; $L = 1 \ m$; steel bars and steel cables).}
\label{nr_super_n1q}
\end{center}
\end{table}

Then, we let parameters $n$, $q$, $\alpha$ and $\beta$ ranging in the following intervals:

\beq
n \in [1, 5],  \ \ \
q \in [1, 3],  \ \ \
\alpha \in (0, 90) \ \mbox{deg}, \ \ \
\beta \in (0, 90) \ \mbox{deg}.
\label{domain2_sup}
\eeq

Refer to Table \ref{nr_super_nq} for the results of the optimizations over the domain (\ref{domain2_sup}). For yielding constraints, the global minimum is obtained for complexities $n^*_Y = 1, q^*_Y = 3$ and for an aspect angle $\alpha^*_Y = 41.41 \ deg$, which corresponds to a minimal mass $\mu^*_Y = 0.6614$ (Table \ref{nr_super_nq}). For the optimizations under buckling constraints, instead, we have obtained a global minimum for complexities $n^*_B = 1, p^*_B = 3$ and for an aspect angle $\alpha^*_B = 45.31 \ deg$, which corresponds to a minimal mass $\mu^*_B = 410.5778$ (Table \ref{nr_super_nq}).

\begin{table}[tb]
\begin{center}
\begin{tabular}{ccccccc} \hline \hline
$n$ & $q$ & $\alpha^{*}_{Y}$ [deg] & $\mu^{*}_{Y}$ & $\alpha^{*}_{B}$ [deg] & $\mu^{*}_{B}$ & $\mu^{*}_{B,s}$ \\ \hline \hline
1 & 1 & 35.26 & 0.7071 & 26.56 & 801.7349 & 0.1250 \\ \hline 
1 & 2 & 35.26 & 0.7071 & 36.22 & 514.7336 & 0.1231 \\ \hline
1 & 3 & 41.41 & 0.6614 & 45.31 & 410.5778 & 0.2087 \\ \hline \hline
2 & 1 & 35.26 & 1.0607 & 26.56 & 1085.2 & 0.1875 \\ \hline
2 & 2 & 35.26 & 1.0607 & 36.22 & 696.7464 & 0.2747 \\ \hline
2 & 3 & 41.41 & 0.9922 & 45.31 & 555.7697 & 0.3130 \\ \hline \hline
3 & 1 & 35.26 & 1.2374 & 26.56 & 1185.4 & 0.2187 \\ \hline
3 & 2 & 35.26 & 1.2374 & 36.22 & 761.1108 & 0.3204 \\ \hline
3 & 3 & 41.41 & 1.1575 & 45.31 & 607.1179 & 0.3652 \\ \hline \hline
4 & 1 & 35.26 & 1.3258 & 26.56 & 1220.9 & 0.2343 \\ \hline
4 & 2 & 35.26 & 1.3258 & 36.22 & 783.8740 & 0.3433 \\ \hline
4 & 3 & 41.41 & 1.2402 & 45.31 & 625.2800 & 0.3913 \\ \hline \hline
5 & 1 & 35.26 & 1.3700 & 26.56 & 1233.4 & 0.2421 \\ \hline
5 & 2 & 35.26 & 1.3701 & 36.22 & 791.9251 & 0.3548 \\ \hline
5 & 3 & 41.41 & 1.2816 & 45.31 & 631.7049 & 0.4043 \\ \hline \hline
\end{tabular}
\caption{Numerical results of superstructures with different complexities $n$ and $p$ under yielding ($_Y$) and combined yielding and buckling constraints ($_B$), ($F = 1 \ N$; $L = 1 \ m$; steel bars and steel cables).}
\label{nr_super_nq}
\end{center}
\end{table}

The optimizations for yielding conducted for the superstructure only bridges over the domains (\ref{domain1_sup}) and (\ref{domain2_sup}) allow to find a global minimum ($q^*_Y \rightarrow  \infty, \ \alpha^*_Y \rightarrow 45 \ deg, \ \mu^*_Y \rightarrow 0.6427$) that matches the minimum founded starting from \emph{nominal} bridges and \emph{substructure} only bridges. As validation of the adopted numerical solution, the here found global minimum corresponds to the result illustrated in Fig. 2 by \cite{Michell1904}. 
It's interesting to note that the results under buckling constraints ($q^*_B \rightarrow  \infty, \ \alpha^*_B \rightarrow 90 \ deg, \ \mu^*_B \rightarrow 4.6151$) show that, differently for what obtained from the combined or the substructure bridges ($p^*_B = 1, \ \beta^*_B = 4.25 \ deg, \ \mu^*_B = 5.0574$), the optimal complexity $q$ is at infinite. Moreover, it is worth noting that increasing complexity $q$ allows a strong reduction of the mass, that is reducing from $\mu^*_B = 801.7357$ for $q = 1$ to $\mu^*_B = 4.6151$ for $q = 500$. Then, with the optimizations carried out in Sects. \ref{numerical_nominal}, \ref{numerical_sub} and \ref{numerical_super}, the case of a centrally loaded beam illustrated in Fig. 2 of \cite{Michell1904} has been extended to accomplish the buckling case.
Tables \ref{nr_sub_n1p}, \ref{nr_sub_np} for substructure bridges and Tables \ref{nr_super_n1q}, \ref{nr_super_nq} for superstructure bridges also show the total masses of cables ($\mu^*_{B,s}$) obtained under buckling for each optimized case. We show that, for the substructures, the total mass of the cables ($\mu^*_{B,s}$) is the most part of the total mass of the structure under buckling ($\mu^*_{B}$). For the case $n= q = p = 1$, eg., $\mu^*_{B,s}/\mu^*_{B} = 0.67$ for the substructure (see Table \ref{nr_sub_n1p}) while $\mu^*_{B,s}/\mu^*_{B} = 1.56x10^{-4}$ for the superstructure (see Table \ref{nr_super_n1q}). This makes clear that the substructure bridges under buckling work mainly with cables and the length and the forces (and then the mass) of the bars can be extremely reduced playing with the aspect angle $\beta$. Moreover, Table \ref{nr_super_n1q} shows that the global minimum for buckling for the superstructure ($q^*_B \rightarrow  \infty, \ \alpha^*_B \rightarrow 90 \ deg, \ \mu^*_B \rightarrow 4.6151$) corresponds to a maximum of the ratio $\mu^*_{B,s}/\mu^*_{B} = 0.085$ over the domain (\ref{domain1_sup}).

\section{Numerical results with deck} \label{results_deck}

In the present Section, we report numerical results of the tensegrity bridges including deck mass and joints mass. Taking into account the results obtained in the case without deck, we performed the numerical simulation only for buckling constraints, since it has been shown that this is the mode of failure in all cases. Moreover, the optimizations will be performed only for \emph{substructure} bridges (Sect. \ref{numerical_sub_deck}) and \emph{superstructure} bridges (Sect. \ref{numerical_super}) and not for \emph{nominal} bridges, since these optimizations bring to solutions keeping only substructure (for GC: to be double checked). The numerical results are presented in terms of $\mu^*_d$, $\mu^*_{B,S}$, $\mu^*_{B,tot}$, $\alpha^*_B$ and $\beta^*_B$ denoting respectively the mass of deck ($2^n \ m_d \ \sigma_s / (\rho_s \ F \ L)$), the mass of the bridge structure, the total minimal mass including bridge structure, deck and joints ($\mu^*_{B,S} + \mu^*_d + \mu^*_J$) and the optimal aspect angles. The total mass of joints $\mu^*_J$ is computed as the product between the number of joints ($n_n$) and a fixed joint factor ($\Omega$). The results are obtained numerically through the MatLab\textsuperscript \textregistered \ program written employing the algorithm illustrated in Sect. 3 of \cite{Skelton2014}. The optimization problems presented in Tables \ref{nr_sub_np_deck}, \ref{nr_sup_nq_deck} are solved assuming $L = 30 \ m$, $F = 450 \ kN$, deck mass computed as defined in (\ref{m_deck}), steel for bars and deck beams, Spectra\textsuperscript \textregistered \ for cables (refer to Table \ref{tab_materials} for the material properties; $\rho = 31.72$; $\eta = 1216.55$). In all the optimized cases, we set step increments of complexities $n$, $p$ and $q$ to $1$ and step increments of $0.01$ $deg$ for the aspect angles $\alpha$ and $\beta$.

\subsection{Substructures} \label{numerical_sub_deck}

In this Section, we show the results obtained for the optimizations of the \emph{substructure} bridges showed in Fig. \ref{fig:substructures_exemplaries} including deck and joints masses. First of all, we let parameters $n$, $p$, $\alpha$ and $\beta$ ranging in the following intervals:

\beq
n \in [1, 5],  \ \ \
p \in [1, 3],  \ \ \
\alpha \in (0, 90) \ \mbox{deg}, \ \ \
\beta \in (0, 90) \ \mbox{deg}.
\label{domain1_sub_deck}
\eeq

A first set of results are presented in Table \ref{nr_sub_np_deck} in which we didn't consider jet joint masses. The global minimum in the domain (\ref{domain1_sub_deck}) is obtained for complexities $n^*_B = 5$, $p^*_B = 1$ and for an aspect angle $\beta^*_B = 4.11 \ deg$, which corresponds to a total minimal mass $\mu^*_{B,tot} = 334.7613$ (Table \ref{nr_sub_np_deck}).

\begin{table}[tb]
\begin{center}
\begin{tabular}{cccccc} \hline \hline
$n$ & $p$ & $\mu^{*}_{d}$ & $\beta^{*}_{B}$ [deg] & $\mu^{*}_{B,S}$ & $\mu^{*}_{B,tot}$ \\ \hline \hline
1 & 1  & 6659.9 & 4.13 & 8.8585 & 6668.8 \\ \hline 
1 & 2  & 6659.9 & 3.69 & 9.9260 & 6669.9 \\ \hline
1 & 3  & 6659.9 & 3.47 & 10.5649 & 6670.5 \\ \hline \hline
2 & 1  & 2917.9 & 4.09 & 10.3025 & 2928.2 \\ \hline
2 & 2  & 2917.9 & 3.66 & 11.5443 & 2929.5 \\ \hline
2 & 3  & 2917.9 & 3.44 & 12.2875 & 2930.2 \\ \hline \hline
3 & 1  & 1364.3 & 4.09 & 10.5151 & 1347.8 \\ \hline
3 & 2  & 1364.3 & 3.66 & 11.7826 & 1376.1 \\ \hline
3 & 3  & 1364.3 & 3.43 & 12.5411 & 1376.9 \\ \hline \hline
4 & 1  & 659.5 & 4.10 & 10.5120 & 670.0878 \\ \hline
4 & 2  & 659.5 & 3.66 & 11.7791 & 671.3548 \\ \hline
4 & 3  & 659.5 & 3.44 & 12.5373 & 672.1131 \\ \hline \hline
5 & 1  & 324.28 & 4.11 & 10.4841 & 334.7613 \\ \hline
5 & 2  & 324.28 & 3.67 & 11.7459 & 336.0232 \\ \hline
5 & 3  & 324.28 & 3.45 & 12.5021 & 336.7793 \\ \hline \hline
\end{tabular}
\caption{Numerical results of substructures with deck for different complexities $n$ and $p$ under buckling constraints ($_B$), ($F = 450 \ kN$; $L = 30 \ m$; $w_d = 3 \ m$, steel bars and deck, Spectra\textsuperscript \textregistered - UHMWPE cables).}
\label{nr_sub_np_deck}
\end{center}
\end{table}

The results in Table \ref{nr_sub_np_deck} identify an optimal complexity $n$ lying on the boundary of the domain (\ref{domain1_sub_deck}). Then we have performed another optimization keeping $p = 1$ and increasing only complexity $n$. In this case, since the number of nodes given in (\ref{nnodes}) is exponentially increasing with $n$, the numerical simulation of such structure would be computationally heavy. For that reason, we made use of the analytical solution given in Theorem \ref{theo:mim_buckling_np1_deck}. In this case, we have also added the mass of joints considering increasing values of the joint factor $\Omega$ and the results are showed in Fig. \ref{fig:NR_deck_n_qp1}. The red curve reports the masses of substructure bridge only ($\mu^*_{B,S}$), the solid curve is the total mass without joints and the dashed and dotted curves include the joint masses. We obtained a finite complexity $n$ ranging between $11$ and $12$ considering joint masses.

\begin{figure}[hb]
\unitlength1cm
\begin{picture}(10,5)
\put(4,0){\psfig{figure=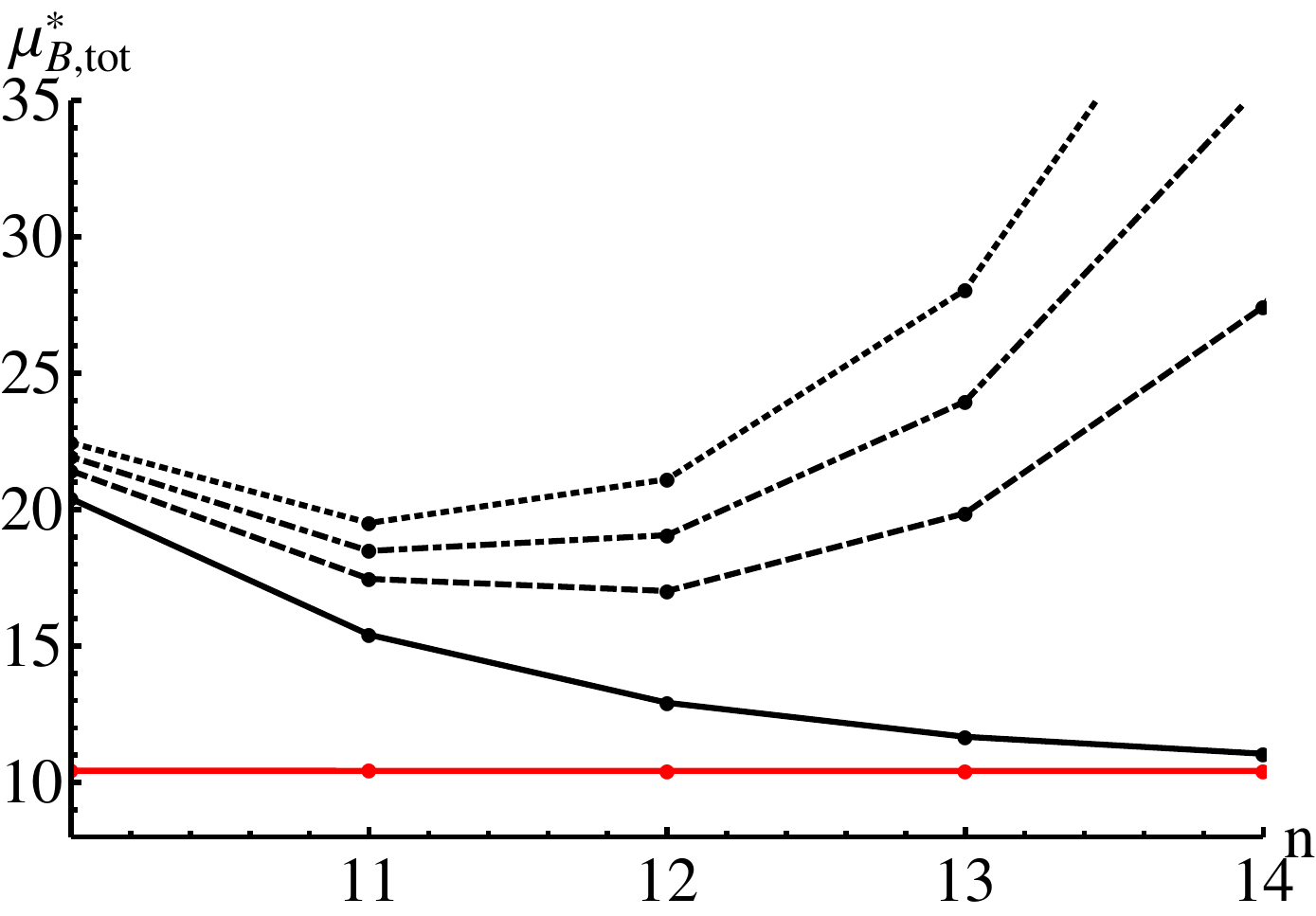,height=4.5cm}}
\put(10.5,1){\psfig{figure=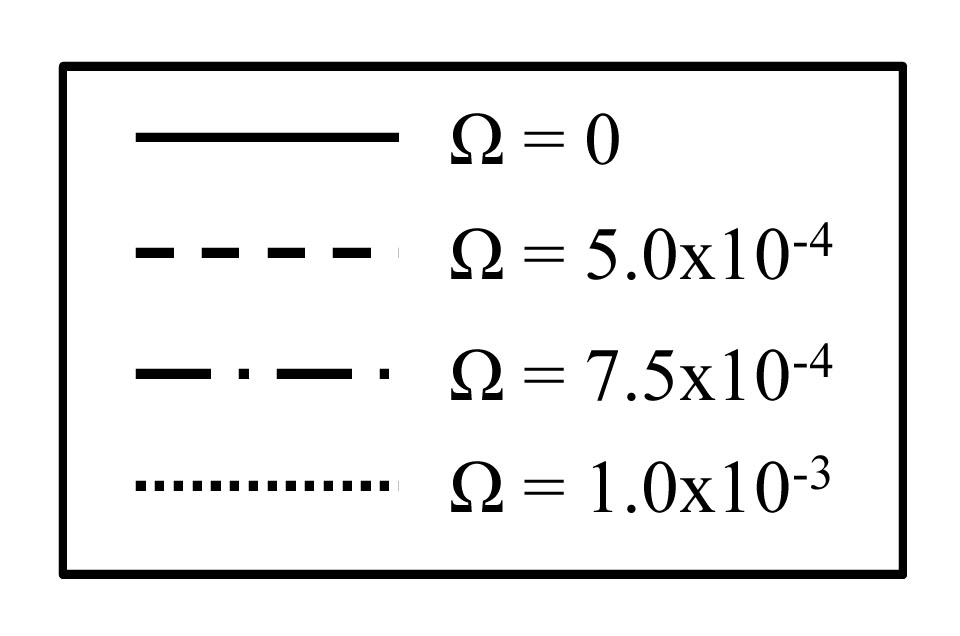,height=2cm}}
\end{picture}
\caption{Total masses ($\mu^*_{B,tot}$, black curves) for different values of the joint factor $\Omega$ and structural masses ($\mu^*_{B,S}$, red curves) under buckling constraints for $p = q = 1$ and different $n$.}
\label{fig:NR_deck_n_qp1}
\end{figure}

\subsection{Superstructures} \label{numerical_super}

In this Section, we show the results obtained for the optimizations of the \emph{superstructure} bridges showed in Fig. \ref{fig:superstructures_exemplaries} including deck and joints masses. First of all, we let parameters $n$, $q$, $\alpha$ and $\beta$ ranging in the following intervals:

\beq
n \in [1, 5],  \ \ \
q \in [1, 3],  \ \ \
\alpha \in (0, 90) \ \mbox{deg}, \ \ \
\beta \in (0, 90) \ \mbox{deg}.
\label{domain1_sup_deck}
\eeq

A first set of results are presented in Table \ref{nr_sup_nq_deck} in which we didn't consider jet joint masses. The global minimum in the domain (\ref{domain1_sup_deck}) is obtained for complexities $n^*_B = 5$, $p^*_B = 3$ and for an aspect angle $\alpha^*_B = 45.31 \ deg$, which corresponds to a total minimal mass $\mu^*_{B,tot} = 1235.3$ (Table \ref{nr_sup_nq_deck}).

\begin{table}[tb]
\begin{center}
\begin{tabular}{cccccc} \hline \hline
$n$ & $q$ & $\mu^{*}_{d}$ & $\alpha^{*}_{B}$ [deg] & $\mu^{*}_{B,S}$ & $\mu^{*}_{B,tot}$ \\ \hline \hline
1 & 1  & 6659.9 & 26.56 & 1484.5 & 8144.4 \\ \hline 
1 & 2  & 6659.9 & 36.22 & 953.0505 & 7613.0 \\ \hline
1 & 3  & 6659.9 & 45.31 & 760.1925 & 7420.1 \\ \hline \hline
2 & 1  & 2917.9 & 26.56 & 1760.7 & 4678.6 \\ \hline
2 & 2  & 2917.9 & 36.22 & 1130.4 & 4048.3 \\ \hline
2 & 3  & 2917.9 & 45.31 & 901.6280 & 3819.6 \\ \hline \hline
3 & 1  & 1364.3 & 26.56 & 1798.5 & 3162.8 \\ \hline
3 & 2  & 1364.3 & 36.22 & 1154.7 & 2519.0 \\ \hline
3 & 3  & 1364.3 & 45.31 & 921.0011 & 2285.3 \\ \hline \hline
4 & 1  & 659.5 &  26.56& 1790.9 & 2450.5 \\ \hline
4 & 2  & 659.5 & 36.22 & 1149.8 & 1809.4 \\ \hline
4 & 3  & 659.5 & 45.31 & 917.1257 & 1576.7 \\ \hline \hline
5 & 1  & 324.25 & 26.56 & 1779.1 & 2103.4 \\ \hline
5 & 2  & 324.25 & 36.22 & 1142.2 & 1466.5 \\ \hline
5 & 3  & 324.25 & 45.31 & 911.0597 & 1235.3 \\ \hline \hline
\end{tabular}
\caption{Numerical results of superstructures with deck for different complexities $n$ and $q$ under buckling constraints ($_B$), ($F = 450 \ kN$; $L = 30 \ m$; $w_d = 3 \ m$, steel bars and deck, Spectra\textsuperscript \textregistered - UHMWPE cables).}
\label{nr_sup_nq_deck}
\end{center}
\end{table}

We then fix parameter $n=5$, and let $q$, $\alpha$ and $\beta$ ranging in the following intervals:

\beq
q \in [1, 50],  \ \ \
\alpha \in (0, 90) \ \mbox{deg}, \ \ \
\beta \in (0, 90) \ \mbox{deg}.
\label{domain2_sup_deck}
\eeq

Results of such optimizations are reported in Table \ref{nr_sup_n5q_deck}, in which we didn't consider jet joint masses. The global minimum in the domain (\ref{domain2_sup_deck}) is obtained for complexities $n^*_B = 5$, $p^*_B = 50$ and for an aspect angle $\alpha^*_B = 86.21 \ deg$, which corresponds to a total minimal mass $\mu^*_{B,tot} = 416.8388$ (Table \ref{nr_sup_n5q_deck}).

\begin{table}[tb]
\begin{center}
\begin{tabular}{cccccc} \hline \hline
$n$ & $q$ & $\mu^{*}_{d}$ & $\alpha^{*}_{B}$ [deg] & $\mu^{*}_{B,S}$ & $\mu^{*}_{B,tot}$ \\ \hline \hline
5 & 1  & 324.28 & 26.56 & 1779.1 & 2103.4 \\ \hline 
5 & 2  & 324.28 & 36.22 & 1142.2 & 1466.5 \\ \hline
5 & 3  & 324.28 & 45.31 & 911.0597 & 1235.3 \\ \hline 
5 & 4  & 324.28 & 51.71 & 769.6990 & 1093.9 \\ \hline
5 & 5  & 324.28 & 56.64 & 668.5721 & 992.8493 \\ \hline
5 & 10  & 324.28 & 70.64 & 402.4183 & 726.6955 \\ \hline 
5 & 15  & 324.28 & 76.86 & 285.4356 & 609.7128 \\ \hline
5 & 20  & 324.28 & 80.19 & 220.3840 & 544.6612 \\ \hline
5 & 25  & 324.28 & 82.21 & 179.2518 & 503.5291 \\ \hline 
5 & 30  & 324.28 & 83.56 & 150.9876 & 475.2646 \\ \hline 
5 & 35  & 324.28 & 84.52 & 130.4024 & 454.6796 \\ \hline 
5 & 40  & 324.28 & 85.23 & 114.7554 & 439.0326 \\ \hline 
5 & 45  & 324.28 & 85.78 & 102.4660 & 426.7433 \\ \hline 
5 & 50  & 324.28 & 86.21 & 92.5615 & 416.8388 \\ \hline \hline
\end{tabular}
\caption{Numerical results of superstructures with deck for $n = 5$ and different complexities $q$ under buckling constraints ($_B$), ($F = 450 \ kN$; $L = 30 \ m$; $w_d = 3 \ m$, steel bars and deck, Spectra\textsuperscript \textregistered - UHMWPE cables).}
\label{nr_sup_n5q_deck}
\end{center}
\end{table}

The results in Table \ref{nr_sup_n5q_deck} identify an optimal complexity $q$ lying on the boundary of the domain (\ref{domain2_sup_deck}). It is worth noting that the above solution is without joint masses. Then, we have performed another optimization over the same domain (\ref{domain2_sup_deck}) but considering joint masses with increasing joint factor $\Omega$ and their results are showed in Fig. \ref{fig:NR_deck_n5_q}. The red curve reports the masses of superstructure bridge ($\mu^*_{B,S}$), the solid curve is the total mass without joints and the dashed and dotted curves include the joint masses. We obtained a finite complexity $q$ ranging between $10$ and $20$ considering joint masses. It must be noticed that, however, the minimum mass obtained with superstructure is bigger then the minimum mass obtained with the substructure, that has been confirmed as the most convenient bridge.

\begin{figure}[hb]
\unitlength1cm
\begin{picture}(10,5)
\put(2.5,0){\psfig{figure=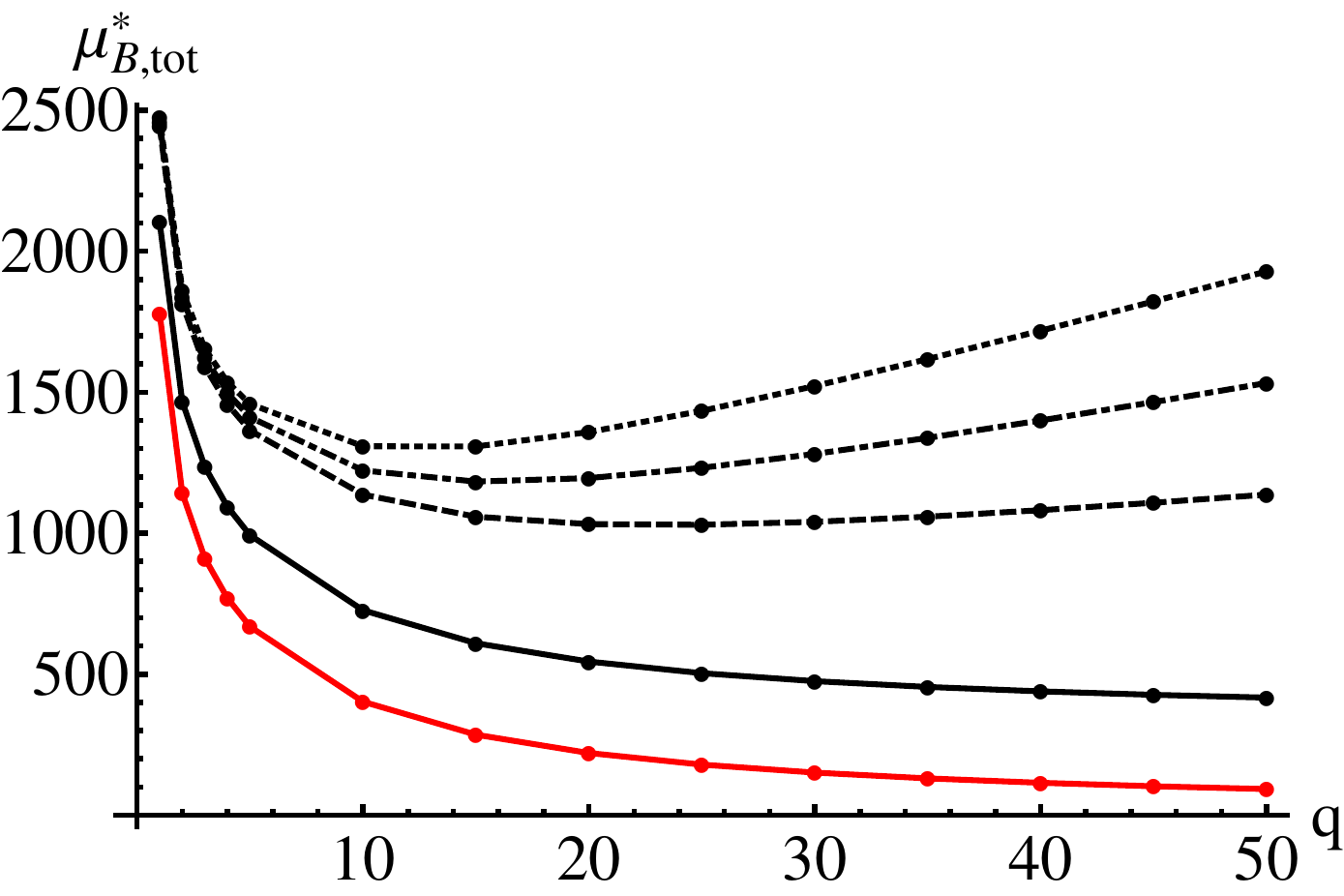,height=5cm}}
\put(10,1.5){\psfig{figure=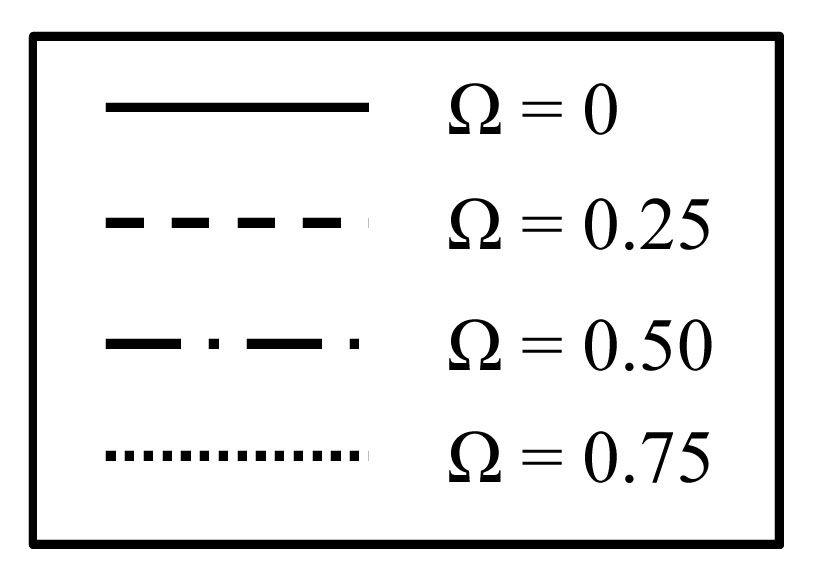,height=2cm}}
\end{picture}
\caption{Total masses ($\mu^*_{B,tot}$, black curves) for different values of the joint factor $\Omega$ and structural masses ($\mu^*_{B,s}$, red curve) under buckling constraints for $n = 5$ vs $q$ for superstructure.}
\label{fig:NR_deck_n5_q}
\end{figure}

\section{Concluding remarks}
\label{conclusions}

This study provides closed form solutions (analytical expressions) for minimal mass tensegrity bridge designs. The forces, locations, and number of members are optimized to minimize mass subject to both buckling and yielding constraints for a planar structure with fixed-hinge/fixed-hinge boundary conditions. 

We designed bridges from the elementary consideration of i) yielding constraints, ii) buckling constraints, iii) without deck mass, iv) with deck mass, v) \emph{superstructure} only, vi) \emph{substructure} only, vii) without joint mass, viii) with joint mass. 

We optimize the complexity of the structure, where structural complexity as the number of members in the design. This can be related to $3$ parameters $(n,p,q)$, where $2^n$ is the number of deck sections along the span; $p$ is the number of compressive members (bars) reaching from the span center to the \emph{substructure}; and $q$ is the number of cables reaching from the span center to the \emph{superstructure}. Hence we refer to $(n,p,q)$ as the three different kinds of \emph{complexities} of the structure. We used a tensegrity structural paradigm which allowed these several kinds of complexities. 
The complexity $n$ is determined by a self-similar law to fill the space of the bridge. As the number of self-similar iterations go to infinity we get a tensegrity fractal topology. However, the number of self-similar iterations $n$ and the complexities $p$ and $q$ required to minimize mass, under different circumstances within the set of $8$ possibilities i),...,viii) listed above, go to an optimal number between $1$ and infinity, where an infinite complexity fills the define space with a continuum.

First we optimized structures under yielding constraints for the simply-supported case ($n=1$) with no deck. The number of self-similar iterations $n$ of the given tensegrity module goes to infinity as the mass approaches the minimum. Our result produces the same topology as \cite{Michell1904}, where there is a compressive member at $45 \ deg$ attached at each boundary, connecting to a $1/4$ pie shaped continuum material piece at the center. The bottom half of the bridge (the \emph{substructure}) is the \emph{dual} of the \emph{superstructure} (\emph{dual} meaning flip the structure about the horizontal axis and replace all tension members with compression members and all tension members with compressive members). We showed that the top half of this structure is the optimal topology for bridge designs which do not allow any \emph{substructure}, and conversely that the bottom half of this structure is the optimal topology for bridges allowing no \emph{superstructure}.

Secondly, we optimized the simply supported bridge ($n=1$) under buckling constraints with no deck. For the \emph{superstructure} design we proved that the minimal mass is achieved at high values of $q$, approaching a continuum (where the shape of the structure is a half disk). It is interesting that this shape (designed under buckling constraints) is the same as the result of \cite{Michell1904}, which was derived under \emph{yielding} constraints and different boundary conditions (our conditions were hinge/hinge and his were hinge/roller).
We also optimized the \emph{substructure} bridge (without deck) to find an optimal complexity $(n,p,q)=(1,1,0)$. This \emph{substructure} bridge has less mass than the \emph{superstructure} bridge except for extremely high complexity ($q>400$). At $q=3000$, the \emph{superstructure} has one fifth the mass of the \emph{substructure} design.
Thirdly, we consider adding a deck to the bridge, since this is the only practical possibility to carry distributed loads. Under \emph{yield} constraints the minimal mass bridge requires infinite complexity $n$ (infinite self-similar iterations of the tensegrity module). The bridge has \emph{superstructure} and \emph{substructure} that are duals of each other. The angle of departure from the boundaries is $35.26 \ deg$ (as opposed to $45 \ deg$ for the no deck mass discussed above). Under buckling constraints the structure $(n,p,q)=(n,1,1)$ has minimal mass at $n=\infty$. The \emph{superstructure} has a departure angle (from the boundary) of approximatively $26.56 \ deg$ as opposed to larger angles for yielding designs and no-deck designs. The \emph{substructure} under buckling constraints has an even more streamlined profile with departure angle approximatively of $5.18 \ deg$. Furthermore the mass of a \emph{substructure} design is much smaller that the mass of a \emph{superstructure} design. 

In all of the design cases studied, we conclude that the infinite complexity \emph{substructure} bridge is the solution which minimizes the sum of deck mass and structural mass. 

Finally, we consider the impact of assigning a mass penalty to the number of required joints. We suppose that the cost per $kg$ of compressive members is $\$_b$, and that the cost per $kg$ of fabricated joints is $\$_j$. The ratio $\Omega=\$_b/\$_j$ is used as a weighting factor to add joint mass to member mass and this sum is minimized. The total minimal mass is always at a finite complexity $n<\infty$ and $p=q=1$. Again, buckling is always the mode of failure in our study, leading to the conclusion that with deck mass and joint mass, this study describes the optimal complexity to obtain a minimal mass bridge, and this bridge is not a continuum (as Michell produced under yield assumptions), but, has finite complexity $n$. The optimal complexity $n$ is given in terms of fabrication costs and material properties.

The tensegrity paradigm used for bridges in this study allows the marriage of composite structures within the design. Our tensegrity approach creates a network of tension and compressive members distributed throughout the system at many different scales (using tensegrity fractals generates many different scales). Furthermore, these tension and compression members can simultaneously serve multiple functions, as load-carrying members of the structure, and as sensing and actuating functions. Moreover, the choice of materials for each member of the network can form a system with special electrical properties, special acoustic properties, special mechanical properties (stiffness, etc). The mathematical tools of this study can be used therefore to design metamaterials and composite materials with unusual and very special properties not available with normal design methods.

\section{Appendix: Numerical results for the bridge constrained with a fixed hinge and rolling hinge ($HR$)}
\label{App_A}

This Appendix reports some numerical results for the cases of \emph{nominal}, \emph{substructure} and \emph{superstructure} bridges illustrated in the study and constrained with a fixed hinge at one end and a rolling hinge at the other end. Both the optimization under yielding and under buckling constraints are illustrated. Table \ref{nr_nom_HR} shows the results obtained for the \emph{nominal} bridges, Table \ref{nr_sub_HR} shows the results obtained for the \emph{substructure} and Table \ref{nr_sup_HR} shows the results obtained for the \emph{superstructure}. For the $HR$ case, the deck elements play an important rule stabilizing the structure. For $HR$ constraints, the so-called \emph{bi-directional} elements must me used since deck elements can be contemporary cables or bars (see \cite{Skelton2013}, \cite{Skelton2014}).
 
\begin{table}[tb]
\begin{center}
\begin{tabular}{cccccccc} \hline \hline
$n$ & $p = q$ & $\alpha^{*}_{Y}$ [deg] & $\beta^{*}_{Y}$ [deg] & $\mu^{*}_{Y}$ & $\alpha^{*}_{B}$ [deg] & $\beta^{*}_{B}$ [deg] & $\mu^{*}_{B}$ \\ \hline \hline
1 & 1 & 35.26 & 35.26 & 0.707 & 16.31 & 29.37 & 592.154 \\ \hline 
1 & 2 & 35.26 & 35.26 & 0.707 & 33.46 & 21.25 & 447.577 \\ \hline
1 & 3 & 41.41 & 41.41 & 0.661 & 43.42 & 15.75 & 380.153 \\ \hline
1 & 4 & 43.23 & 43.23 & 0.651 & 50.35 & 12.14 & 330.506 \\ \hline
1 & 5 & 43.96 & 43.96 & 0.648 & 55.61 & 9.78 & 291.635 \\ \hline
1 & 6 & 44.32 & 44.32 & 0.646 & 59.79 & 8.09 & 260.365 \\ \hline
1 & 7 & 44.52 & 44.52 & 0.645 & 63.18 & 6.86 & 234.711 \\ \hline
1 & 8 & 44.65 & 44.65 & 0.644 & 65.98 & 5.93 & 213.346 \\ \hline
1 & 9 & 44.73 & 44.73 & 0.644 & 68.32 & 5.17 & 195.321 \\ \hline
1 & 10 & 44.79 & 44.78 & 0.644 & 70.29 & 4.57 & 179.943 \\ \hline
1 & 15 & 44.91 & 44.91 & 0.643 & 76.71 & 2.70 & 128.319 \\ \hline
1 & 20 & 44.96 & 44.98 & 0.643 & 80.12 & 1.82 & 99.240 \\ \hline
1 & 25 & 44.96 & 44.98 & 0.643 & 82.18 & 1.37 & 80.763 \\ \hline
1 & 30 & 44.97 & 44.99 & 0.643 & 83.54 & 1.01 & 68.061 \\ \hline
1 & 35 & 44.97 & 44.99 & 0.643 & 84.51 & 0.90 & 58.799 \\ \hline 
1 & 40 & 44.98 & 45.00 & 0.643 & 85.22 & 0.73 & 51.754 \\ \hline 
1 & 45 & 44.98 & 45.00 & 0.643 & 85.78 & 0.61 & 46.219 \\ \hline 
1 & 50 & 44.98 & 45.00 & 0.643 & 86.22 & 0.55 & 41.758 \\ \hline \hline
\end{tabular}
\caption{Numerical results of nominal bridges constrained with a fixed hinge and a rolling hinge ($HR$) and with complexities $n = 1$ and different $p = q$ under yielding ($_Y$) and combined yielding and buckling constraints ($_B$), ($F = 1 \ N$; $L = 1 \ m$; steel bars and steel cables).}
\label{nr_nom_HR}
\end{center}
\end{table}

\begin{table}[tb]
\begin{center}
\begin{tabular}{c||cccc||cccc} \hline \hline
 & $n = 1$ & &  &  & $n = 2$ &  &  & \\ \hline \hline
$p$ & $\beta^{*}_{Y}$ [deg] & $\mu^{*}_{Y}$ & $\beta^{*}_{B}$ [deg] & $\mu^{*}_{B}$ & $\beta^{*}_{Y}$ [deg] &  $\mu^{*}_{Y}$ & $\beta^{*}_{B}$ [deg] & $\mu^{*}_{B}$ \\ \hline \hline
1 & 45.00 & 1.000 & 33.42 & 660.522 & 45.00 & 1.500 & 25.69 & 474.588 \\ \hline
2 & 45.00 & 1.000 & 36.56 & 668.307 & 45.00 & 1.500 & 25.25 & 494.442 \\ \hline
3 & 45.25 & 0.912 & 35.28 & 691.574 & 60.00 & 1.299 & 23.73 & 513.060 \\ \hline
4 & 67.50 & 0.828 & 32.88 & 716.173 & 67.50 & 1.243 & 22.30 & 529.458 \\ \hline
5 & 72.00 & 0.812 & 31.02 & 736.332 & 72.00 & 1.218 & 21.21 & 542.791 \\ \hline
6 & 75.00 & 0.804 & 29.59 & 753.097 & 75.00 & 1.206 & 20.36 & 553.928 \\ \hline
7 & 77.14 & 0.799 & 28.45 & 767.398 & 77.14 & 1.198 & 19.67 & 563.481 \\ \hline
8 & 78.75 & 0.796 & 27.51 & 779.861 & 78.75 & 1.193 & 19.09 & 571.847 \\ \hline
9 & 80.00 & 0.793 & 26.72 & 790.909 & 80.00 & 1.190 & 18.60 & 579.295 \\ \hline
10 & 81.00 & 0.792 & 26.04 & 800.834 & 81.00 & 1.188 & 18.17 & 586.010 \\ \hline
15 & 84.00 & 0.788 & 23.63 & 839.445 & 84.00 & 1.182 & 16.63 & 612.330 \\ \hline
20 & 85.50 & 0.787 & 22.11 & 867.300 & 85.50 & 1.181 & 15.64 & 631.487 \\ \hline
25 & 86.40 & 0.786 & 21.01 & 889.214 & 86.40 & 1.180 & 14.91 & 646.642 \\ \hline
30 & 87.00 & 0.786 & 20.17 & 907.341 & 87.00 & 1.179 & 14.34 & 659.224 \\ \hline
35 & 87.43 & 0.786 & 19.49 & 922.833 & 87.43 & 1.179 & 13.88 & 670.010 \\ \hline
40 & 87.75 & 0.786 & 18.92 & 936.383 & 87.75 & 1.179 & 13.50 & 679.464 \\ \hline
45 & 88.00 & 0.786 & 18.44 & 948.441 & 88.00 & 1.179 & 13.17 & 687.893 \\ \hline
50 & 88.18 & 0.786 & 18.02 & 959.313 & 88.20 & 1.178 & 12.88 & 695.504 \\ \hline \hline
\end{tabular}
\caption{Numerical results of substructures constrained with a fixed hinge and a rolling hinge ($HR$) with different complexities $n$ and $p$ under yielding ($_Y$) and combined yielding and buckling constraints ($_B$), ($F = 1 \ N$; $L = 1 \ m$; steel bars and steel cables).}
\label{nr_sub_HR}
\end{center}
\end{table}

\begin{table}[tb]
\begin{center}
\begin{tabular}{c||cccc||cccc} \hline \hline
 & $n = 1$ & &  &  & $n = 2$ &  &  & \\ \hline \hline
$q$ & $\alpha^{*}_{Y}$ [deg] & $\mu^{*}_{Y}$ & $\alpha^{*}_{B}$ [deg] & $\mu^{*}_{B}$ & $\alpha^{*}_{Y}$ [deg] &  $\mu^{*}_{Y}$ & $\alpha^{*}_{B}$ [deg] & $\mu^{*}_{B}$ \\ \hline \hline
1 & 45.00 & 1.000 & 26.58 & 802.235 & 45.00 & 1.500 & 26.58 & 1085.960 \\ \hline
2 & 45.00 & 1.000 & 36.24 & 515.075 & 45.00	& 1.500 & 36.24 & 697.258 \\ \hline
3 & 60.00 & 0.866 & 45.32 & 410.825 & 60.00	& 1.299 & 45.32 & 556.141 \\ \hline
4 & 67.50 & 0.828 & 51.72 & 347.048 & 67.50	& 1.243 & 51.72 & 469.811 \\ \hline
5 & 72.00 & 0.812 & 56.65 & 301.473 & 72.00	& 1.218 & 56.65 & 408.120 \\ \hline
6 & 75.00 & 0.804 & 60.60 & 266.610 & 75.00	& 1.206 & 60.60 & 360.931 \\ \hline
7 & 77.14 & 0.799 & 63.82 & 238.878 & 77.14	& 1.198 & 63.83 & 323.393 \\ \hline
8 & 78.75 & 0.796 & 66.50 & 216.235 & 78.75	& 1.193 & 66.50 & 292.744 \\ \hline
9 & 80.00 & 0.794 & 68.75 & 197.388 & 80.00	& 1.190 & 68.75 & 267.234 \\ \hline
10 & 81.00 &	0.792 & 70.65 & 181.463 & 81.00 & 1.188 & 70.65 & 245.678 \\ \hline
15 & 84.00 & 0.788 & 76.88 & 128.719 & 84.00 & 1.182 & 76.88 & 174.286 \\ \hline
20 & 85.50 & 0.787 & 80.21 & 99.391 & 85.50	& 1.181 & 80.21 & 134.588 \\ \hline
25 & 86.40 & 0.786 & 82.23 & 80.847 & 86.40	& 1.180 & 82.23 & 109.488 \\ \hline
30 & 87.00 & 0.786 & 83.58 & 68.104 & 87.00	& 1.179 & 83.58 & 92.240 \\ \hline
35 & 87.43 & 0.786 & 84.53 & 58.824 &	 87.43	& 1.179 & 84.53 & 79.679 \\ \hline
40 & 87.75 & 0.786 & 85.24 & 51.770 & 87.75	& 1.179 & 85.24 & 70.131 \\ \hline
45 & 88.00 & 0.786 & 85.79 & 46.223 & 88.00	& 1.179 & 85.79 & 62.632 \\ \hline
50 & 88.18 & 0.786 & 86.23 & 41.765 & 88.20	& 1.178 & 86.23 & 56.588 \\ \hline \hline
\end{tabular}
\caption{Numerical results of superstructures constrained with a fixed hinge and a rolling hinge ($HR$) with different complexities $n$ and $q$ under yielding ($_Y$) and combined yielding and buckling constraints ($_B$), ($F = 1 \ N$; $L = 1 \ m$; steel bars and steel cables).}
\label{nr_sup_HR}
\end{center}
\end{table}

\section{References}

\small{ 
\bibliographystyle{apalike}

\end{document}